\documentclass[11pt,twoside]{article}
\usepackage{fancyhdr}
\pdfoutput=1
\usepackage[colorlinks,citecolor=blue,urlcolor=blue,linkcolor=blue,bookmarks=false]{hyperref}
\usepackage{amsfonts,epsfig,graphicx}
\usepackage{afterpage}
\usepackage{pgfplots}

\usepackage{tikz}
\usetikzlibrary{arrows.meta,calc,positioning}
\usepackage{nicefrac}
\usepackage{comment}
\usepackage{bbm}
\usepackage{amsmath,amssymb,amsthm} 
\usepackage{fullpage}
\usepackage[T1]{fontenc} 
\usepackage{epsf} 
\usepackage{cancel}
\newcommand{\ind}{\perp\!\!\!\!\perp}
\usepackage{graphics} 
\usepackage{amsfonts,amsmath}
\usepackage[round]{natbib} 
\usepackage{psfrag,xspace}
\usepackage{color,etoolbox}
\usepackage{subcaption} 
\usepackage{listings}
\usepgfplotslibrary{groupplots}
\usepackage[noabbrev,capitalize, nameinlink]{cleveref}

\usepgfplotslibrary{fillbetween}

\def\ind{\perp\!\!\!\perp}

\newcommand{\var}{\text{var}}

\newcommand{\Pb}{\mathbb{P}}

\newcommand{\E}{\mathbb{E}}

\usepackage{pgfplots}
\pgfmathdeclarefunction{poiss}{1}{%
  \pgfmathparse{(#1^x)*exp(-#1)/(x!)}%
}

\DeclareSymbolFont{bbold}{U}{bbold}{m}{n}
\DeclareSymbolFontAlphabet{\mathbbold}{bbold}

\newtheorem{theorem}{Theorem}
\newtheorem{lemma}{Lemma}
\newtheorem{corollary}{Corollary}
\newtheorem{proposition}{Proposition}
\newtheorem{algorithm}{Algorithm}
\theoremstyle{definition}

\theoremstyle{remark}

\newtheorem{remark}{Remark}

\let\hat\widehat
\definecolor{dkgreen}{rgb}{0,0.6,0}
\definecolor{gray}{rgb}{0.5,0.5,0.5}
\definecolor{mauve}{rgb}{0.58,0,0.82}
\definecolor{brightblue}{HTML}{00BFC4}

\newcommand\blfootnote[1]{%
  \begingroup
  \renewcommand\thefootnote{}%
  \footnote{#1}%
  \addtocounter{footnote}{-1}%
  \endgroup
}

\pgfplotsset{compat=1.18}
\begin{document}

\def\spacingset#1{\renewcommand{\baselinestretch}%
{#1}\small\normalsize} \spacingset{1}

\raggedbottom
\allowdisplaybreaks[1]


  \title{\vspace*{-.4in} {Incremental effects for continuous exposures}}
   \author{\\ $\text{Kyle Schindl}^{\dagger}$, $\text{Shuying Shen}^{\ddag}$, $\text{Edward H. Kennedy}^{\ddag}$ \\ \\
    $^{\dag}$Department of Statistics \\
    Iowa State University \\
    \texttt{kschindl@iastate.edu} \\ \\ 
    $^\ddag$Department of Statistics \& Data Science \\
    Carnegie Mellon University \\
    \texttt{shuyings@andrew.cmu.edu} \\
    \texttt{edward@stat.cmu.edu}
\date{}
    }
    
  \maketitle
  \thispagestyle{empty}

   \blfootnote{Accompanying \texttt{R} code is available via \href{https://github.com/kyleschindl/incremental-effects-continuous-exposures}{github.com/kyleschindl/incremental-effects-continuous-exposures}}
  
\begin{abstract}
Causal inference problems often involve continuous treatments, such as dose, duration, or frequency. However, identifying and estimating standard dose-response estimands requires that everyone has some chance of receiving any level of the exposure (i.e., positivity). To avoid this assumption, we consider stochastic interventions based on exponentially tilting the treatment distribution by some parameter $\delta$ (an incremental effect); this increases or decreases the likelihood a unit receives a given treatment level. We derive the efficient influence function and semiparametric efficiency bound for these incremental effects under continuous exposures. We then show estimation depends on the size of the tilt, as measured by $\delta$. In particular, we derive new minimax lower bounds illustrating how the best possible root mean squared error scales with an effective sample size of $n / \delta$, instead of $n$. Further, we establish new convergence rates and bounds on the bias of double machine learning-style estimators. Our novel analysis gives a better dependence on $\delta$ compared to standard analyses by using mixed supremum and $L_2$ norms. Finally, we define a ``reflected'' exponential tilt around any interior point and show that taking $\delta \to \infty$ yields a new estimator of the dose-response curve across the treatment support.
\end{abstract}

\noindent%
\bigskip 

\section{Introduction}

Continuous treatments or exposures are common in practice, such as medication dosages, duration of therapy, or frequency of a medical procedure. A standard approach to understanding the effects of these continuous exposures is through dose-response curves, which characterize the mean counterfactual outcomes if all individuals had received a given treatment level \citep{diaz2013targeted, kennedy2017cte, semenova2021debiased}. However, identification of the dose-response curve requires strong positivity assumptions, meaning that all subjects must have a positive treatment density at each possible treatment level. This assumption is often unrealistic in practice. For example, \cite{haneuserotnitzky2013} consider surgical operating time as a continuous treatment. They point out how interventions that simultaneously prolong surgery for some but shorten it for others may be impractical or unrealistic; namely, the counterfactual world where all surgeries last the same amount of time may not provide useful information for practitioners. Consequently, there has been a recent push to develop causal inference methods for evaluating dynamic or stochastic interventions, which can be more realistic and capture more feasible changes to exposure distributions. Recent work in this area includes \cite{munoz2012population, haneuserotnitzky2013, dudik2014, young2014identification,  kennedy2019nonparametric, diaz2020causal}.

Some particularly relevant works are \cite{munoz2012population},  \cite{haneuserotnitzky2013}, and \cite{young2014identification}, who considered dynamic interventions with continuous treatments. However, the methods proposed in these papers still require some nontrivial positivity assumptions. For example, \cite{munoz2012population} consider stochastic shift interventions, where the probability of receiving some treatment level $a$ depending on the covariates, $X$, is shifted by some chosen function $\delta$; i.e., $\pi(a - \delta(X) \mid X)$ where $\pi(\cdot)$ is the probability density function. \cite{haneuserotnitzky2013} consider modified treatment policies where the observed treatment is shifted from $A$ to $A+\delta$. These proposals still lead to positivity violations if the support of $A$ is bounded, or has any holes, where the density is zero. Instead, in this paper we consider a generalization of the incremental effects defined for binary variables in \cite{kennedy2019nonparametric}, where we exponentially tilt the treatment density.

Notably, exponentially tilted intervention distributions have also been considered in the context of mediation analysis by \cite{diaz2020causal}; we make several complementary extensions to their work. First, throughout their analysis they implicitly make the assumption that $\delta$ is finite and bounded above by a constant. In contrast in this paper we allow the size of the exponential tilt $\delta$ to be unbounded, and illustrate how error bounds depend quite severely on $\delta$.  For instance, we show that the variance of the efficient influence function (i.e., the classical nonparametric efficiency bound) is bounded below by $\delta$, illustrating how $\sqrt{n}$ estimation is not feasible for large $\delta$. Indeed, we go on to derive a new minimax lower bound, showing that our incremental effects cannot be estimated at rates faster than $1/\sqrt{n/\delta}$, yielding an effective sample size of $n/\delta$ rather than $n$. Further, we provide a new analysis of the convergence rate and asymptotic distribution of a one-step double machine learning estimator, which explicitly depends on $\delta$. Our analysis opens up the possibility of estimating incremental effects under arbitrary choices of $\delta$, even possibly infinite $\delta$ values. In fact, we show that as $\delta \to \infty$ the incremental effects estimator yields a new estimator of the dose-response curve at the edge of the support, which we study in detail.

The remainder of the paper is organized as follows. In \cref{notation_section} we define all relevant notation and definitions. In \cref{incremental_effects_section} we provide a brief review of incremental effects and establish several useful properties of exponential tilts. In \cref{efficiency_theory_section} we establish the efficient influence function for the incremental effect, and derive the nonparametric efficiency bound, along with lower/upper bounds thereof, explicitly in terms of $\delta$. Finally we derive a new minimax lower bound for estimation of the incremental effect, for arbitrary $\delta$. In \cref{estimation_inference_section} we discuss estimation and inference, and derive new convergence rates that depend on the tilting parameter $\delta$. In \cref{dose_response_section} we consider the setting where $\delta \to \infty$ and show that this yields a novel estimator of the dose-response curve at the edge of the support. Finally, in \cref{experiments_and_simulations} we apply our methodology to estimate the incremental effect of political advertisements on campaign contributions. In \cref{discussion_section} we discuss our results and directions for future work.

\section{Setup \& Notation} \label{notation_section}

We suppose we observe a sample $(Z_1, \ldots, Z_n)$ of independent and identically distributed observations from some distribution $\mathbb{P}$, where $Z_i = (X_i, A_i, Y_i)$ for $X_i \in \mathbb{R}^d$ a vector of covariates, $A_i \in \mathbb{R}$ the treatment of interest, and $Y_i$ the observed outcome. We often use $\mathcal{A}$ for the set of all possible treatment values and $\mathcal{X}$ for all possible covariate values. For simplicity, throughout the paper we assume the support of $A$ is $[0, 1]$, although this could be replaced by any bounded interval, with a simple rescaling. We use $\pi(a \mid x)$ to denote the conditional density of $A$ given $X = x$ and $\mu(x, a) = \mathbb{E}(Y \mid X=x, A=a)$ to denote the regression function. Throughout the paper we use potential outcomes so that $Y^a$ denotes the outcome that would have been observed under $A = a$. We define $||f||^2_2 = \int f(z)^2 d \mathbb{P}(z)$ to be the squared $L_2(\mathbb{P})$ norm, $\mathbb{P}_n(f) = \frac{1}{n} \sum^n_{i=1} f(Z_i)$ to be the empirical measure, and $ ||f||^2_{L^2_x, L^\infty_a} = \int  (\text{sup}_a  \, | f(x, a) | )^2 d \mathbb{P}(x)$ to be a mixed $L_2(\mathbb{P})$-sup norm. Finally, to avoid breaking discussion of results into tedious sub-cases, throughout the paper we assume $\delta > 0$, although all of our results generalize to $\delta < 0$ as well. Often this means we use $\delta$ instead of $|\delta|$, or only consider $\delta \to \infty$ in the statement of results instead of $\delta \to -\infty$ as well. However, in the proofs of each result, when relevant, we include a discussion of how the result holds in the case of $\delta < 0$.

\section{Incremental Effects} \label{incremental_effects_section}

For binary treatments, incremental effects characterize the outcome distribution if every subject's odds of receiving treatment were multiplied by some factor, instead of deterministically setting the treatment level to a specific value. Proposed by \cite{kennedy2019nonparametric} for binary (and possibly time-varying) treatments, the interventional propensity score is
\begin{align*}
     q_\delta(1 \mid x) = \frac{\exp(\delta) \pi(1 \mid x)}{\exp(\delta) \pi(1 \mid x) + 1-\pi(1 \mid x)}
\end{align*}
where $\delta \in (-\infty, \infty)$. Interpretation of this intervention under binary treatments is simple and straightforward, as $\exp(\delta)$ can be written as an odds ratio:
\begin{align*}
    \exp(\delta) = \frac{q_\delta(1 \mid x)/(1 - q_\delta(1 \mid x))}{\pi(1 \mid x)/(1 - \pi(1 \mid x))} = \frac{\text{odds}_q(A = 1 \mid X = x)}{\text{odds}_\pi(A = 1 \mid X = x)}.
\end{align*}
or alternatively as a difference in log-likelihood ratios after taking the natural logarithm. Consequently, the intervention corresponds to multiplying every subject's odds of treatment by $\exp(\delta)$. For example, setting $\exp(\delta) = 3/2$ corresponds to increasing the odds by 50\%. 

Recently, \cite{diaz2020causal} cleverly pointed out that the incremental propensity score for binary treatments is equivalent to an exponential tilt of the observational treatment distribution. This leads to a  natural version of incremental effects for continuous treatments. Namely, for $\delta \in \mathbb{R}$ the exponential tilt of the observational treatment density $\pi(a \mid x)$ is given by
\begin{align} \label{q_definition}
    q_\delta(a \mid x) = \frac{\exp(\delta a) \pi(a \mid x)}{\int \exp(\delta t) \pi(t \mid x) \, dt}. 
\end{align}
 \cite{diaz2020causal} utilize this in the context of mediation analysis, where direct and indirect effects relative to a mediator are of interest. In this paper, we leverage this generalized definition for mean outcomes under exponentially tilted intervention distributions, which we term incremental effects for continuous exposures. Notably, in our analysis \textit{we do not treat $\delta$ as fixed}, thereby differentiating our results from \cite{diaz2020causal}; many fundamental results (including asymptotic normality and rates of convergence) will be impacted when allowing $\delta$ to be unbounded. \cref{Exp_tilted_density} illustrates numerous exponentially tilted versions of a particular observational treatment density (shown in black). Note that negative tilts push toward a point mass at the left end of the support, while positive tilts push toward a point mass at the right end of the support.

\begin{figure}[h]
    \centering
\begin{tikzpicture}
  \begin{groupplot}[
    group style={
      group size=2 by 1,
      horizontal sep=0.75cm,
    },
    width=0.53\textwidth,
    height=0.465\textwidth,
    grid=both,
    grid style={line width=0.1pt, draw=gray!20, opacity=0.5},
    ymax=10,
    xlabel={$a$},
    ylabel={$q_\delta(a \mid x)$},
        ylabel style={yshift=-0.3cm},
    legend cell align=left,
        xtick pos=bottom,
    ytick pos=left,
    legend image post style={sharp plot,-, very thick},
  ]
   \definecolor{ggred}{HTML}{fb2424}
  \definecolor{brightblue}{HTML}{00BFC4}

  \nextgroupplot[title={$\delta < 0$},   legend style={
  at={(0.975,0.975)},
        anchor=north east
    }]
  \pgfplotstableread[col sep=comma]{files/delta_neg.csv}\datatableneg

  \addplot[color=lightgray, opacity=0.75,forget plot] table[x="a", y="delta_3", col sep=comma] {\datatableneg};
  \addplot[color=lightgray, opacity=0.75,forget plot] table[x="a", y="delta_4", col sep=comma] {\datatableneg};
  \addplot[color=lightgray, opacity=0.75,forget plot] table[x="a", y="delta_5", col sep=comma] {\datatableneg};
  \addplot[color=lightgray, opacity=0.75,forget plot] table[x="a", y="delta_8", col sep=comma] {\datatableneg};
  \addplot[color=lightgray, opacity=0.75,forget plot] table[x="a", y="delta_10", col sep=comma] {\datatableneg};
  \addplot[color=lightgray, opacity=0.75,forget plot] table[x="a", y="delta_13", col sep=comma] {\datatableneg};
  \addplot[color=lightgray, opacity=0.75,forget plot] table[x="a", y="delta_16", col sep=comma] {\datatableneg};
  \addplot[color=lightgray, opacity=0.75,forget plot] table[x="a", y="delta_20", col sep=comma] {\datatableneg};
  \addplot[color=lightgray, opacity=0.75,forget plot] table[x="a", y="delta_30", col sep=comma] {\datatableneg};

    \addplot[color=black, thick] table[x="a", y="density", col sep=comma] {\datatableneg};
  \addlegendentry{$\delta = 0$}
    \addplot[color=lightgray, opacity=0.75] table[x="a", y="delta_2", col sep=comma] {\datatableneg};
  \addlegendentry{$\delta \in (0, -40)$}
  \addplot[color=brightblue, ultra thick] table[x="a", y="delta_40", col sep=comma] {\datatableneg};
  \addlegendentry{$\delta = -40$}

  \nextgroupplot[title={$\delta > 0$}, ylabel=\empty, 
  legend style={at={(0.025,0.975)},
        anchor=north west
    }]
  \pgfplotstableread[col sep=comma]{files/delta_pos.csv}\datatablepos

  \addplot[color=lightgray, opacity=0.75,forget plot] table[x="a", y="delta_3", col sep=comma] {\datatablepos};
  \addplot[color=lightgray, opacity=0.75,forget plot] table[x="a", y="delta_4", col sep=comma] {\datatablepos};
  \addplot[color=lightgray, opacity=0.75,forget plot] table[x="a", y="delta_5", col sep=comma] {\datatablepos};
  \addplot[color=lightgray, opacity=0.75,forget plot] table[x="a", y="delta_8", col sep=comma] {\datatablepos};
  \addplot[color=lightgray, opacity=0.75,forget plot] table[x="a", y="delta_10", col sep=comma] {\datatablepos};
  \addplot[color=lightgray, opacity=0.75,forget plot] table[x="a", y="delta_13", col sep=comma] {\datatablepos};
  \addplot[color=lightgray, opacity=0.75,forget plot] table[x="a", y="delta_16", col sep=comma] {\datatablepos};
  \addplot[color=lightgray, opacity=0.75,forget plot] table[x="a", y="delta_20", col sep=comma] {\datatablepos};
  \addplot[color=lightgray, opacity=0.75,forget plot] table[x="a", y="delta_30", col sep=comma] {\datatablepos};
\addplot[color=black, thick] table[x="a", y="density", col sep=comma] {\datatablepos};
    \addlegendentry{$\delta = 0$}
  \addplot[color=lightgray, opacity=0.75] table[x="a", y="delta_2", col sep=comma] {\datatablepos};
     \addlegendentry{$\delta \in (0, 40)$}
  \addplot[color=ggred, ultra thick] table[x="a", y="delta_40", col sep=comma] {\datatablepos};
   \addlegendentry{$\delta = 40$}

    \end{groupplot}
\end{tikzpicture}
    \caption{Exponentially tilted densities, where $\delta = 0$ represents no tilt.}
    \label{Exp_tilted_density}
\end{figure}

\subsection{Some History}

Exponential tilts have a long history. They appear to have originated in financial mathematics and the actuarial sciences \citep{escher1932probability}, where they are referred to as Esscher transforms (for examples see \citet{gerber1993option, elliott2005option} and references therein). They are also central in importance sampling \citep{siegmund1976importance} and large deviations theory \citep{ellis2007entropy}, and are employed in many varied contexts in statistics, from the bootstrap to asymptotic approximation to density estimation, for example \citep{efron1981nonparametric, field1982small, vernon_johns_1988, efron1996using, bayes_emp_2005, qin1999empirical, rare_event}. Also see \citet{maity2022understanding} for connections in machine learning, e.g., distribution shift. In causal inference, exponential tilts have been used in sensitivity analyses by \citet{robins2000sensitivity, franks2020flexible, scharfstein2021semiparametric}, but in a very different context from what we consider here (there they parameterize the extent of unmeasured confounding, rather than an intervention distribution). Notably, \cite{rakshit2024localeffectscontinuousinstruments}  recently proposed using exponential tilts to estimate local average treatment effects with continuous instruments while avoiding the positivity assumption, and \cite{schindl2025causalgeodesycounterfactualestimation} consider exponential tilts as a path of stochastic interventions between distributions on a statistical manifold.

\begin{remark}
As far as we know, \cite{diaz2020causal} were the first to propose causal effects based on exponential tilts of a generic (possibly continuous) observational treatment distribution, but our work complements and adds to theirs in substantial ways. Crucially, we characterize errors as a function of the shift parameter $\delta$, rather than treating $\delta$ as a hidden constant, which substantially changes the problem, e.g., allowing dose-response estimation and showing how convergence rates are altered quite severely. A detailed comparison of our work and theirs can be found in \cref{estimation_inference_section}.
\end{remark}

\subsection{Interpretation}

For binary treatments, the tilt parameter $\delta$ is simply a log odds ratio or, equivalently, a difference in log-likelihood ratios. When $A$ is continuous, there is an analogous albeit less simple interpretation, generalizing the binary case. Namely, here $\delta$ is a rate of change in log-likelihood ratios, rather than a difference, i.e.,
\begin{align*}
    \delta = \frac{\partial}{\partial a} \left\{ \text{log}\left(\frac{q_\delta(a \mid x)}{\pi(a \mid x)} \right) \right\}.
\end{align*}
Thus, when moving from binary to continuous treatments, the interpretation of $\delta$ coincides with how derivatives are continuous analogs of discrete finite differences. 
Here $\delta$ measures the local sensitivity to the log-likelihood ratio for a given value of $a$; for an infinitesimal increase in $a$, the log-likelihood ratio increases by $\delta \, da$. Consequently, we can see that setting $\delta > 0$ corresponds to an increase in the likelihood of receiving a higher treatment level, and $\delta < 0$ a corresponding decrease. Since interpretation of $\delta$ is more complex with continuous treatments, an important area for future work is to look at functionals of the intervention distribution $q_\delta(a \mid x)$ (e.g., means, quantiles). In the next section, we explore some of the key properties of exponentially tilted intervention distributions. 

\subsection{Properties}

Incremental effects for binary treatments have at least three important properties. First, they do not require positivity for identification. Second, the intervention distribution is a smooth function of the observational distribution; this naturally bypasses the positivity assumption without the challenging statistical task of identifying subjects for whom positivity is violated. Third, when positivity does hold, incremental effects interpolate between standard static interventions where either everyone versus no one is treated. As described below, it turns out that these properties also hold true for continuous incremental effects under the exponentially tilted intervention distribution defined in \cref{q_definition}.

\paragraph{Positivity:} Identifying effects under the tilted distribution does not require any positivity assumption, in stark contrast to more typical dose-response-style effects which require $\pi(a \mid x)$ to be bounded away from zero, for all $a$ of interest and $x$ in the support. Positivity is not needed because no tilting occurs in zero-density portions of $\pi(a \mid x)$, i.e., when $\pi(a \mid x) = 0$ then $q_\delta(a \mid x) = 0$, by definition. Note this also differs substantially from other stochastic intervention effects, such as the shift interventions of  \cite{munoz2012population}, which characterize what would happen if observed treatments $A$ were shifted to $A + \delta$ (i.e., the treatment density is set to $\pi(a - \delta \mid x)$); these interventions are allowed to place mass on portions of the treatment density outside the support. 

\paragraph{Smoothness:} Another advantage of exponential tilts is that the intervention distribution is a smooth function of the treatment density. One could consider alternative intervention distributions that avoid  positivity assumptions by essentially not intervening on units with $\pi(a \mid x) < \varepsilon$ for some $\varepsilon > 0$ \citep{van2007causal, branson2023causal}; however, this requires estimating which subjects have $\pi(a \mid x)<\epsilon$, a nonsmooth estimation problem that can yield slower convergence rates.

\paragraph{Static Interventions:} A third useful property of exponential tilts is that they interpolate between standard static interventions, when positivity does hold. In binary settings, as $\delta \to \infty$ then $q_{\delta}(a \mid x) \to 1$, corresponding to the intervention in which all units deterministically receive treatment level $A = 1$. Similarly, taking $\delta \to - \infty$ is equivalent to setting $A = 0$. This property of binary incremental effects is convenient as this recovers and generalizes the classical definition of the average treatment effect of giving all or none treatment. When $A$ is continuous, we again have an analogous interpretation. Here it follows that taking $\delta \to \infty$ pushes the mass of the distribution to the edge of its support. In \cref{Exp_tilted_density} we can see that for $A \in [0, 1]$, taking $\delta \to \infty$ creates a point mass at $A = 1$. Similarly, taking $\delta \to - \infty$ creates a point mass at $A = 0$. Indeed, in the setting of continuous exposures, the extreme values of $\delta$ represent the dose-response curve evaluated at the edges of the support of the treatment distribution. We explore this setting further in \cref{dose_response_section}.

\subsection{Identification} \label{identification_section}

In the previous section, we discussed exponentially tilted intervention distributions, which shift the log-likelihood ratio of receiving some treatment level $a$. In this section, we establish the assumptions required for identifying the mean potential outcomes $\mathbb{E}[Y^{Q(\delta)}]$ where $Q(\delta)$ is the post-intervention joint distribution of $(X, A)$ obtained by sampling $X \sim \mathbb{P}_X$ and then $A \mid X \sim q_\delta( \cdot \mid X)$. As previously discussed, we do not require the positivity assumption for identification -- only consistency and exchangeability. Specifically we assume:
\begin{enumerate}
    \item[(i)] \textit{Consistency:} $Y = Y^a$ if $A = a$.
    \item[(ii)] \textit{Exchangeability:} $A \ind Y^a \mid X$ for $a \in \mathbb{R}$
\end{enumerate}
Under assumptions $(i)$ and $(ii)$, the incremental effect $\psi(\delta) = \mathbb{E}[Y^{Q(\delta)}]$ is equal to
\begin{align} \label{identified_eq}
    \psi(\delta) =  \int_\mathcal{X} \int_\mathcal{A} \mu(x, a) q_\delta(a \mid x) \, da \, d\mathbb{P}(x).
\end{align}

The consistency assumption requires that the observed outcome for a given unit corresponds to the potential outcome under the observed treatment; i.e., there is no interference or spillover effects between one unit and the outcomes of another. This assumption is often violated under network structures; for example, if receiving treatment that lowers the risk of a viral infection will also reduce the likelihood that neighbors become infected. The exchangeability assumption requires that, conditional on the collected covariates, the treatment assignment is as good as random. This assumption is satisfied in randomized experiments, but can be violated in observational settings; there, the assumption is untestable, and is justified based on subject matter expertise.

Removing the positivity assumption is of great practical importance in the estimation of causal effects. Positivity is often unrealistic, particularly when treatment is continuous; in many settings it is just not possible for every unit to have some non-zero chance of receiving each treatment level. For example, in studying the effects of political advertisements on campaign contributions \citep{urban2014dollars, fong2018covariate}, the number of political advertisements in a given ZIP code is essentially a continuous random variable. In this setting,  many ZIP codes fail to satisfy the positivity assumption, since politicians are unlikely to ignore politically important ZIP codes, and are unlikely to heavily advertise in unimportant ones. Incremental effects via exponential tilts of the density of political advertisements allow for estimation of the effect of increasing the likelihood that a given ZIP code receives a certain number of advertisements, without requiring positivity. We explore this application further in \cref{experiments_and_simulations}.

\section{Fundamental Limits} \label{efficiency_theory_section}

In this section we explore efficiency theory and fundamental limits of estimation of incremental effects under exponential tilts. For more details on nonparametric efficiency theory, we refer to \cite{bickel1993efficient, laan2003unified,  van2002semiparametric, tsiatis2006semiparametric, kosorok2008introduction, kennedy2023semiparametric}.  In particular, we show how error rates depend on not just sample size $n$, but also the increment parameter $\delta$. To do so, we first derive the efficient influence function and nonparametric efficiency bounds for $\psi(\delta)$ when $\delta$ is finite. This finite $\delta$ condition is necessary because for unbounded $\delta$ the variance of the efficient influence function is no longer finite, leading to non-pathwise differentiability and the non-existence of usual influence functions. However, later, we go on to derive a minimax lower bound on the estimation error that allows for $\delta$ to be unbounded. Further, in later sections we show that estimation and inference are still possible for unbounded $\delta$, albeit at slower convergence rates.

\begin{remark}
In what follows we define $\frac{q_\delta(a \mid x)}{\pi(a \mid x)} = \frac{q_\delta(a \mid x)}{\pi(a \mid x)}$ if $\pi(a \mid x) > 0$ and $0$ otherwise to resolve ambiguities that arise in zero-density sections of $\pi(a \mid x)$ (i.e., define $0/0=0$). 
\end{remark}

\subsection{Efficiency Bound}

The following proposition establishes the efficient influence function for $\psi(\delta) = \mathbb{E}[Y^{Q(\delta)}]$. Although this result is also given by \cite{diaz2020causal}, our statement and proofs are different since their analysis includes mediators, which must be removed in order to draw an exact connection between the two settings.

\begin{proposition} \label{eif_theorem}
Suppose $\delta \in [-M,M]$ for some $M<\infty$. Then, the efficient influence function of $\psi(\delta)$ under a nonparametric model is given by $\varphi(Z;\delta) = D_Y + D_{q, \mu} + D_\psi$ for
\begin{align*}
    D_Y &= \frac{q_\delta(A \mid X)}{\pi(A \mid X)} \Big(Y - \mu(X, A) \Big) \\
    D_{q, \mu} &= \frac{q_\delta(A \mid X)}{\pi(A \mid X)}\Big( \mu(X, A) - \mathbb{E}_Q(\mu(X, A) \mid X) \Big) \\
    D_\psi &= \mathbb{E}_Q(\mu(X, A) \mid X)  - \psi(\delta)
\end{align*}
where $\mathbb{E}_Q(\mu(X, A) \mid X) = \int_a \mu(X, a) q_\delta(a \mid X) \, da$ is the conditional mean of $\mu(X, A)$ under the exponentially tilted distribution.
\end{proposition}

 \cref{eif_theorem} shows how the influence function for the incremental effect parameter is split into three parts. The first part, $D_Y$ represents the difference between the outcomes and the regression function, adjusted for the tilted likelihood of receiving treatment. The second term $D_{q, \mu}$ represents the difference between the regression function and its conditional mean under the tilted density $q_\delta(a \mid x)$, again adjusted for the tilted likelihood. Finally, the third term $D_\psi$ represents a total difference between the regression function averaged over $q_\delta(a \mid x)$ and the stochastic intervention effect, $\psi$.

Next, we further develop the efficiency theory for exponentially tilted stochastic interventions by deriving an explicit expression for the nonparametric efficiency bound. The nonparametric efficiency bound, i.e., the variance of the efficient influence function, acts as a nonparametric analogue of the Cramer-Rao lower bound, since no other estimator can have a smaller mean squared error, in a local asymptotic minimax sense. More formally, if a functional $\psi : \mathbb{P} \to \mathbb{R}$ is pathwise differentiable with efficient influence function $\varphi$, then
\begin{align*}
    \underset{\varepsilon > 0}{\text{inf}} \: \underset{n \to \infty}{\text{lim}} \: \underset{TV(\mathbb{P}, \overline{\mathbb{P}}) < \varepsilon}{\text{sup}} \left( n \mathbb{E}_{\overline{\mathbb{P}}} \left[ \left\{ \widehat{\psi} - \psi(\overline{\mathbb{P}}) \right\}^2\right] \right) \geq \var\{\varphi(Z; \mathbb{P}) \}
\end{align*}
for any estimator sequence $\widehat{\psi} = \widehat{\psi}_n$ \citep{van2000asymptotic}. Deriving a closed form expression for the nonparametric efficiency bound can be useful for characterizing precisely when estimation is more or less difficult, in a nonparametric model. The following theorem establishes the nonparametric efficiency bound in this problem.

\begin{theorem} \label{nonpar_efficiency_bound}
The variance of the influence function $\varphi(Z;\delta)$, which is the nonparametric efficiency bound when $\delta \in [-M,M]$ for some $M<\infty$, is given by
 \begin{align*}
    \sigma^2_\delta = \mathbb{E}\left[\left(\frac{q_\delta(A \mid X)}{\pi(A \mid X)} \right)^2 (Y - \mathbb{E}_Q\left[\mu(X, A) \mid X \right])^2 \right] + \var\Big( \mathbb{E}_Q\left[\mu(X, A) \mid X \right] \Big)
\end{align*}
where the subscript $Q$ denotes expectations with respect to $q_\delta(A \mid X)$. 
\end{theorem}

An important consequence of \cref{nonpar_efficiency_bound} is that the variance of the efficient influence function is largely dependent on the likelihood ratio ${q_\delta(a \mid x)}/{\pi(a \mid x)}$, which is a function of the parameter $\delta$. One might hope that modest choices of $\delta$ would not materially impact this likelihood ratio, but in fact it 
can scale with $\delta$. For example, even if $\pi(a \mid x)=1$ is uniform, then 
\begin{align*}
    \frac{q_\delta(a \mid x)}{\pi(a \mid x)} = \frac{\delta \exp(\delta a)}{\exp(\delta) - 1}
\end{align*}
is no smaller than $\delta$ for $a=1$. Contrast this with standard causal inference problems with binary treatments, where one often assumes the analogous inverse propensity scores is bounded above by some constant. In fact, in the next section we show that the nonparametric efficiency bound is not just upper bounded by $\delta$, but lower bounded as well. Intuitively,  the larger the $\delta$, the farther the intervention distribution is from the original treatment density. Thus, the best possible variance should be increasing with the size of the tilt, and with the farther we deviate from the observed density. We formalize this point by considering the Kullback-Leibler divergence between $q_\delta(a \mid X)$ and $\pi(a \mid X)$. Suppose $A \sim q_\delta( \cdot \mid X)$ is some random variable drawn from the intervention distribution. Then, it follows that $ D_{\text{KL}}(q_\delta(a \mid X) \: || \: \pi(a \mid X)) = \delta \mathbb{E}_Q[A \mid X] - \kappa(\delta)$ where $\kappa(\delta) = \text{log}\left(\int_a \exp(\delta a) \pi(a \mid x) \, da \right)$ is the cumulant generating function. The following result establishes that the divergence between $q_\delta(a \mid X)$ and $\pi(a \mid X)$ is increasing with $\delta$, and the rate at which the divergence increases depends on the skewness of $q_\delta(a \mid X)$.

\begin{proposition}
     \label{kl_divergence_lemma}
    Suppose $\delta > 0$ and that $\pi(a \mid X)$ is continuous. Then
    \begin{align*}
        \frac{\partial}{\partial \delta} \left\{ D_{\text{KL}}\Big(q_\delta(a \mid X) \: || \: \pi(a \mid X) \Big) \right\} &= \delta \var_Q(A \mid X) \\
        \frac{\partial^2}{\partial \delta^2}  \left\{ D_{\text{KL}}\Big(q_\delta(a \mid X) \: || \: \pi(a \mid X) \Big) \right\} 
        &= \var_Q(A \mid X) + \delta \mathbb{E}_Q\left[(A - \mathbb{E}_Q[A \mid X])^3 \mid X \right]
    \end{align*} 
\end{proposition}

One takeaway from \cref{kl_divergence_lemma} is that the rate at which $q_\delta(a \mid X)$ grows apart from $\pi(a \mid X)$ depends on the variance of the treatment exposure under the intervention distribution. This is intuitive, as tilting a high-variance distribution is only going to spread mass farther away from its original position. Importantly, this discussion of the nonparametric efficiency bound and the KL-divergence exemplifies a theme we  find across our results, that estimation is more complicated as $\delta$ grows larger. In the next section we explore this further, via minimax lower bounds. 

\subsection{Minimax Lower Bound} \label{minimax_section}

In this section, we establish a minimax lower bound across a flexible nonparametric model with bounded treatment density. Minimax lower bounds are an important way to benchmark estimation error across a statistical model $\mathcal{P}$; they demonstrate that no estimator can achieve lower estimation error without adding more assumptions. In order to prove our result, we first prove a lemma showing that the nonparametric efficiency bound is both lower and upper bounded by $\delta$, which is interesting in its own right. In this result, and several subsequent ones, in order to avoid a superlinear (or even exponential) penalty in $\delta$ we make use of what we refer to as a ``weak positivity" assumption. That is, we assume there exists some $\eta \in [0, 1)$ such that $\pi(a \mid x) \geq \pi_{\min} > 0$ for all $a \in [\eta, 1]$. Although positivity is not required for identification, it turns out that weak positivity is useful for non-trivial estimation rates. One intuitive way to understand this assumption is to consider the limiting behavior of $q_\delta(a \mid x)$ as $\delta \to \infty$. As $q_\delta(a \mid x)$ begins to concentrate near the edge of the support, if there is not sufficient concentration of $\pi(a \mid x)$ in this neighborhood then $\mathbb{E}\left[q_\delta(A \mid X)^2 / \pi(A \mid X)^2\right]$ can scale arbitrarily with $\delta$. Note that this distinction is made for a \textit{positive} tilt; conversely, if we were to take $\delta \to -\infty$ then we would need to assume positivity holds in a neighborhood of $[0, \eta]$ for $\eta \in (0, 1]$. To avoid cumbersome notation, in all subsequent results we assume $\delta > 0$, although all results hold for $\delta < 0$ as well. Furthermore, in \cref{weak_positivity_section}, we provide an extended discussion of weak positivity, and how removing a uniform lower bound at the edge of the treatment support can lead to arbitrarily bad penalties in $\delta$. Fortunately, weak positivity is quite weak overall, and allows for arbitrary holes in the treatment density outside of the edges of support. The following theorem establishes bounds on the variance of the efficient influence function.

\begin{lemma} \label{sigma_bounds}
Assume:
\begin{enumerate}
\item[(i)] There exists some $\eta \in [0, 1)$ such that $\pi(a \mid x) \geq \pi_{\min} > 0$ for all $a \in [\eta, 1]$.
\item[(ii)] $\pi(a \mid x) \leq \pi_{\max}$ for all $a$, $x$.
\item[(iii)] $\sigma^2_{\min} \leq \var(Y \mid X=x,A=a)$ for all $a,x$.
\item[(iv)] $|Y|\leq B$ with probability one.
\end{enumerate}
Then the variance of the influence function $\varphi(Z; \delta)$ from \cref{eif_theorem} satisfies
\begin{align*}
     \delta \left( \frac{\pi_{\min} \sigma^2_{\min}}{\pi^2_{\max}} \cdot \frac{2}{5}(1 - \eta) \right) \leq \var \left\{ \varphi(Z;\delta) \right\}  \leq B^2 \left[1 + \frac{5}{2} \cdot \frac{\pi_{\max}}{\pi^2_{\min}} \left(\delta + \frac{2}{1 - \eta} \right) \right].
\end{align*}
\end{lemma}

Note that in the case of a negative tilt, the bounds would be in terms of $|\delta|$. \cref{sigma_bounds} highlights the importance of explicitly studying the dependence on $\delta$ in analyzing incremental effects. Since $\delta$ lower bounds the variance of the efficient influence function, as $\delta$ increases, the variance can become unbounded, increasing to infinity. Ignoring dependencies on $\delta$ misses this complication, and can blind practitioners as to how larger choices of $\delta$ lead to substantially larger errors; consequently, in what follows we focus on deriving results that consider both the sample size and tilt size $\delta$ in tandem. We show that the best possible root mean squared error is lower bounded by $1/\sqrt{n/\delta}$ in a nonparametric model with bounded treatment density.

\begin{theorem} \label{minimax_lb}
Let $\mathcal{P}$ denote a model where assumptions $(ii)$, $(iii)$, and $(iv)$ of \cref{sigma_bounds} hold. Then for any $\delta \geq 2$ and universal constant $C$, the minimax rate is lower bounded as
\begin{align*}
    \underset{\widehat{\psi}}{\text{inf}} \, \underset{P \in \mathcal{P}}{\text{sup}} \, \mathbb{E}_P  \left|\widehat{\psi} - \psi_P(\delta) \right|  \geq \sqrt{\frac{C}{n/\delta}}.
\end{align*}
\end{theorem}

We note that the assumption $\delta \geq 2$ is not necessary; it only ensures that the constant $C$ does not depend on $\delta$, even in a trivial way.  The proof of \cref{minimax_lb} relies on an application of Le Cam's two-point method \citep{tsbakov2009}. Specifically, we rely on a reduction to testing between two distributions: a null density and fluctuated alternative, which are different enough to exhibit some separation in $\psi(\delta)$, but similar enough that they cannot be reliably distinguished from each other. If we define these two densities to be
\begin{align*}
    p_0(z) &= p(y \mid x, a) p(a \mid x) p(x) \\
    p_1(z) &= \Big[p(y \mid x, a)(1 + \varepsilon \phi_y(z; p)) \Big]p(a \mid x) p(x)
\end{align*}
for some $\varepsilon$, where $ \phi_y(z; p) = \frac{q_\delta(a \mid x)}{p(a \mid x)}(y - \mu(x, a))$. Then it can be shown that the functional separation evaluates to
\begin{align*}
    |\psi(p_0) - \psi(p_1)|  = \varepsilon \left| \int  \phi_y(z; p)^2 p(z) dz \right|.
\end{align*}
Importantly, this functional separation recovers a key piece of the nonparametric efficiency bound, given by the term containing the squared likelihood ratio and conditional variance of the outcomes. This allows us to directly leverage \cref{sigma_bounds} to show that the functional separation scales with $\varepsilon \delta$, thereby establishing the minimax lower bound, after noting that the distance between the two distributions remains bounded.

Thus, \cref{minimax_lb} shows that the best possible root mean squared error increases with the size of the exponential tilt, and in particular that the effective sample size in these problems is $n/\delta$ instead of $n$. This highlights how estimation error must necessarily diminish with larger tilts $\delta$. These findings reinforce the importance of explicitly considering the dependence on $\delta$ in the estimation of incremental effects; treating $\delta$ as a fixed value misses this, and suggests that convergence rates are faster than they actually are. In the next section, we explore how $\delta$ impacts the bias and convergence rate during estimation. 

\section{Estimation \& Inference} \label{estimation_inference_section}

In this section we study an estimator for the incremental effect $\psi(\delta)$ and establish its asymptotic distribution and convergence rate. We take special care to make explicit the dependence of error bounds on $\delta$ throughout, and do not constrain $\delta$ to be bounded, even allowing it to grow with sample size $n$, or be infinite. Importantly, this also allows practitioners to make more informed choices about the range of $\delta$ to consider in analyses. 

\subsection{Proposed Estimator} \label{proposed_estimator_section}

Here, we propose using the ``one-step'' double machine learning estimator to estimate the incremental effect $\psi(\delta)$. The one-step estimator is obtained by subtracting the plug-in estimator's estimated first-order bias, i.e. $-\int \varphi(z; \widehat{\mathbb{P}}) d\mathbb{P}(z)$. Since $\varphi(\cdot) $ is exactly the efficient influence function we derived in \cref{eif_theorem}, it follows that the bias-corrected estimator $\psi(\widehat{\mathbb{P}}) + \mathbb{P}_n\{ \varphi(Z; \widehat{\mathbb{P}}) \}$ is given by
\begin{align*}
    \widehat{\psi}(\delta) = \mathbb{P}_n \left[\frac{\widehat{q}_\delta(A \mid X)}{\widehat{\pi}(A \mid X)} \left(Y - \int_a \widehat{\mu}(X, a) \widehat{q}_\delta(a \mid X) da \right) + \int_a \widehat{\mu}(X, a) \widehat{q}_\delta(a \mid X) da  \right] , 
\end{align*}
which also coincides with that of \cite{diaz2020causal} when all mediators are removed. 
In practice, estimation of $\widehat{\psi}(\delta)$ is relatively straightforward. The regression function $\mu$ can be flexibly estimated using standard nonparametric regression methods. If we assume that $\pi(a \mid x)$ is $\alpha$-smooth, then estimation of $\pi$ can similarly be reduced to a regression problem. Indeed,  it can be shown that for a standard kernel $K(\cdot)$, bandwidth $h$, and constant $C$,
\begin{align*}
    \left| \mathbb{E}\left[\frac{1}{h} K \left( \frac{A - a}{h} \right) \big| \ X = x\right] - \pi(a \mid x) \right| \leq h^\alpha \left(\frac{C}{\lfloor \alpha \rfloor!  } \right) \int |u|^\alpha |K(u)| du, 
\end{align*}
so the conditional expectation is within $h^\alpha$ of the true conditional density. Alternatively, one could use a semiparametric approach for estimating $\pi(a \mid x)$ by assuming a model $A = \lambda(X) + \gamma(X) \varepsilon$, where $\varepsilon$ follows some unspecified density that has zero mean and unit variance given the covariates, $\lambda(x) = \mathbb{E}[A \mid X = x]$, and $\gamma(x) = \var(A \mid X = x)$ \citep{hansen2004nonparametric}; for an example of how this method can be implemented in practice, see \cite{kennedy2017cte}. Regardless of which method is used to estimate $\pi$, one can then  numerically integrate the integrands defined in $\widehat{\psi}(\delta)$. When estimating  nuisance functions, we recommend sample-splitting as in \cite{robins2008higher, zheng2010asymptotic, Chernozhukov2018} for finite-sample robustness and to remove the need for Donsker-type restrictions on $(\widehat{\mu}, \widehat{\pi})$. We propose the following algorithm for calculating $\widehat{\psi}(\delta)$.

\begin{algorithm} \label{cross_validated_algo}
    Split the data into $K$ folds such that fold $k \in \{1, \ldots, K\}$ has $n_k$ observations. For each $k \in K$:
    \begin{enumerate}
        \item[(i)] Estimate $\widehat{\mu}_{-k}(X, A)$ on the observations not contained in fold $k$.
        \item[(ii)] Generate $D$ design points drawn from the support of $a$. For each $a_1, \ldots, a_D$, estimate $\widehat{\pi}_{-k}(a_d \mid X_i)$ either via kernel transformed outcomes 
        \begin{align*}
           \widehat{\pi}_{-k}(a_d \mid X_i) = \widehat{\mathbb{E}}\left[\frac{1}{h} K \left( \frac{A - a_d}{h} \right) \big| \ X = X_i\right]
        \end{align*}
        or some alternate method of estimation.
        \item[(iii)] Numerically integrate the estimates of $\exp(\delta a) \pi(a \mid x)$ and $\mu(x, a) q_\delta(a \mid x)$ by averaging over the design points to obtain
        \begin{align*}
            \widehat{\nu}_{-k}(X_i) &= \frac{1}{D} \sum^D_{d=1} \exp(\delta a_d) \widehat{\pi}_{-k}(a_d \mid X_i) \quad \text{and} \\
            \widehat{\xi}_{-k}(X_i) &=  \frac{1}{D} \sum^D_{d=1} \widehat{\mu}(X_i, a_d) \left(\frac{\exp(\delta a_d) \widehat{\pi}_{-k}(a_d \mid X_i)}{\widehat{\nu}_{-k}(X_i)}\right).
        \end{align*}
        Use $\widehat{\nu}_{-k}(X_i)$ to obtain an estimate of the likelihood ratio,
        \begin{align*}
            \frac{\widehat{q}_{\delta, - k}(a \mid X_i)}{\widehat{\pi}_{-k}(a \mid X_i) } = \frac{\exp(\delta a)}{\widehat{\nu}_{-k}(X_i)}.
        \end{align*}
        \item[(iv)] Calculate the estimate for $\psi(\delta)$ in fold $k$ by taking the average
        \begin{align*}
            \widehat{\psi}_k(\delta) = \frac{1}{n_k} \sum_{j \in k} \left\{ \frac{\widehat{q}_{\delta, - k}(a_j \mid X_j)}{\widehat{\pi}_{-k}(a_j \mid X_j) }\left(Y_j - \widehat{\xi}_{-k}(X_j) \right) + \widehat{\xi}_{-k}(X_j) \right\}
        \end{align*}
    \end{enumerate}
    Then, to obtain the final estimate for $\psi(\delta)$ average across the $K$ folds, $\widehat{\psi}(\delta) = \frac{1}{K} \sum^K_{k=1} \widehat{\psi}_k(\delta)$.
\end{algorithm}

\begin{remark}
    Alternatively, to avoid directly estimating the conditional density,  note that each nuisance function can be framed as a regression problem. Specifically
\begin{align*}
    \frac{q_\delta(a \mid X)}{\pi(a \mid X)} = \frac{\exp(\delta a)}{\int_t \exp(\delta t) \pi(t \mid X) dt} = \frac{\exp(\delta a)}{\mathbb{E}\left[\exp(\delta A) \mid X \right]} =: \frac{\exp(\delta a)}{\nu(X)}
\end{align*}
where $\nu(X)$ is estimated by regressing $\exp(\delta A)$ on $X$ and
\begin{align*}
    \int_a \mu(X, a) q_\delta(a \mid X) da = \frac{\int_a \exp(\delta a) \mu(X, a) \pi(a \mid X) da}{\int_a \exp(\delta a) \pi(a \mid X) da} = \frac{\mathbb{E}\left[\exp(\delta a) \mu(X, A) \mid X \right]}{\mathbb{E}\left[ \exp(\delta a) \mid X\right]} =: \frac{\eta(X)}{\nu(X)}
\end{align*}
where $\eta(X)$ is estimated by regressing $\exp(\delta A) Y$ on $X$ or by first estimating $\mu(X, A)$ by regressing $Y$ on $(X, A)$. Thus, estimation of $\widehat{\psi}(\delta)$ can be reduced to two (or three) regressions, all of which can be flexibly estimated using nonparametric methods: $\widehat{\mu}, \widehat{\nu}$, and $\widehat{\eta}$. However, in practice we found that this procedure of estimating the ratio of $\widehat{\eta}$ to $\widehat{\nu}$ can lead to instability. Additionally, since this is a different parameterization of the problem, it does not have the same double-robustness properties established in the next section. Despite these caveats, we still include it  due to its simplicity and convenience; it may be useful in future work. Next, we derive the asymptotic distribution and convergence rate for $\widehat{\psi}(\delta)$.
\end{remark}

\subsection{Asymptotic Distribution \& Convergence Rate}

In this section we establish the asymptotic distribution and convergence rate of our proposed estimator of the incremental effect for continuous exposures. We illustrate how traditional methods for establishing asymptotic properties, like those utilized in \cite{diaz2020causal}, may over or understate errors by not considering the size of the exponential tilt  $\delta$. Therefore, our goal in this section is to derive new results that remove and/or make explicit the dependence of the tilt $\delta$. It is important to make this dependence on $\delta$ 
 explicit for a number of reasons. First, our analysis is more accurate since it shows how the convergence rate actually depends on the tilt $\delta$. Furthermore, letting $\delta \to \infty$ or $\delta \to -\infty$ is interesting in its own right, since it yields novel estimators of the dose-response curve. Finally,  practitioners typically estimate $\psi(\delta)$ across a range of values; establishing an explicit dependence on $\delta$ allows them to choose a range motivated by the sample size. 

It is instructive to consider a standard decomposition of the form
\begin{align*}
    \widehat{\psi}(\delta) - \psi(\delta) &= (\mathbb{P}_n - \mathbb{P})\{\varphi(Z; \mathbb{P})\} + (\mathbb{P}_n - \mathbb{P})\{\varphi(Z; \widehat{\mathbb{P}}) - \varphi(Z; \mathbb{P})\} + R_2(\widehat{\mathbb{P}}, \mathbb{P}) 
\end{align*}
where $R_2(\widehat{\mathbb{P}}, \mathbb{P}) = \psi(\widehat{\mathbb{P}}) - \psi(\mathbb{P}) + \int \phi(z; \widehat{\mathbb{P}}) d\mathbb{P}(z)$, as discussed in \cite{kennedy2023semiparametric}. The first term is a centered sample average of a fixed function, so after scaling by $\sqrt{n}$ it converges to a normally distributed random variable by the Central Limit Theorem. The second term, often referred to as the empirical process term, is typically of smallest order, since it is a centered sample average of a random variable with shrinking variance. The third term, $R_2(\widehat{\mathbb{P}}, \mathbb{P})$, is arguably most important as it represents bias; it is ideally second-order (i.e., involving products/squares of differences between $\widehat\Pb$ and $\Pb$) in order for it not to dominate the other terms. For small finite values of $\delta$,  asymptotic normality can be established using standard methods; however, these methods otherwise break down, as each of the terms depends on $\delta$ in complicated ways. To show this, we first consider the remainder $R_2(\widehat{\mathbb{P}}, \mathbb{P})$.

\begin{proposition} \label{remainder_term}
    The remainder of the von Mises expansion  $\psi(\widehat{\mathbb{P}}) - \psi(\mathbb{P}) + \int \varphi(z; \widehat{\mathbb{P}}) \ d\mathbb{P}(z)$ is given by $R_2(\widehat{\mathbb{P}}, \mathbb{P}) = R_1 - R_2$ where
    \begin{align*}
        R_1 &=  \mathbb{E} \left[\int_a \left(\frac{\widehat{q}_\delta}{\widehat{\pi}_a} - \frac{q_\delta}{\pi_a} \right)\Big((\pi_a - \widehat{\pi}_a) \widehat{\mu}_a + (\mu_a - \widehat{\mu}_a) \pi_a \Big) da\right] \\
        R_2 &= \mathbb{E} \left[\left(\int_a \frac{q_\delta}{\pi_a} \widehat{\mu}_a\widehat{\pi}_a da \right)\left(\int_a \frac{\widehat{q}_\delta}{\widehat{\pi}_a} (\pi_a - \widehat{\pi}_a) da \right)^2\right]
    \end{align*}
    for $\pi_a = \pi( a \mid X)$, $q_\delta = q_\delta(a \mid X)$, and $\mu_a = \mu(X,  a)$.
\end{proposition}

Importantly, \cref{remainder_term} demonstrates that the remainder exhibits many of the same issues addressed in \cref{sigma_bounds}; namely, $R_1$ depends on the squared likelihood ratio ${q_\delta(a \mid x)}/{\pi(a \mid x)}$ which is lower bounded by $\delta$. Consequently,  traditional methods of bounding the remainder (like the Cauchy-Schwarz inequality) lead to an unsatisfactory dependence on $\delta$; this is a key point where our  analysis differs from that of \cite{diaz2020causal}. We elaborate on these differences below.
\begin{remark} \label{rem:csbd}
    The usual way to bound $R_2(\widehat{\mathbb{P}}, \mathbb{P})$ is via the Cauchy-Schwarz inequality, to obtain a doubly-robust product of $L_2(\Pb)$ errors. This is the approach utilized in \cite{diaz2020causal}, where the authors argue that
\begin{align} \label{l2_upper_bound}
    R_2(\widehat{\mathbb{P}}, \mathbb{P}) &\lesssim \left| \left| \frac{\widehat{q}_\delta}{\widehat{\pi}} - \frac{q_\delta}{\pi} \right| \right|_2 \Big( \left| \left| \widehat{\pi} - \pi \right| \right|_2 + \left| \left| \widehat{\mu} - \mu \right| \right|_2\Big) + \left| \left| \widehat{\pi} - \pi \right| \right|^2_2
\end{align}
for $||f||^2_2 = \int f^2(z) d \mathbb{P}(z)$. However, \cref{l2_upper_bound} is suboptimal in several ways. First, strong positivity is required in the proof of this expression. Positivity is not required for identification, and so we would ideally also avoid it in estimation. For example, in our analysis we only use a weak version of positivity, allowing holes in the support, as described in \cref{minimax_section}. Second, the $L_2(\mathbb{P})$-norm of the difference in likelihood ratios has an irreducible dependence on $\delta$; specifically, it is lower bounded by $\delta$, so the norm becomes unbounded as $\delta$ grows large. In fact, if we let $\gamma(\delta) = \int_a \exp(\delta a)\pi_a da$, it follows that
\begin{align} \label{l2_lower_bound}
    \delta  \left| \left| \frac{\gamma(\delta) }{\widehat{\gamma}(\delta)} - 1 \right| \right|^2_2 \lesssim \left| \left| \frac{\widehat{q}_\delta}{\widehat{\pi}} - \frac{q_\delta}{\pi} \right| \right|^2_2,
\end{align}
so as long as $|\widehat{\pi}_a - \pi_a|$ is not exactly zero across the support of $A$, then the $L_2(\mathbb{P})$ norm is lower bounded by $\delta$. Consequently, for larger tilts, the convergence rates established in \cite{diaz2020causal} can be inaccurate due to the omission of $\delta$.
\end{remark}

Surprisingly, it is possible to remove any dependence on $\delta$ and still obtain a doubly-robust style upper bound on $R_2(\widehat{\mathbb{P}}, \mathbb{P})$, via a new analysis. Indeed, it is also possible to do so without requiring positivity across the entire support; we do so with the same weak positivity from \cref{sigma_bounds}, which we establish in the following theorem. 

\begin{theorem} \label{remainder_bound}
    Suppose $\text{min}\left\{ \pi_a, \widehat{\pi}_a \right\} \geq \pi_{\min} >0$ for all $a \in  [\eta, 1]$ for some $\eta \in [0, 1)$ and that $\pi_a, \widehat{\pi}_a$, and $\widehat{\mu}$ are uniformly bounded from above. Then,
    \begin{align*}
        R_2(\widehat{\mathbb{P}}, \mathbb{P}) \lesssim  || \widehat{\pi} - \pi||_{L^2_x, L^\infty_a} \Big(  || \widehat{\mu} - \mu ||_{L^2_x, L^\infty_a}  +  || \widehat{\pi} - \pi||_{L^2_x, L^\infty_a} \Big)
    \end{align*}
    where $||f||^2_{L^2_x, L^\infty_a} =  \int \big( \underset{a}{\text{sup}} \, |f(x, a) | \big)^2 d \mathbb{P}(x)$ is a mixed $L_2(\mathbb{P})$-sup norm.
\end{theorem}

The key to removing the dependence on $\delta$ in the remainder bound lies in using a different characterization of the  errors of $\widehat{\pi}$ and $\widehat{\mu}$ across the support of $a$, via a mixed $L_2(\Pb)$-sup norm. This is necessary since,  when using standard $L_2$ norms, the tilt size $\delta$ essentially acts as a penalty on the nuisance estimation error; if either is not well estimated, then larger $\delta$ values can exacerbate the bias, as described in \cref{rem:csbd}. By instead requiring that the $L_2$ error over $X$ is uniformly controlled over the support of $A$, the bias can be bounded completely independently of $\delta$. Importantly, when it comes to mixed $L_2(\Pb)$-sup norm control over the estimation error, we may estimate $\pi$ as described in \cref{proposed_estimator_section}. However, uniform control for nonparametric estimation of $\mu$ requires some structure over low-density regions; for example, a lower bound on the joint density of $(X, A)$ on the region of uniformity or a model allowing for extrapolation \citep{hansen_uniform}. Having established this new upper bound for the remainder, we now use it to give weak conditions for asymptotic normality. Importantly, we explicitly show how the convergence rate is a function of $\delta$.

\begin{theorem} \label{asymptotic_normality}
Assume $\text{min}\left\{ \pi_a, \widehat{\pi}_a \right\} \geq \pi_{\min} >0$ for all $a \in  [\eta, 1]$ for some $\eta \in [0, 1)$,  that $\pi_a, \widehat{\pi}_a$, and $\widehat{\mu}$ are uniformly bounded from above, $\mathbb{P}(|Y| < C) = 1$ for some $C < \infty$, and
\begin{enumerate}
    \item[(i)] $\sqrt{n / \delta} \to \infty$
    \item[(ii)] $ || \widehat{\pi} - \pi||_{L^2_x, L^\infty_a} = o_{\mathbb{P}}(1)$ and  $ || \widehat{\mu} - \mu ||_{L^2_x, L^\infty_a} =  o_{\mathbb{P}}(1)$
    \item[(iii)] $ || \widehat{\pi} - \pi||_{L^2_x, L^\infty_a} \cdot   || \widehat{\mu} - \mu ||_{L^2_x, L^\infty_a}  +  || \widehat{\pi} - \pi||^2_{L^2_x, L^\infty_a} = o_{\mathbb{P}}\left(1 / \sqrt{n / \delta}\right)$
\end{enumerate}
Then 
\begin{align*}
    \frac{\sqrt{n}}{\sigma_\delta}\left(\widehat{\psi}(\delta) - \psi(\delta)\right)\overset{d}{\longrightarrow}N(0,1)
\end{align*}
where $\sigma^2_\delta = \var\{\varphi(Z;\delta)\}$ is the nonparametric efficiency bound defined in \cref{nonpar_efficiency_bound}.
\end{theorem}

\cref{asymptotic_normality} shows that the estimator is still asymptotically  normal after proper scaling, albeit with a non-standard $\sqrt{n/\delta}$ rate of convergence instead of the more usual $\sqrt{n}$ rate. Establishing a $\delta$-dependent convergence rate clarifies the interplay between $n$ and $\delta$ and the estimation error; specifically, this illustrates how the effective sample size for continuous incremental effect estimation is $n/\delta$ instead of $n$. In the next section, we consider the setting where $\delta \to \infty$, and show that this can be used to recover a novel estimator of the dose-response curve at the edge of the support.

\section{Dose-Response Estimation} \label{dose_response_section}

In this section we discuss how incremental effect estimation at the infinite $\delta$ extremes can yield new nonparametric estimators of the dose-response curve. Recall that, as illustrated in \cref{Exp_tilted_density}, when $\delta \to -\infty$ or $\delta \to \infty$ then the tilted density converges to a point mass at the boundary of the treatment support. In fact, we can extend this framework to estimate the dose-response at any interior $a^\prime \in [0, 1]$ by applying a positive tilt $(\delta \to \infty)$ for all $a < a^\prime$ and a negative tilt $(\delta \to -\infty)$ for all $a > a^\prime$, which yields the ``reflected'' exponentially tilted intervention distribution,
\begin{align*}
    r_\delta(a \mid x) =  \omega_{a^\prime}\frac{\text{exp}(\delta a) \pi(a \mid x) \mathbbm{1}(a \leq a^\prime) }{\int^{a^\prime}_0 \text{exp}(\delta t) \pi(t \mid x) dt} +  (1-\omega_{a^\prime})\frac{\text{exp}(-\delta a) \pi(a \mid x) \mathbbm{1}(a > a^\prime)}{\int^{1}_{a^\prime} \text{exp}(-\delta t) \pi(t \mid x) dt}
\end{align*}
where $\omega_{a^\prime} = \int^{a^{\prime}}_0 \pi(t \mid x) dt$ applies a weight based on the probability mass of $\pi(a \mid x)$ above and below $a^\prime$. Consequently, we can estimate the incremental effect under the reflected exponential tilt,
\begin{align*}
    \psi_R(\delta, a^\prime) =  \int_\mathcal{X} \int_\mathcal{A} \mu(x, a) r_\delta(a \mid x) \, da \, d\mathbb{P}(x)
\end{align*}
to obtain an estimate of $\mathbb{E}[Y^{a^\prime}]$ by taking $\delta \to \infty$. The following theorem establishes the efficient influence function of the incremental effect under $r_\delta(a \mid x)$.
\begin{theorem} \label{reflected_eif}
    Suppose $\delta \in [-M,M]$ for $M<\infty$. Then, the efficient influence function of $\psi_R(\delta, a^\prime)$ under a nonparametric model is given $\varphi_R(Z;\delta, a^\prime) = D_Y + D_{r, \mu} + D_\psi$ where
\begin{align*}
    D_Y &= \frac{r_\delta(A \mid X)}{\pi(A \mid X)} \Big(Y - \mu(X, A) \Big) \\
    D_{r, \mu} &= \frac{r_\delta(A \mid X)}{\pi(A \mid X)} \mu(X, A) + \Big(\mathbbm{1}(A \leq a^\prime)L(X, A) + \mathbbm{1}(A > a^\prime) U(X, A) \Big) - \mathbb{E}_R(\mu(X, A) \mid X) \\
    D_\psi &= \mathbb{E}_R(\mu(X, A) \mid X)  - \psi_R(\delta; a^\prime)
\end{align*}
where $\mathbb{E}_R(\mu(X, A) \mid X) = \int_a \mu(X, a) r_\delta(a \mid X) da$ is the conditional mean of $\mu(X, A)$ under the reflected exponentially tilted distribution and
\begin{align*}
    L(X, A) &= \mathbb{E}_R(\mu(X, A) \mid X, A \leq a^\prime) - \frac{r_\delta(A \mid X)}{\pi(A \mid X)} \mathbb{E}_R(\mu(X, A) \mid X, A \leq a^\prime) \\
    U(X, A) &= \mathbb{E}_R(\mu(X, A) \mid X, A > a^\prime) - \frac{r_\delta(A \mid X)}{\pi(A \mid X)} \mathbb{E}_R(\mu(X, A) \mid X, A > a^\prime).
\end{align*}
\end{theorem}

By \cref{reflected_eif}, we can see that the efficient influence function of the reflected incremental effect is quite similar in composition to the influence function described in \cref{eif_theorem}; in fact, taking $a^\prime \to 1$ or $a^\prime \to 0$ recovers $\varphi(Z; \delta)$ exactly. Now, the behavior of $\varphi_R(Z; \delta, a^\prime)$ is governed by the ratio $r_\delta(a \mid x) / \pi(a \mid x)$ instead of $q_\delta(a \mid x) / \pi(a \mid x)$. Furthermore, the adjustment term $D_{r, \mu}$ is more intricate; rather than a single term involving $\mu(X, A)$, the reflection introduces $L(X, A)$ and $U(X, A)$ to account for differing behavior of $\mu(X, A)$ under $r_\delta(a \mid x)$ above and below $a^\prime$. However, $\varphi_R(Z; \delta, a^\prime)$ still scales with $\delta$ at the same rate described in \cref{sigma_bounds}. Furthermore, asymptotic normality of the one-step estimator for $\psi_R(\delta, a^\prime)$ can be established under identical conditions to those described in \cref{asymptotic_normality}, although now we require weak positivity in a neighborhood of $a^\prime$ rather than the edges of the treatment support. Since these results are largely the same as before, we provide a longer, more technical discussion of reflected incremental effect in the supplementary material. The following theorem establishes the rate of convergence of $\widehat{\psi}_R(\delta, a^\prime)$, as an estimator of the dose-response curve.

\begin{theorem} \label{dose_response_theorem}
    Suppose that $|\mu(x, a) - \mu(x, a^\prime)| \leq L|a - a^\prime|$ for all $x$ and that the assumptions of \cref{asymptotic_normality} hold, where now weak positivity holds in a neighborhood of $a^\prime$, i.e. that $\pi_{\min} \leq \pi(a \mid x)$ for all $a \in  [a^{\prime} - \eta, a^\prime + \eta]$ for some $\eta \in (0, \text{min}\{a^\prime, 1 - a^\prime\}]$. Then,
    \begin{align*}
        \left|\widehat{\psi}_R(\delta, a^\prime) - \E(Y^{a^\prime}) \right| = O_\mathbb{P}\left(\frac{1}{\delta} + \sqrt{\frac{1}{n/\delta}} \right)
    \end{align*}
\end{theorem}

In this result, the $1/\delta$ term is the bias and the $1/(n/\delta)$ term is the variance, illustrating a bias-variance trade-off in $\delta$: larger $\delta$ yields smaller bias for estimating the dose-response curve $\E(Y^{a^\prime})$, but larger variance, and vice versa. Thus, for the purposes of dose-response estimation $\delta$ plays the role of a bandwidth parameter. Balancing the  bias and standard deviation by taking $\delta \sim n^{1/3}$ yields the rate $ |\widehat{\psi}_R(n^{1/3}, a^\prime) - \E(Y^{a^\prime}) | = O_\mathbb{P}(n^{-1/3} )$ which is the minimax optimal rate for estimating Lipschitz regression functions. An interesting distinction between the incremental effect estimator of the dose-response and a more standard kernel estimator (e.g., from \cite{kennedy2017cte}) is that the latter is localized, only using treatment values near the target point, where the former uses all treatment values, albeit heavily weighted towards the target point $A=a^\prime$.

In the next result, we show that after scaling by $\delta$, the asymptotic variance of the incremental effect estimator (i.e., the variance of $\varphi$) is identical to that of the kernel estimator used in \cite{kennedy2017cte} when $\delta \rightarrow \infty$.

\begin{corollary}\label{sigma_delta_ratio}
    Suppose that $\pi$, $\mu$, and $\var(Y \mid \cdot, \cdot)$ are Lipschitz continuous, $\pi$ is uniformly bounded from above, $|Y|\leq B$ with probability one,  and $\pi_{\min} \leq \pi(a \mid x)$ for all $a \in  [a^{\prime} - \eta, a^\prime + \eta]$ for some $\eta \in (0, \text{min}\{a^\prime, 1 - a^\prime\}]$. Then, it follows that
    \begin{align*}
        \underset{\delta \to \infty}{\text{lim}}\left\{ \frac{ \var(\varphi_R(Z;\delta, a^\prime)) }{\delta}\right\} =  \mathbb{E} \left[ \frac{1}{2} \left(\frac{\var(Y\mid X,a^\prime)}{\pi(a^\prime\mid X)} \right) \right] .
    \end{align*}
\end{corollary}

\cref{sigma_delta_ratio} shows that $\psi_R(\delta, a^\prime)$ does not suffer any loss of efficiency in estimating the dose-response curve and (perhaps surprisingly) matches the performance of the DR-learner-style estimator of \cite{kennedy2017cte}, despite only requiring weak positivity around $a^\prime$.

In addition to confirming how incremental effect estimators for continuous treatments can be harnessed for dose-response estimation, the results in this section point towards interesting directions for future work. For example, it would be useful to explore whether the tilted density can be viewed more formally as a kind of kernel weighting, whether incremental effect-based estimators can ever exploit higher-order Holder smoothness, beyond the Lipschitz smoothness used here, data-driven tuning of $\delta$, etc.

\section{Application} \label{experiments_and_simulations}

In this section we estimate the incremental effect of political advertisements on campaign contributions (as well as the associated dose-response curve), following \cite{urban2014dollars}. Typically, political campaigns tend to avoid advertising in states that are not competitive, since this could be a poor use of a limited budget on something unlikely to change election outcomes. However, it could be the case that advertising in non-competitive states could meaningfully increase campaign contributions, which in turn could help finance the campaign. To assess the relationship between the number of political advertisements and campaign contributions, \cite{urban2014dollars} make use of the fact that media markets often cross state boundaries; for example, residents of southwest Illinois (a non-competitive state) were exposed to advertisements in the St. Louis market (located in a competitive state). Therefore, by restricting their focus to non-competitive states the authors can avoid spillover effects due to other campaign events that may also drive donations.

\begin{figure}[h]
\centering
\begin{tikzpicture}  

\begin{axis}[
    xlabel={$\text{log}(\text{Total Ads} + 1)$},
    ylabel={Probability Density},
    ylabel style={yshift=0.1cm},
    width=0.7\textwidth,
    height=0.45\textwidth,
    ymax=1,
    grid=both,
    bar width=0.7cm,
    grid style={line width=0.1pt, draw=gray!20},
    legend cell align=left,
    xtick pos=bottom,
    ytick pos=left,
    legend style={at={(1.05,0.5)}, anchor=west},
]

\definecolor{brightblue}{HTML}{00BFC4}
  \definecolor{ggred}{HTML}{fb2424}
  
    \addplot[hist={bins=25, density}, forget plot, style={fill=brightblue!25}] table [y index=0] {files/histogram_data.csv};

\addplot[color=gray, opacity=0.75, forget plot] table[x="a", y="delta_0.05", col sep=comma] {files/density_data.csv};
\addplot[color=gray, opacity=0.75, forget plot] table[x="a", y="delta_0.1", col sep=comma] {files/density_data.csv};
\addplot[color=gray, opacity=0.75, forget plot] table[x="a", y="delta_0.2", col sep=comma] {files/density_data.csv};
\addplot[color=gray, opacity=0.75, forget plot] table[x="a", y="delta_0.3", col sep=comma] {files/density_data.csv};
\addplot[color=gray, opacity=0.75, forget plot] table[x="a", y="delta_0.4", col sep=comma] {files/density_data.csv};
\addplot[color=gray, opacity=0.75, forget plot] table[x="a", y="delta_0.5", col sep=comma] {files/density_data.csv};
\addplot[color=gray, opacity=0.75, forget plot] table[x="a", y="delta_1", col sep=comma] {files/density_data.csv};
\addplot[color=gray, opacity=0.75, forget plot] table[x="a", y="delta_1.25", col sep=comma] {files/density_data.csv};
\addplot[color=gray, opacity=0.75, forget plot] table[x="a", y="delta_1.5", col sep=comma] {files/density_data.csv};
\addplot[color=gray, opacity=0.75, forget plot] table[x="a", y="delta_1.75", col sep=comma] {files/density_data.csv};
\addplot[color=gray, opacity=0.75, forget plot] table[x="a", y="delta_2", col sep=comma] {files/density_data.csv};
\addplot[color=gray, opacity=0.75, forget plot] table[x="a", y="delta_2.5", col sep=comma] {files/density_data.csv};
\addplot[color=gray, opacity=0.75, forget plot] table[x="a", y="delta_5", col sep=comma] {files/density_data.csv};

\addplot[very thick, black, smooth] table[x="a", y="pi", col sep=comma] {files/density_data.csv};
\addlegendentry{$\delta = 0$}

\addplot[color=gray, opacity=0.75] table[x="a", y="delta_0.8", col sep=comma] {files/density_data.csv};
\addlegendentry{$\delta \in (0, 10)$}

\addplot[color=ggred, ultra thick] table[x="a", y="delta_10", col sep=comma] {files/density_data.csv};
\addlegendentry{$\delta = 10$}
    
\end{axis}
\end{tikzpicture}
\caption{Histogram of the treatment variable, \textit{Total Ads}}
\label{total_ads}
\end{figure}

Although \cite{urban2014dollars} are interested in the number of political advertisements as a treatment variable, they dichotimize this measurement in order to use standard propensity score matching techniques to estimate the effect on campaign contributions. However, dichotimizing a continuous treatment variable loses potentially valuable information and can complicate interpretation of effects. Incremental effects provide an interesting alternative for exploring how contributions might change if the distribution of political advertisements shifted up or down. We use the data collected by \cite{urban2014dollars} (freely available through the \href{https://dataverse.harvard.edu/dataset.xhtml?persistentId=doi:10.7910/DVN/AIF4PI}{Harvard Dataverse} via \cite{fong2018covariate}), which features Federal Election Commission data on campaign contributions, Wisconsin Advertising Project data on televised campaign ads, and demographic data from the 2000 U.S. census. In total there are 16,265 observations, each representing one ZIP code. We keep the same set of covariates included in the analysis from \cite{urban2014dollars} and \cite{fong2018covariate}: log total population, log median household income, population density, percent Hispanic, percent African American, percent college graduate, percent over the age of 65, and an indicator for whether  it is possible to commute to the ZIP code from a competitive state or not. Notably, the treatment variable (log of total political advertisements by ZIP code) has holes in its density, as shown in \cref{total_ads}, further motivating the use of incremental effects to estimate the relationship between political advertisements and campaign contributions.

In our analysis, we estimate $\mu(x, a)$ by fitting a Random Forest using the \texttt{ranger} package in \texttt{R}. Similarly, we estimate $\pi(a \mid x)$ using \texttt{ranger} on the kernel-transformed outcomes as discussed in \cref{proposed_estimator_section}, where we select the bandwidth $h$ by cross-validation. Note that other flexible regression methods (or ensembles thereof) could also be used here as well. For dose-response estimation we let $\delta = n^{1/3}$ since this matches the minimax optimal rate for estimating Lipschitz regression functions, as discussed in \cref{dose_response_section}. Additionally, in \cref{simulation_section}, we provide a simulation study that showcases width and coverage as a function of $\delta$ for different orders of $n$, as well as a longer discussion on how one chooses $\delta$ in practice.

\begin{figure}[h]
  \centering
  \begin{tikzpicture}
    \begin{groupplot}[
      group style={group size=2 by 1, horizontal sep=1.5cm},
      width=0.495\textwidth,
      height=0.45\textwidth,
      grid=both,
      grid style={line width=0.1pt, draw=gray!20, opacity=0.5},
      legend cell align=left,
      xtick pos=bottom,
      ytick pos=left,
    ]

    \nextgroupplot[
      xlabel={$\delta$},
      ylabel={$\widehat{\psi}(\delta)$},
              ylabel style={yshift=-0.25cm},
    ]

      \definecolor{brightblue}{HTML}{00BFC4}

      \pgfplotstableread[col sep=comma]{files/incremental_plot_updated.csv}\incrementaltable

      \addplot[color=black]
        table[x="delta", y="psihat"] {\incrementaltable};

      \addplot[name path=ci_lower, color=brightblue, dashed]
        table[x="delta", y="ci_lower"] {\incrementaltable};
      \addplot[name path=ci_upper, color=brightblue, dashed]
        table[x="delta", y="ci_upper"] {\incrementaltable};
      \addplot[brightblue!50, fill opacity=0.3]
        fill between[of=ci_lower and ci_upper];

\nextgroupplot[
  xlabel={$\log(\text{Total Ads}+1)$},
  ylabel={Dose-Response},
  ylabel style={yshift=-0.25cm},
  xmin=0, xmax=6,
  ymin=-0.5, ymax=2.5,
  restrict x to domain*=0:6,
  restrict y to domain*=-0.5:2.5,
  unbounded coords=discard,
  filter discard warning=false,
]

      \definecolor{npred}{HTML}{F8766D}    
      \definecolor{reflblue}{HTML}{00BFC4} 

      \pgfplotstableread[col sep=comma]{files/ctseff_new.csv}\ctstable
      \pgfplotstableread[col sep=comma]{files/reflected_dose_df_final.csv}\refltable

      \addplot[name path=cts_low,  draw=npred!60, forget plot]
        table[x="a.vals", y="ci.ll"]  {\ctstable};
      \addplot[name path=cts_high, draw=npred!60, forget plot]
        table[x="a.vals", y="ci.ul"] {\ctstable};
      \addplot[npred!40, fill opacity=0.25, forget plot]
        fill between[of=cts_low and cts_high];

      \addplot[draw=npred, thick]
        table[x="a.vals", y="est"] {\ctstable};

      \addplot[name path=refl_low,  draw=reflblue!60, forget plot]
        table[x="orig_a", y="ci_lower"]  {\refltable};
      \addplot[name path=refl_high, draw=reflblue!60, forget plot]
        table[x="orig_a", y="ci_upper"] {\refltable};
      \addplot[reflblue!40, fill opacity=0.25, forget plot]
        fill between[of=refl_low and refl_high];

      \addplot[draw=reflblue, thick]
        table[x="orig_a", y="psi_hat"] {\refltable};
    \end{groupplot}
  \end{tikzpicture}
  \caption{Left: Estimated incremental effect of total political advertisements on individual log campaign contributions. Right: estimated dose-response $a \mapsto \mathbb{E}[Y^a]$ via $\psi_R(\delta, a)$ (blue) and the nonparametric estimator of \cite{kennedy2017cte} (red).}
  \label{fig:incremental_and_dose_response}
\end{figure}

Our findings complement and validate those of \cite{urban2014dollars}. As visualized in the left plot of \cref{fig:incremental_and_dose_response}, we find a slight increasing relationship between total political advertisements per ZIP code and log campaign contributions. In particular, we find that moving from $\delta = 0$ to $\delta = 5$ corresponds roughly to an increase in average individual contributions from \$2.00 to \$2.20, where most of the benefit occurs for $\delta < 1$ and then levels out. Aggregating that difference across all residents could suggest substantial benefits to advertising in noncompetitive ZIP codes. On the other hand, we can see that \textit{decreasing} the the number of campaign ads leads to a substantial and asymmetric loss of campaign contributions. Given the small potential upside, but large potential downside, it would suggest that the 2008 presidential campaign was already well optimized. 

In the case of dose-response estimation (as shown in the right plot of \cref{fig:incremental_and_dose_response}), we find that the reflected incremental effect appears to struggle less in low-density regions than nonparametric kernel-based approach of \cite{kennedy2017cte}. One plausible explanation could be that the reflected incremental effect relies globally on $\pi(a \mid x)$, so it may be less sensitive to localized areas of low density. Aside from this, the two curves largely agree with each other. These results are promising, and suggest that dose-response estimation via incremental effects could be a fruitful area for future research.

\section{Discussion} \label{discussion_section}

In this paper we considered nonparametric estimation of incremental effects for continuous exposures. The incremental effect is defined by exponentially tilting the observed treatment density, which corresponds to an increase or decrease in the log-likelihood of receiving a given treatment level. We showed that incremental effects can be identified without requiring the positivity assumption, although estimation and inference require a weakened version of positivity. Importantly, in our analysis we emphasized that estimation is highly dependent on the size of the exponential tilt, as measured by the parameter $\delta$. To demonstrate this, we derived a minimax lower bound illustrating the fundamental limits of incremental effect estimation in terms of both sample size $n$ and tilt size $\delta$. This result shows that the effective sample size for incremental effect estimation is really $n/\delta$ instead of $\delta$. We developed a new analysis of the bias of a doubly robust incremental effect estimator, showing that $n/\delta$ rates can in fact be achieved under weak nonparametric conditions. Finally, we showed that taking $\delta \to \infty$ yields a new incremental effect-based estimator of  the dose-response curve, and analyzed its convergence rate, showing it can achieve optimal rates for Lipschitz functions.

There are lots of avenues for future work, including extending to time-varying and/or multivariate treatments, weakening the positivity assumption, further exploring dose-response estimation (e.g., cross-validation, higher smoothness, etc.), developing sensitivity analysis methods in the presence of unmeasured confounders, extending to other identification schemes, exploring derivative-based effects for continuous exposures \citep{zhou2022marginalinterventionaleffects, McCleanBransonKennedy+2024}, and more. 

\section*{Acknowledgments}
This research was supported by NSF CAREER Award 2047444  and  a fellowship provided by the Novartis Pharmaceuticals Corporation. The authors thank both Alec McClean and Alex Levis for very helpful discussions, and in particular Alex Levis for his construction of the Kumaraswamy distribution treatment density counterexample.

\bigskip

\section*{References}
\vspace{-1cm}
\bibliographystyle{abbrvnat}
\bibliography{references}

@article{kennedy2019nonparametric,
  title={Nonparametric causal effects based on incremental propensity score interventions},
  author={Kennedy, Edward H},
  journal={Journal of the American Statistical Association},
  volume={114},
  number={526},
  pages={645--656},
  year={2019},
  publisher={Taylor \& Francis}
}

@misc{schindl2025causalgeodesycounterfactualestimation,
      title={Causal Geodesy: Counterfactual Estimation Along the Path Between Correlation and Causation}, 
      author={Kyle Schindl and Larry Wasserman},
      year={2025},
      eprint={2508.08499},
      archivePrefix={arXiv},
      primaryClass={stat.ME},
      url={https://arxiv.org/abs/2508.08499}, 
}

@book{kish_65,
  author    = {Kish, Leslie},
  title     = {Survey Sampling},
  year      = {1965},
  address   = {New York},
  publisher = {John Wiley \& Sons, Inc.},
}

@article{haneuserotnitzky2013,
author = {Haneuse, S. and Rotnitzky, A.},
title = {Estimation of the effect of interventions that modify the received treatment},
journal = {Statistics in Medicine},
volume = {32},
number = {30},
pages = {5260-5277},
doi = {https://doi.org/10.1002/sim.5907},
url = {https://onlinelibrary.wiley.com/doi/abs/10.1002/sim.5907},
eprint = {https://onlinelibrary.wiley.com/doi/pdf/10.1002/sim.5907},
year = {2013}
}

@article{kennedy2017cte,
 ISSN = {13697412, 14679868},
 URL = {https://www.jstor.org/stable/26773159},
 author = {Edward H. Kennedy and Zongming Ma and Matthew D. McHugh and Dylan S. Small},
 journal = {Journal of the Royal Statistical Society. Series B (Statistical Methodology)},
 number = {4},
 pages = {1229--1245},
 publisher = {[Royal Statistical Society, Wiley]},
 title = {Non-parametric methods for doubly robust estimation of continuous treatment effects},
 urldate = {2024-01-18},
 volume = {79},
 year = {2017}
}

@article{munoz2012population,
  title={Population intervention causal effects based on stochastic interventions},
  author={D{\'\i}az, Iv{\'a}n  and {van der Laan}, Mark},
  journal={Biometrics},
  volume={68},
  number={2},
  pages={541--549},
  year={2012},
  publisher={Wiley Online Library}
}

@article{McCleanBransonKennedy+2024,
url = {https://doi.org/10.1515/jci-2023-0024},
title = {Nonparametric estimation of conditional incremental effects},
title = {},
author = {Alec McClean and Zach Branson and Edward H. Kennedy},
pages = {20230024},
volume = {12},
number = {1},
journal = {Journal of Causal Inference},
doi = {doi:10.1515/jci-2023-0024},
year = {2024},
lastchecked = {2025-12-11}
}

@article{diaz2020causal,
  title={Causal mediation analysis for stochastic interventions},
  author={D{\'\i}az, Iv{\'a}n and Hejazi, Nima S},
  journal={Journal of the Royal Statistical Society Series B: Statistical Methodology},
  volume={82},
  number={3},
  pages={661--683},
  year={2020},
  publisher={Oxford University Press}
}

@article{escher1932probability,
  title={On the probability function in the collective theory of risk},
  author={Esscher, F},
  journal={Skand. Aktuarie Tidskr.},
  volume={15},
  pages={175--195},
  year={1932}
}

@book{gerber1993option,
  title={Option pricing by Esscher transforms},
  author={Gerber, Hans U and Shiu, Elias SW and others},
  year={1993},
  publisher={HEC Ecole des hautes {\'e}tudes commerciales}
}

@article{elliott2005option,
  title={Option pricing and Esscher transform under regime switching},
  author={Elliott, Robert J and Chan, Leunglung and Siu, Tak Kuen},
  journal={Annals of Finance},
  volume={1},
  pages={423--432},
  year={2005},
  publisher={Springer}
}

@misc{kennedy2023semiparametric,
      title={Semiparametric doubly robust targeted double machine learning: a review}, 
      author={Edward H. Kennedy},
      year={2023},
      eprint={2203.06469},
      archivePrefix={arXiv},
      primaryClass={stat.ME}
}

@article{bayes_emp_2005,
    author = {Schennach, Susanne M.},
    title = "{Bayesian exponentially tilted empirical likelihood}",
    journal = {Biometrika},
    volume = {92},
    number = {1},
    pages = {31-46},
    year = {2005},
    month = {03},
    issn = {0006-3444},
    doi = {10.1093/biomet/92.1.31}
}

@article{vernon_johns_1988,
author = {M. Vernon Johns},
title = {Importance Sampling for Bootstrap Confidence Intervals},
journal = {Journal of the American Statistical Association},
volume = {83},
number = {403},
pages = {709--714},
year = {1988},
publisher = {Taylor \& Francis},
doi = {10.1080/01621459.1988.10478651}

}

@article{rare_event,
    author = {Fithian, William and Wager, Stefan},
    title = "{Semiparametric exponential families for heavy-tailed data}",
    journal = {Biometrika},
    volume = {102},
    number = {2},
    pages = {486-493},
    year = {2014},
    month = {03},
    issn = {0006-3444},
    doi = {10.1093/biomet/asu065}
}

@Inbook{tsbakov2009,
  title = "Introduction to Nonparametric Estimation",
  author = "Tsybakov, Alexandre B.",
  year = "2009",
  publisher = "Springer"
}

@misc{rubio2024populationsizeestimationlists,
      title={Population Size Estimation with Many Lists and Heterogeneity: A Conditional Log-Linear Model Among the Unobserved}, 
      author={Mateo Dulce Rubio and Edward Kennedy},
      year={2024},
      eprint={2407.03539},
      archivePrefix={arXiv},
      primaryClass={stat.ME},
      url={https://arxiv.org/abs/2407.03539}, 
}

@article{siegmund1976importance,
  title={Importance sampling in the Monte Carlo study of sequential tests},
  author={Siegmund, David},
  journal={The Annals of Statistics},
  pages={673--684},
  year={1976},
  publisher={JSTOR}
}

@article{urban2014dollars,
  title={Dollars on the sidewalk: Should US presidential candidates advertise in uncontested states?},
  author={Urban, Carly and Niebler, Sarah},
  journal={American Journal of Political Science},
  volume={58},
  number={2},
  pages={322--336},
  year={2014},
  publisher={Wiley Online Library}
}

@article{fong2018covariate,
  title={Covariate balancing propensity score for a continuous treatment: Application to the efficacy of political advertisements},
  author={Fong, Christian and Hazlett, Chad and Imai, Kosuke},
  journal={The Annals of Applied Statistics},
  volume={12},
  number={1},
  pages={156--177},
  year={2018},
  publisher={JSTOR}
}

@article{young2014identification,
  title={Identification, estimation and approximation of risk under interventions that depend on the natural value of treatment using observational data},
  author={Young, Jessica G and Hern{\'a}n, Miguel A and Robins, James M},
  journal={Epidemiologic methods},
  volume={3},
  number={1},
  pages={1--19},
  year={2014},
  publisher={De Gruyter}
}

@article{dudik2014,
author = {Miroslav Dud{\'i}k and Dumitru Erhan and John Langford and Lihong Li},
title = {{Doubly Robust Policy Evaluation and Optimization}},
volume = {29},
journal = {Statistical Science},
number = {4},
publisher = {Institute of Mathematical Statistics},
pages = {485 -- 511},
year = {2014},
doi = {10.1214/14-STS500},
URL = {https://doi.org/10.1214/14-STS500}
}

@article{van2007causal,
  title={Causal effect models for realistic individualized treatment and intention to treat rules},
  author={{van der Laan}, Mark J and Petersen, Maya L},
  journal={The international journal of biostatistics},
  volume={3},
  number={1},
  year={2007},
  publisher={De Gruyter}
}

@book{bickel1993efficient,
  title={Efficient and adaptive estimation for semiparametric models},
  author={Bickel, Peter J and Klaassen, Chris AJ and Bickel, Peter J and Ritov, Ya’acov and Klaassen, J and Wellner, Jon A and Ritov, YA'Acov},
  volume={4},
  year={1993},
  publisher={Springer}
}

@incollection{van2002semiparametric,
  title={Semiparametric statistics},
  author={{van der Vaart}, Aad W},
  booktitle={Lectures on probability theory and statistics (Saint-Flour, 1999)},
  pages={331--457},
  year={2002},
  publisher={Springer}
}

@book{tsiatis2006semiparametric,
  title={Semiparametric theory and missing data},
  author={Tsiatis, Anastasios A},
  volume={4},
  year={2006},
  publisher={Springer}
}

@book{kosorok2008introduction,
  title={Introduction to empirical processes and semiparametric inference},
  author={Kosorok, Michael R},
  volume={61},
  year={2008},
  publisher={Springer}
}

@article{Chernozhukov2018,
    author = {Chernozhukov, Victor and Chetverikov, Denis and Demirer, Mert and Duflo, Esther and Hansen, Christian and Newey, Whitney and Robins, James},
    title = "{Double/debiased machine learning for treatment and structural parameters}",
    journal = {The Econometrics Journal},
    volume = {21},
    number = {1},
    pages = {C1-C68},
    year = {2018},
    month = {01},
    issn = {1368-4221},
    doi = {10.1111/ectj.12097},
    url = {https://doi.org/10.1111/ectj.12097},
    eprint = {https://academic.oup.com/ectj/article-pdf/21/1/C1/27684918/ectj00c1.pdf},
}

@article{hansen2004nonparametric,
  title={Nonparametric conditional density estimation},
  author={Hansen, Bruce E},
  journal={Unpublished manuscript},
  volume={695},
  year={2004},
  publisher={Citeseer}
}

@article{hansen_uniform,
 ISSN = {02664666, 14694360},
 URL = {http://www.jstor.org/stable/20142515},
 author = {Bruce E. Hansen},
 journal = {Econometric Theory},
 number = {3},
 pages = {726--748},
 publisher = {Cambridge University Press},
 title = {Uniform Convergence Rates for Kernel Estimation with Dependent Data},
 urldate = {2025-12-13},
 volume = {24},
 year = {2008}
}

@book{laan2003unified,
  title={Unified methods for censored longitudinal data and causality},
  author={{van der Laan}, Mark J and Robins, James M},
  year={2003},
  publisher={Springer}
}

@book{van2000asymptotic,
  title={Asymptotic statistics},
  author={{van der Vaart}, Aad W},
  volume={3},
  year={2000},
  publisher={Cambridge university press}
}

@article{efron1981nonparametric,
  title={Nonparametric standard errors and confidence intervals},
  author={Efron, Bradley},
  journal={canadian Journal of Statistics},
  volume={9},
  number={2},
  pages={139--158},
  year={1981},
  publisher={Wiley Online Library}
}

@book{ellis2007entropy,
  title={Entropy, large deviations, and statistical mechanics},
  author={Ellis, Richard S},
  year={2007},
  publisher={Springer}
}

@article{field1982small,
  title={Small-sample asymptotic distributions of M-estimators of location},
  author={Field, Christopher A and Hampel, Frank R},
  journal={Biometrika},
  volume={69},
  number={1},
  pages={29--46},
  year={1982},
  publisher={Oxford University Press}
}

@article{qin1999empirical,
  title={Empirical likelihood ratio based confidence intervals for mixture proportions},
  author={Qin, Jing},
  journal={The Annals of Statistics},
  volume={27},
  number={4},
  pages={1368--1384},
  year={1999},
  publisher={Institute of Mathematical Statistics}
}

@article{efron1996using,
  title={Using specially designed exponential families for density estimation},
  author={Efron, Bradley and Tibshirani, Robert},
  journal={The Annals of Statistics},
  volume={24},
  number={6},
  pages={2431--2461},
  year={1996},
  publisher={Institute of Mathematical Statistics}
}

@article{maity2022understanding,
  title={Understanding new tasks through the lens of training data via exponential tilting},
  author={Maity, Subha and Yurochkin, Mikhail and Banerjee, Moulinath and Sun, Yuekai},
  journal={arXiv preprint arXiv:2205.13577},
  year={2022}
}

@article{scharfstein2021semiparametric,
  title={Semiparametric sensitivity analysis: Unmeasured confounding in observational studies},
  author={Scharfstein, Daniel O and Nabi, Razieh and Kennedy, Edward H and Huang, Ming-Yueh and Bonvini, Matteo and Smid, Marcela},
  journal={arXiv preprint arXiv:2104.08300},
  year={2021}
}

@incollection{robins2000sensitivity,
  title={Sensitivity analysis for selection bias and unmeasured confounding in missing data and causal inference models},
  author={Robins, James M and Rotnitzky, Andrea and Scharfstein, Daniel O},
  booktitle={Statistical models in epidemiology, the environment, and clinical trials},
  pages={1--94},
  year={2000},
  publisher={Springer}
}

@article{franks2020flexible,
  title={Flexible sensitivity analysis for observational studies without observable implications},
  author={Franks, AlexanderM and D’Amour, Alexander and Feller, Avi},
  journal={Journal of the American Statistical Association},
  year={2020},
  publisher={Taylor \& Francis}
}

@article{diaz2013targeted,
  title={Targeted data adaptive estimation of the causal dose--response curve},
  author={D{\'\i}az, Iv{\'a}n and van der Laan, Mark J},
  journal={Journal of Causal Inference},
  volume={1},
  number={2},
  pages={171--192},
  year={2013},
  publisher={De Gruyter}
}

@misc{zhou2022marginalinterventionaleffects,
      title={Marginal Interventional Effects}, 
      author={Xiang Zhou and Aleksei Opacic},
      year={2022},
      eprint={2206.10717},
      archivePrefix={arXiv},
      primaryClass={stat.ME},
      url={https://arxiv.org/abs/2206.10717}, 
}

@article{semenova2021debiased,
  title={Debiased machine learning of conditional average treatment effects and other causal functions},
  author={Semenova, Vira and Chernozhukov, Victor},
  journal={The Econometrics Journal},
  volume={24},
  number={2},
  pages={264--289},
  year={2021},
  publisher={Oxford University Press}
}

@article{branson2023causal,
  title={Causal effect estimation after propensity score trimming with continuous treatments},
  author={Branson, Zach and Kennedy, Edward H and Balakrishnan, Sivaraman and Wasserman, Larry},
  journal={arXiv preprint arXiv:2309.00706},
  year={2023}
}

@incollection{robins2008higher,
  title={Higher order influence functions and minimax estimation of nonlinear functionals},
  author={Robins, James and Li, Lingling and Tchetgen, Eric and van der Vaart, Aad and others},
  booktitle={Probability and statistics: essays in honor of David A. Freedman},
  volume={2},
  pages={335--422},
  year={2008},
  publisher={Institute of Mathematical Statistics}
}

@article{zheng2010asymptotic,
  title = {Asymptotic Theory for Cross-validated Targeted Maximum Likelihood Estimation},
  author = {Zheng, Wenjing and van der Laan, Mark J.},
  year = {2010},
  month = {November},
  institution = {U.C. Berkeley Division of Biostatistics Working Paper Series},
  type = {Working Paper},
  number = {273},
  url = {https://biostats.bepress.com/ucbbiostat/paper273}
}

@misc{rakshit2024localeffectscontinuousinstruments,
      title={Local Effects of Continuous Instruments without Positivity}, 
      author={Prabrisha Rakshit and Alexander Levis and Luke Keele},
      year={2024},
      eprint={2409.07350},
      archivePrefix={arXiv},
      primaryClass={stat.ME},
      url={https://arxiv.org/abs/2409.07350}, 
}

\begin{center}
{\large\bf SUPPLEMENTARY MATERIAL}
\end{center}

\begin{description}

\item \cref{reflected_ie_section}: Contains an extended discussion about the reflected incremental effect.
\item \cref{weak_positivity_section}: Contains an extended discussion about the weak positivity assumption.
\item \cref{simulation_section}: Contains additional simulations of coverage and interval width by $\delta$ and $n$.

\item \cref{sec:proofs}: Contains all proofs from the main text, including:
\begin{description}
    \item \cref{identified_eq_proof}: Proof of \cref{identified_eq}.
    \item \cref{eif_theorem_proof}: Proof of \cref{eif_theorem}.
    \item \cref{nonpar_efficiency_bound_proof}: Proof of \cref{nonpar_efficiency_bound}.
    \item \cref{sigma_bounds_proof}: Proof of \cref{sigma_bounds}.
    \item \cref{kl_divergence_lemma_proof}: Proof of \cref{kl_divergence_lemma}.
    \item \cref{minimax_lb_proof}: Proof of \cref{minimax_lb}.
    \item \cref{remainder_term_proof}: Proof of \cref{remainder_term}.
    \item \cref{l2_upper_bound_proof}: Proof of \cref{l2_upper_bound}.
    \item \cref{l2_lower_bound_proof}: Proof of \cref{l2_lower_bound}.
    \item \cref{remainder_bound_proof}: Proof of \cref{remainder_bound}.
    \item \cref{asymptotic_normality_proof}: Proof of \cref{asymptotic_normality}.
    \item \cref{reflected_eif_proof}: Proof of \cref{reflected_eif}.
    \item \cref{dose_response_theorem_proof}: Proof of \cref{dose_response_theorem}.
    \item \cref{sigma_delta_ratio_proof}: Proof of \cref{sigma_delta_ratio}.
    \item \cref{reflected_efficiency_bound_proof}: Proof of \cref{reflected_efficiency_bound}
\end{description}
\end{description}

\newpage

\section{Discussion of the reflected incremental effect} \label{reflected_ie_section}

In this section we provide a technical overview of the incremental effect under the reflected exponentially tilted intervention distribution. That is, for $a^\prime \in [0, 1]$ we consider
\begin{align*}
    \psi_R(\delta; a^\prime) =  \int_\mathcal{X} \int_\mathcal{A} \mu(x, a) r_\delta(a \mid x) \, da \, d\mathbb{P}(x)
\end{align*}
where we define
\begin{align*}
    r_\delta(a \mid x) =  \omega_{a^\prime}\frac{\text{exp}(\delta a) \pi(a \mid x) \mathbbm{1}(a \leq a^\prime) }{\int^{a^\prime}_0 \text{exp}(\delta t) \pi(t \mid x) dt} +  (1-\omega_{a^\prime})\frac{\text{exp}(-\delta a) \pi(a \mid x) \mathbbm{1}(a > a^\prime)}{\int^{1}_{a^\prime} \text{exp}(-\delta t) \pi(t \mid x) dt}
\end{align*}
such that $\omega_{a^\prime} = \int^{a^{\prime}}_0 \pi(t \mid x) dt$. Notably, the conditions required for estimating $\psi_R(\delta; a^\prime)$ are identical to those required to estimate $\psi(\delta)$, despite differences in the functional form of each estimator. As intuition, we first consider the variance of the efficient influence function $\varphi_R(Z;\delta, a^\prime)$, as defined in \cref{reflected_eif}. Importantly, $\var(\varphi_R(Z;\delta, a^\prime))$ is still subject to the same dependency on the magnitude of $\delta$, and consequently, the rate of convergence for our estimator of $\psi_R(\delta; a^\prime)$ must account for the mangtitude of the tilt. To formalize this point, in the following theorem we derive the nonparametric efficiency bound for the reflected incremental effect.

\begin{lemma} \label{reflected_efficiency_bound}
    The variance of the influence function $\varphi_R(Z;\delta, a^\prime )$, which is the nonparametric efficiency bound when $\delta \in [-M,M]$ for some $M<\infty$, is given by $\var(\varphi_R(Z;\delta, a^\prime)) = \mathbb{E}[D^2_Y + D^2_{r, \mu} + D^2_\psi]$ $D_Y$, $D_{r, \mu}$, and $D_\psi$ follow the same definitions given in \cref{reflected_eif}. Moreover, it follows that
    \begin{align*}
        \var(\varphi_R(Z;\delta, a^\prime)) \asymp \delta.
    \end{align*}
\end{lemma}

By \cref{reflected_efficiency_bound}, we can analogously see that the dominating terms in $\var(\varphi_R(Z;\delta, a^\prime))$ are determined by the ratio $r_\delta(A \mid X) / \pi(A \mid X)$, much like how $q_\delta(A \mid X) / \pi(A \mid X)$ dominates the variance of $\varphi(Z;\delta)$. Since $r_\delta(a \mid x)$ is defined by splitting $q_\delta(a \mid x)$ above and below $a^{\prime}$, it is clear that the same dependence on $\delta$ established in \cref{sigma_bounds} holds for the reflected incremental effect as well. Having established an analogous dependence on $\delta$ for $\var(\varphi_R(Z;\delta, a^\prime))$, we now consider estimation of $\psi_R(\delta; a^\prime)$.

In order to estimate the reflected incremental effect, we again use the one-step estimator, defined as
\begin{align*}
  \widehat\psi_R(\delta, a^{\prime}) = 
  \mathbb{P}_n \left[
      \frac{\widehat r_\delta(A \mid X)}{\widehat\pi(A \mid X)}
       Y +
        \mathbbm{1}(A \leq a^{\prime})\widehat L(X, A) + 
        \mathbbm{1}(A > a^{\prime})\widehat U(X, A)
    \right]
\end{align*}
where
\begin{align*}
  \widehat L(X, A)
  &= \widehat{\mathbb{E}}_R\big[\widehat\mu(X, A) \mid X, A \leq a^{\prime}\big] -
     \frac{\widehat r_\delta(A \mid X)}{\widehat\pi(A \mid X)}\,
     \widehat{\mathbb{E}}_R\big[\widehat\mu(X, A) \mid X, A \le a^{\prime} \big] \\
  \widehat U(X, A)
  &= \widehat{\mathbb{E}}_R\big[\widehat\mu(X, A) \mid X, A > a^{\prime}\big]  - \frac{\widehat r_\delta(A \mid X)}{\widehat\pi(A \mid X)}
     \widehat{\mathbb{E}}_R\big[\widehat\mu(X, A) \mid X, A > a^{\prime}\big]
\end{align*}
and
\begin{align*}
  \widehat{\mathbb{E}}_R\big[\widehat\mu(X, A) \mid X\big]
  &= \int \widehat\mu(X, a)\widehat r_\delta(a \mid X) da \\[0.4em]
  \widehat{\mathbb{E}}_R\big[\widehat\mu(X, A) \mid X, A \le a^{\prime}\big]
  &= \int_{a \le a^{\prime}} \widehat\mu(X, a) \left(
      \frac{\widehat r_\delta(a \mid X)}
           {\widehat{\mathbb{P}}(A \le a^{\prime} \mid X)} \right) da \\
  \widehat{\mathbb{E}}_R\big[\widehat\mu(X, A) \mid X, A > a^{\prime}\big]
  &= \int_{a > a^{\prime}} \widehat\mu(X, a) \left(
      \frac{\widehat r_\delta(a \mid X)}
           {\widehat{\mathbb{P}}(A > a^{\prime} \mid X)} \right) da.
\end{align*}
Notably, we can write this more compactly as
\begin{align*}
  m_R(X, A; a^{\prime})
  &= \mathbbm{1}(A \le a^{\prime})
        \mathbb{E}_R\big[\mu(X, A) \mid X, A \le a^{\prime}\big] +
      \mathbbm{1}(A > a^{\prime})
        \mathbb{E}_R\big[\mu(X, A) \mid X, A > a^{\prime}\big],
\end{align*}
which allows us to write the one-step estimator as
\begin{align*}
  \widehat\psi_R(\delta, a^{\prime})
  &= \mathbb{P}_n \left[
    \frac{\widehat r_\delta(A \mid X)}{\widehat\pi(A \mid X)}
      \Big(Y - \widehat m_R(X, A; a^{\prime})\Big)
    + \widehat m_R(X, A; a^{\prime})
  \right].
\end{align*}
As one would expect, taking $a^{\prime} \to 1$ or $a^{\prime} \to 0$ recovers the one-step estimator defined in \cref{proposed_estimator_section}. Furthermore, $\widehat{\psi}_R(\delta, a^{\prime})$ can be estimated following the same procedure described in \cref{cross_validated_algo}. Having now established our estimator for the reflected incremental effect, we explore how asymptotic normality of $\widehat{\psi}_R(\delta, a^{\prime})$ can be established under the same conditions of \cref{asymptotic_normality}.

To establish asymptotic normality for the reflected incremental effect, consider the standard decomposition of the reflected incremental effect, $\widehat{\psi}_R(\delta, a^\prime) - \psi_R(\delta, a^\prime)$, given by
\begin{align*}
     (\mathbb{P}_n - \mathbb{P})\{\varphi_R(Z; \mathbb{P})\} + (\mathbb{P}_n - \mathbb{P})\{\varphi_R(Z; \widehat{\mathbb{P}}) - \varphi_R(Z; \mathbb{P})\} + R_2(\widehat{\mathbb{P}}, \mathbb{P}) 
\end{align*}
where $R_2(\widehat{\mathbb{P}}, \mathbb{P}) = \psi_R(\widehat{\mathbb{P}}) - \psi_R(\mathbb{P}) + \int \phi_R(z; \widehat{\mathbb{P}}) d\mathbb{P}(z)$. Since $\var(\varphi_R(Z;\delta, a^\prime))$ also scales with $\delta$, it follows that Lindeberg's condition still holds under the same requirements as in \cref{asymptotic_normality} (although now with weak positivity holding in a neighborhood of $a^\prime$) and therefore the central limit theorem may be applied to the first term. For the empirical process term, we want to show that $||\varphi_R(Z; \widehat{\mathbb{P}}) - \varphi_R(Z; \mathbb{P})||_2 = o_\mathbb{P}(\sqrt{\delta})$. In order to show this condition holds, the proof is essentially the same as before. However, due to the $L(X, A)$ and $U(X, A)$ terms defined in \cref{reflected_eif}, we now technically require consistency above and below $a^\prime$, i.e. that
\begin{align*}
    \left( \int \underset{a \leq a^\prime}{\text{sup}} \, | \widehat{\pi} - \pi_a|^2 d\mathbb{P}(x) \right)^{1/2} = o_{\mathbb{P}}(1) \quad \text{ and } \quad  \left( \int \underset{a \leq a^\prime}{\text{sup}} \, | \widehat{\mu} - \mu_a|^2 d\mathbb{P}(x) \right)^{1/2} = o_{\mathbb{P}}(1)
\end{align*}
as well as 
\begin{align*}
    \left( \int \underset{a > a^\prime}{\text{sup}} \, | \widehat{\pi} - \pi_a|^2 d\mathbb{P}(x) \right)^{1/2} = o_{\mathbb{P}}(1) \quad \text{ and } \quad  \left( \int \underset{a > a^\prime}{\text{sup}} \, | \widehat{\mu} - \mu_a|^2 d\mathbb{P}(x) \right)^{1/2} = o_{\mathbb{P}}(1).
\end{align*}
However, these conditions are identical to those of \cref{asymptotic_normality}, since there we simply take the supremum across all $a \in \mathcal{A}$, which implies consistency both above and below $a^\prime$. Finally, following the proof of \cref{reflected_eif}, we know that $R_2(\widehat{\mathbb{P}}, \mathbb{P})$ can be decomposed into the following three terms:
\begin{align*}
     R_1 &= \mathbb{E} \left[ \int_{a}\left(\frac{\widehat{r}_\delta}{\widehat{\pi}_a} - \frac{r_\delta}{\pi_a} \right)\Big( (\mu_a - \widehat{\mu}_a)\pi_a + (\pi_a - \widehat{\pi}_a) \widehat{\mu}_a \Big) da \right] \\
    R_2 &= \mathbb{E}\left[\frac{\widehat{\mathbb{E}}_R\left(\mu(X, A) \mid X, A \leq a^{\prime} \right)}{\widehat{\mathbb{P}}(A \leq a^{\prime} \mid X)}\left( \int_{a \leq a^{\prime}} (\widehat{\pi}_a - \pi_a)\left( 1 - \frac{r_\delta}{\pi_a} \right)  da \right)\int_{a \leq a^{\prime}} \frac{\widehat{r}_\delta}{\widehat{\pi}_a} (\widehat{\pi}_a - \pi_a) da \right] \\
    R_3 &= \mathbb{E}\left[ \frac{\widehat{\mathbb{E}}_R\left(\mu(X, A) \mid X, A > a^{\prime} \right)}{\widehat{\mathbb{P}}(A > a^{\prime} \mid X)}\left( \int_{a > a^{\prime}} (\widehat{\pi}_a - \pi_a)\left( 1 - \frac{r_\delta}{\pi_a} \right)  da \right)\int_{a > a^{\prime}} \frac{\widehat{r}_\delta}{\widehat{\pi}_a} (\widehat{\pi}_a - \pi_a) da\right].
 \end{align*}
Notably, $R_1$ is integrated across the entire support of $a$, and therefore behaves exactly the same as its counterpart in \cref{remainder_term}. In order to control $R_2$ and $R_2$, we require that both
\begin{align*}
     \left( \int \underset{a \leq a^\prime}{\text{sup}} \, | \widehat{\pi} - \pi_a|^2 d\mathbb{P}(x) \right) = o_{\mathbb{P}}\left(1 / \sqrt{n / \delta}\right)
\end{align*}
and 
\begin{align*}
    \left( \int \underset{a > a^\prime}{\text{sup}} \, | \widehat{\pi} - \pi_a|^2 d\mathbb{P}(x) \right) = o_{\mathbb{P}}\left(1 / \sqrt{n / \delta}\right).
\end{align*}
However, similar to the case of the empirical process term, this condition is equivalent to simply assuming that $|| \widehat{\pi} - \pi||^2_{L^2_x, L^\infty_a} = o_{\mathbb{P}}(1 / \sqrt{n / \delta})$, which is exactly the condition required in \cref{asymptotic_normality}. Consequently, we can see that the conditions required for asymptotic normality to hold for the reflected incremental effect are identical to those required for the standard incremental effect, only now under the assumption that weak positivity holds in a neighborhood of $a^\prime$.

\section{Extended discussion of weak positivity} \label{weak_positivity_section}

In this section, we provide an extended discussion of why the ``weak positivity'' assumption is required for estimation inference of incremental effects (if one wants to avoid exponential penalties in $\delta$, that is). This result is perhaps a bit counterintuitive, since positivity itself is not required for identification, as shown in \cref{identified_eq_proof}; nevertheless, it is necessary. To formalize this discussion, recall that the weak positivity assumption is the assumption that there exists some $\eta > 0$ such that $\pi(a \mid x) \geq \pi_{\min} > 0$ for all $a \in [\eta, 1]$. Or, in the case of the reflected exponential tilt (as discussed in \cref{dose_response_section}), that  $\pi(a \mid x) \geq \pi_{\min} > 0$ for all $a \in [a^\prime - \eta, a^\prime + \eta]$. To proceed, we provide an informal discussion for the necessity of weak positivity, then follow up with a specific counterexample that illustrates why we cannot weaken the positivity requirement any further.

One way to understand the necessity of weak positivity is to think about the limiting behavior of $q_\delta(a \mid x)$ as $\delta \to \infty$. As discussed in \cref{incremental_effects_section}, as $\delta \to \infty$ then $q_\delta(a \mid x)$ approaches a point mass at the edge of the treatment support. This is important, because the nonparametric efficiency bound is driven primarily by the term
\begin{align*}
     \mathbb{E}\left[\frac{q_\delta(A \mid X)^2}{\pi(A \mid X)^2}\right].
\end{align*}
Intuitively, as $q_\delta(a \mid x)$ begins to concentrate near the edge of the support, if there is not sufficient concentration of $\pi(a \mid x)$ in this neighborhood then $\mathbb{E}\left[q_\delta(A \mid X)^2 / \pi(A \mid X)^2\right]$ can blow up as $\delta \to \infty$. Thus, weak positivity ensures that enough overlap exists such that the variance of our estimator doesn't increase too quickly in $\delta$. A natural follow-up question to this is, ``why is a uniform lower bound in this neighborhood necessary?'' In what follows, we show that without a uniform lower bound, then the nonparametric efficiency bound can scale superlinearly in $\delta$.

Now, let us consider what happens if $\pi(a \mid x)$ has some vanishing, but nonzero density in a neighborhood of one. To explore this setting, we consider the parametric family of distributions over the unit interval with densities given by
\begin{align*}
    p(x; \alpha, \beta) = \alpha \beta x^{a-1}(1-x^\alpha) \mathbbm{1}(x \in [0, 1])
\end{align*}
i.e. the Kumaraswamy distribution. Specifically, we let $\alpha = 1$ so that the conditional treatment density is defined as $\pi(a \mid X) = p(a; 1, \beta)$ for all $a \in [0, 1]$ and $\beta \geq 1$. Two important features of this distribution are that $p(a; 1, \beta)$ is bounded above by $\beta$ (so it satisfies the bounded properties invoked throughout the paper) and furthermore that the mass of $\pi(a; 1, \beta)$ is vanishing as $a \to 1$. By direct computation, it follows that
\begin{align*}
    \mathbb{E}\left[\exp(\delta A) \mid X \right] &= \beta \int^1_0 \exp(\delta a) (1 - a)^{\beta - 1} da \\
    &\overset{(i)}{=} \beta \int^1_0 \exp(\delta(1 - \kappa)) \kappa^{\beta - 1} d \kappa  \\
    &\overset{(ii)}{=} \beta \frac{\exp(\delta)}{\delta^\beta} \int^\delta_0 \exp(-u) u^{\beta - 1} du \\
    &=  \beta \frac{\exp(\delta)}{\delta^\beta} \gamma(\beta, \delta)
\end{align*}
where in $(i)$ we use the substitution $\kappa = 1 - a$ and in $(ii)$ the substitution $u = \delta \kappa$ and $\gamma(\cdot, \cdot)$ denotes the lower incomplete gamma function. Notably, $\gamma(\beta, \delta)$ is bounded above by $\Gamma(\beta)$ uniformly over $\delta$ and for $\delta \geq 1$, 
\begin{align*}
    \gamma(\beta, \delta) \geq \frac{1}{\exp(1)} \int^1_0 u^{\beta - 1} du = \frac{1}{\exp(1) \beta}.
\end{align*}
Consequently, over $\delta \in [1, \infty)$ it follows that $\mathbb{E}\left[\exp(\delta A) \mid X \right] \asymp \frac{\exp(\delta)}{\delta^\beta}$. This allows us to see that
\begin{align*}
      \mathbb{E}\left[\frac{q_\delta(A \mid X)^2}{\pi(A \mid X)^2}\right] = \mathbb{E}\left[ \frac{\exp(2 \delta A)}{\mathbb{E}\left[ \exp(\delta A) \mid X \right]^2} \right] \asymp \frac{\exp(2 \delta)}{(2 \delta)^\beta} \cdot \frac{\delta^{2 \beta}}{\exp(2 \delta)} \asymp \delta^\beta.
\end{align*}
At a high level, this shows that without some form of weak positivity requiring a sufficient of concentration of mass near the edge of the support then the variance of the influence function can grow arbitrarily polynomially fast in $\delta$.

We have now shown that without the weak positivity assumption, the nonparametric efficiency bound can scale at faster than linear rates in $\delta$. This is the key connection between why positivity is not required at all for identification, but weak positivity is required for inference. As discussed in the asymptotic normality result (see \cref{asymptotic_normality}), the exponentially tilted intervention distribution has an effective sample size of $n / \delta$, which is predicated on the linear scaling in $\delta$ of the variance of the efficient influence function. Further weakening of the weak positivity assumption could be theoretically possible, but it would lead to arbitrarily large reductions in the effective sample size based on $\delta$. Therefore, in order to maintain a practically useful estimator, we must assume at least weak positivity holds in order for there to be a sufficient concentration of mass near the edge of the support as $\delta \to \infty$.

\section{Simulations} \label{simulation_section}

In this section we provide simulations for the incremental effect that validate our theoretical findings. We also provide additional guidance for practical decisions, such as the range of $\delta$ to consider. All code can be found online on the associated github repository.

\begin{remark}
    An important decision to make when analyzing incremental effects is the range of $\delta$ to evaluate. Primarily, this should be viewed as a scientific question. That is to say, the domain expert should choose $\delta$ to answer a particular scientific question by considering both the increased log-odds of receiving a given treatment level, and also how far from $\pi(a \mid x)$ they think reasonable to deviate.  However, another way to construct a reasonable range is by considering the effective sample size, following work from the survey sampling literature, e.g. \cite{kish_65}. Let $w_i(\delta) = q_\delta(A_i \mid X_i) / \pi(A_i \mid X_i)$ and define
    \begin{align*}
        \text{ESS}(\delta) = \frac{\left( \sum^n_{i=1} w_i(\delta) \right)^2}{\sum^n_{i=1} w_i(\delta)^2} \approx \frac{n}{\mathbb{E}[w^2(\delta)]}.
    \end{align*}
    Since $\mathbb{E}[w^2(\delta)]$ governs the variance of our estimator, it is clear that $\text{ESS}(\delta)$ establishes a natural correspondence between $\delta$ and the width of our confidence intervals. In practice, one should choose some $\varsigma \in (0, 1)$ and $\delta$ such that $\text{ESS}(\delta) \geq \varsigma n$; this describes the amount of variance inflation one is willing to accept. For example, $\varsigma = 1/2$ implies $n / \text{ESS}(\delta) \leq 2$, and therefore our confidence intervals will be roughly will be 41\% wider (since $\sqrt{2} \approx 1.41$). 
\end{remark}

\begin{figure}[h]
\centering

\begin{tikzpicture}
\begin{groupplot}[
  group style={
    group size=2 by 1,
    horizontal sep=1.25cm,
  },
  width=0.48\textwidth,
  height=0.48\textwidth,
  xmode=log,
  ymin=0,
  log basis x=10,
  xtick={1,10,100,1000},
  xticklabels={1,10,100,1000},
  minor x tick num=8,         
  xmajorgrids=true,
  xminorgrids=true,
  ymajorgrids=true,
  grid style={gray!25},
  major grid style={gray!50},
  xlabel={$\delta$},
  ylabel={Coverage},
  legend cell align=left
]

\definecolor{npred}{HTML}{F8766D} 
\definecolor{reflblue}{HTML}{00BFC4} 
\definecolor{nicepurple}{HTML}{C77CFF}
\definecolor{nicegreen}{HTML}{7CAE00}

\nextgroupplot[
  title={$\alpha = 0.25$},
  legend style={
    at={(0.03,0.03)},        
    anchor=south west
  }
]

\addplot[thick, color=reflblue,
  mark=*,
  mark options={draw=reflblue, fill=reflblue}]
table[
  col sep=comma,
  x="delta",
  y="coverage",
  restrict expr to domain={\thisrow{"alpha"}}{0.24:0.26}, 
  restrict expr to domain={\thisrow{"n"}}{500:500}
]{files/sim_results_pgplots.csv};

\addplot[thick, color=npred,
  mark=square*,
  mark options={draw=npred, fill=npred}
]
table[
  col sep=comma,
  x="delta",
  y="coverage",
  restrict expr to domain={\thisrow{"alpha"}}{0.24:0.26},
  restrict expr to domain={\thisrow{"n"}}{1000:1000}
]{files/sim_results_pgplots.csv};

\addplot[thick, color=nicepurple,
  mark=diamond*,
  mark options={draw=nicepurple, fill=nicepurple}
]
table[
  col sep=comma,
  x="delta",
  y="coverage",
  restrict expr to domain={\thisrow{"alpha"}}{0.24:0.26},
  restrict expr to domain={\thisrow{"n"}}{2500:2500}
]{files/sim_results_pgplots.csv};

\addplot[thick, color=nicegreen,
  mark=triangle*,
  mark options={draw=nicegreen, fill=nicegreen}
]
table[
  col sep=comma,
  x="delta",
  y="coverage",
  restrict expr to domain={\thisrow{"alpha"}}{0.24:0.26},
  restrict expr to domain={\thisrow{"n"}}{5000:5000}
]{files/sim_results_pgplots.csv};

\addplot[densely dashed, domain=1:1000, samples=2] {0.95};

\legend{$n = 500$, $n = 1000$, $n = 2500$, $n=5000$}

\nextgroupplot[
  title={$\alpha = 0.5$},
  ylabel={} 
]

\addplot[thick, color=reflblue,
  mark=*,
  mark options={draw=reflblue, fill=reflblue}]
table[
  col sep=comma,
  x="delta",
  y="coverage",
  restrict expr to domain={\thisrow{"alpha"}}{0.49:0.51}, 
  restrict expr to domain={\thisrow{"n"}}{500:500}
]{files/sim_results_pgplots.csv};

\addplot[thick, color=npred,
  mark=square*,
  mark options={draw=npred, fill=npred}
]
table[
  col sep=comma,
  x="delta",
  y="coverage",
  restrict expr to domain={\thisrow{"alpha"}}{0.49:0.51},
  restrict expr to domain={\thisrow{"n"}}{1000:1000}
]{files/sim_results_pgplots.csv};

\addplot[thick, color=nicepurple,
  mark=diamond*,
  mark options={draw=nicepurple, fill=nicepurple}
]
table[
  col sep=comma,
  x="delta",
  y="coverage",
  restrict expr to domain={\thisrow{"alpha"}}{0.49:0.51},
  restrict expr to domain={\thisrow{"n"}}{2500:2500}
]{files/sim_results_pgplots.csv};

\addplot[thick, color=nicegreen,
  mark=triangle*,
  mark options={draw=nicegreen, fill=nicegreen}
]
table[
  col sep=comma,
  x="delta",
  y="coverage",
  restrict expr to domain={\thisrow{"alpha"}}{0.49:0.51},
  restrict expr to domain={\thisrow{"n"}}{5000:5000}
]{files/sim_results_pgplots.csv};

\addplot[densely dashed, domain=1:1000, samples=2] {0.95};

\end{groupplot}
\end{tikzpicture}

\caption{Empirical coverage across $\delta$ for $\alpha \in \{0.25, 0.5\}$ and $n \in \{500, 1000, 2500, 5000\}$}
\label{coverage_sim}
\end{figure}

Now that we have discussed heuristics for choosing $\delta$, we consider a simulation study to show how coverage and confidence interval width depend on $\delta$ and $n$. We generate independent samples such that $X_1,X_2 \sim \mathrm{Uniform}(0,1)$, $A \sim \mathrm{Uniform}(0,1)$ independent of $X$ so that $\pi(a \mid x) = 1$ on $[0,1]$, and $Y = \mu(X,A) + \varepsilon$ where $\mu(x,a) = 1 + x_1 + x_2 + 2a$ and $\varepsilon \sim N\big(b n^{-\alpha}, n^{-2\alpha}\big)$. The choice of the parameters $\alpha$ and $b$ follows the simulation setting of \cite{rubio2024populationsizeestimationlists}; $\alpha$ controls the rate of convergence (such that $\alpha = 0.5$ corresponds to parametric rates, and $\alpha = 0.25$ to nonparametric rates), and $b$ controls the bias. Under this design, it follows that $\psi_n(\delta) = 2 + 2\mathbb{E}_Q[A] + b n^{-\alpha}$ where 
\begin{align*}
    \mathbb{E}[A_\delta] = \frac{\delta \exp(\delta) - (\exp(\delta) - 1)}{\delta (\exp(\delta) - 1)}
\end{align*}
for $\delta \neq 0$, and $\mathbb{E}_Q[A] = 1/2$ when $\delta = 0$, since $A \sim q_\delta( \cdot \mid x)$ has density proportional to $\exp(\delta a)$ on $[0, 1]$. For each sample size $n \in \{500, 1000, 2500, 5000\}$, we consider $\alpha \in \{0.25, 0.5\}$, and $\delta \in \{0, 1, 2, 5, 10, 25, 50, 100, 250, 500, 1000 \}$. We then simulate 10,000 data sets, estimate $\mu(x,a)$ via (cross validated) random forests and compute the one-step estimator defined in \cref{estimation_inference_section} using the known treatment density. For each combination of $n$, $\alpha$, and $\delta$ we calculate the empirical coverage and average interval width.

\begin{figure}[h]
\centering

\begin{tikzpicture}
\begin{groupplot}[
  group style={
    group size=2 by 1,
    horizontal sep=1.25cm,
  },
  width=0.48\textwidth,
  height=0.48\textwidth,
  xmode=log,
  ymin=0,
  log basis x=10,
  xtick={1,10,100,1000},
  xticklabels={1,10,100,1000},
  minor x tick num=8,         
  xmajorgrids=true,
  xminorgrids=true,
  ymajorgrids=true,
  grid style={gray!25},
  major grid style={gray!50},
  xlabel={$\delta$},
  ylabel={Mean Interval Width},
  legend cell align=left
]

\definecolor{npred}{HTML}{F8766D} 
\definecolor{reflblue}{HTML}{00BFC4} 
\definecolor{nicepurple}{HTML}{C77CFF}
\definecolor{nicegreen}{HTML}{7CAE00}

\nextgroupplot[
  title={$\alpha = 0.25$},
  legend style={
    at={(0.03,0.97)},        
    anchor=north west
  }
]

\addplot[thick, color=reflblue,
  mark=*,
  mark options={draw=reflblue, fill=reflblue}]
table[
  col sep=comma,
  x="delta",
  y="mean_width",
  restrict expr to domain={\thisrow{"alpha"}}{0.24:0.26}, 
  restrict expr to domain={\thisrow{"n"}}{500:500}
]{files/sim_results_pgplots.csv};

\addplot[thick, color=npred,
  mark=square*,
  mark options={draw=npred, fill=npred}
]
table[
  col sep=comma,
  x="delta",
  y="mean_width",
  restrict expr to domain={\thisrow{"alpha"}}{0.24:0.26},
  restrict expr to domain={\thisrow{"n"}}{1000:1000}
]{files/sim_results_pgplots.csv};

\addplot[thick, color=nicepurple,
  mark=diamond*,
  mark options={draw=nicepurple, fill=nicepurple}
]
table[
  col sep=comma,
  x="delta",
  y="mean_width",
  restrict expr to domain={\thisrow{"alpha"}}{0.24:0.26},
  restrict expr to domain={\thisrow{"n"}}{2500:2500}
]{files/sim_results_pgplots.csv};

\addplot[thick, color=nicegreen,
  mark=triangle*,
  mark options={draw=nicegreen, fill=nicegreen}
]
table[
  col sep=comma,
  x="delta",
  y="mean_width",
  restrict expr to domain={\thisrow{"alpha"}}{0.24:0.26},
  restrict expr to domain={\thisrow{"n"}}{5000:5000}
]{files/sim_results_pgplots.csv};

\legend{$n = 500$, $n = 1000$, $n = 2500$, $n = 5000$}

\nextgroupplot[
  title={$\alpha = 0.5$},
  ylabel={} 
]

\addplot[thick, color=reflblue,
  mark=*,
  mark options={draw=reflblue, fill=reflblue}]
table[
  col sep=comma,
  x="delta",
  y="mean_width",
  restrict expr to domain={\thisrow{"alpha"}}{0.49:0.51}, 
  restrict expr to domain={\thisrow{"n"}}{500:500}
]{files/sim_results_pgplots.csv};

\addplot[thick, color=npred,
  mark=square*,
  mark options={draw=npred, fill=npred}
]
table[
  col sep=comma,
  x="delta",
  y="mean_width",
  restrict expr to domain={\thisrow{"alpha"}}{0.49:0.51},
  restrict expr to domain={\thisrow{"n"}}{1000:1000}
]{files/sim_results_pgplots.csv};

\addplot[thick, color=nicepurple,
  mark=diamond*,
  mark options={draw=nicepurple, fill=nicepurple}
]
table[
  col sep=comma,
  x="delta",
  y="mean_width",
  restrict expr to domain={\thisrow{"alpha"}}{0.49:0.51},
  restrict expr to domain={\thisrow{"n"}}{2500:2500}
]{files/sim_results_pgplots.csv};

\addplot[thick, color=nicegreen,
  mark=triangle*,
  mark options={draw=nicegreen, fill=nicegreen}
]
table[
  col sep=comma,
  x="delta",
  y="mean_width",
  restrict expr to domain={\thisrow{"alpha"}}{0.49:0.51},
  restrict expr to domain={\thisrow{"n"}}{5000:5000}
]{files/sim_results_pgplots.csv};

\end{groupplot}
\end{tikzpicture}

\caption{Average confidence interval width across $\delta$ for $\alpha \in \{0.25, 0.5\}$ and $n \in \{500, 1000, 2500, 5000\}$}
\label{width_sim}
\end{figure}

\cref{coverage_sim} verifies our theoretical results. It is clear to see that $1 - \alpha$ coverage holds for modest choices of $\delta$ relative to the sample size, but eventually decays when the ratio of $\delta$ to $n$ grows too large. However, for very large sample sizes coverage holds even up to extremely large values of $\delta$, e.g. $\delta = 1,000$. Meanwhile, in \cref{width_sim} we find simulated verification results such as \cref{sigma_bounds}. It is clear that mean confidence interval widths are growing with $\delta$, regardless of how large the sample size is. 

\newpage 

\section{Proofs}
\label{sec:proofs}
\subsection{Proof of \texorpdfstring{\cref{identified_eq}}{Equation 2}} \label{identified_eq_proof}

\begin{proof}[\textbf{Proof:}] Here, we provide a short proof for the identification of \cref{identified_eq}, as discussed in \cref{identification_section}. Recall that we assume: $(i)$ \textit{consistency} ($Y = Y^a$ if $A = a$), and $(ii)$ \textit{exchangeability} ($A \ind Y^a \mid X$ for $a \in \mathbb{R}$). Key to the identification proof is the fact that $\mathbb{E}[Y^{Q(\delta)}]$ can be written as a weighted average of the potential outcomes, with weights given by $q_\delta(a \mid x)$. To see this, we let $A_Q \mid X \sim q_\delta( \cdot \mid X)$ denote a draw from our stochastic intervention. Then, applying the law of iterated expectations,
\begin{align*}
    \mathbb{E}[Y^{Q(\delta)}] &= \mathbb{E}_X\left[ \mathbb{E}_{A_Q \mid X}\left[ \mathbb{E}_{Y(\cdot) \mid X, A_Q}\left[Y^{A_Q} \mid X, A_Q \right] \right] \right].
\end{align*}
From here, note that by construction, $A_Q \ind Y^a \mid X$ since  we sample $A_Q \sim q_\delta(\cdot \mid X)$ independently of the potential outcomes, conditional on $X$. Thus, for any fixed $a$, 
\begin{align*}
    \mathbb{E}_{Y(\cdot) \mid X, A_Q}\left[Y^{A_Q} \mid X, A_Q \right] =  \left. \mathbb{E}\left[Y^{a} \mid X\right] \, \right|_{a=A_Q}.
\end{align*}
This allows us to see that
\begin{align*}
    \mathbb{E}_{A_Q \mid X}\left[ \mathbb{E}_{Y(\cdot) \mid X, A_Q}\left[Y^{A_Q} \mid X, A_Q \right] \right] &= \mathbb{E}_{A_Q \mid X}\left[\left. \mathbb{E}\left[Y^{a} \mid X\right] \, \right|_{a=A_Q} \right] \\
    &= \int_\mathcal{A} \mathbb{E}[Y^a \mid X] q_\delta(a \mid X) da.
\end{align*}
Now, plugging this back into the original expectation yields
\begin{align*}
    \mathbb{E}[Y^{Q(\delta)}] &= \mathbb{E}_X\left[\int_\mathcal{A} \mathbb{E}[Y^a \mid X] q_\delta(a \mid X) da \right] \\
    &= \mathbb{E}_X\left[\int_\mathcal{A} \mathbb{E}[Y^a \mid X, A = a] q_\delta(a \mid X) da \right] \tag{By $(ii)$} \\
    &=  \mathbb{E}_X\left[\int_\mathcal{A} \mathbb{E}[Y \mid X, A = a] q_\delta(a \mid X) da \right] \tag{By $(i)$} \\ 
    &= \int_\mathcal{X} \int_\mathcal{A} \mu(x, a) q_\delta(a \mid x) \, da \, d\mathbb{P}(x)
\end{align*}
where $\mu(x, a) = \mathbb{E}[Y \mid X = x, A = a]$.

\end{proof}

\subsection{Proof of \texorpdfstring{\cref{eif_theorem}}{Proposition 1}} \label{eif_theorem_proof}
\begin{proof}[\textbf{Proof:}] We will derive the efficient influence function  for a general tilt,
\begin{align*}
    q_{\delta}(a \mid x) = \frac{f_\delta(\pi(a \mid x))}{\int_a f_\delta(\pi(a \mid x)) da }
\end{align*}
where $f(\cdot)$ is some smooth function and $\delta$ is assumed to be finite. To derive the efficient influence we will follow the ``derivative rules'' approach outlined in \cite{kennedy2023semiparametric} where we assume the data are discrete and use simple influence functions as building blocks. These assumptions greatly simplify the derivation process. While this process does not technically result in the efficient influence function (while the data is assumed to be discrete), we will then show the remainder term from the von Mises expansion is a second-order product of errors. This confirms that our derived influence function is valid in the general continuous case. Thus, suppose that the data are discrete, let $\mu(x, a) = \mathbb{E}[Y \mid X=x, A=a]$ and $p(x)$ be the probability mass function of $X$. Additionally, for notational convenience, let $f_\delta(\pi_a) = f_\delta(\pi(a \mid X))$. Then, applying the ``derivative rules" it follows that the influence function for $\psi(\delta)$ can be written as
\begin{align*}
    \mathbb{IF}(\psi(\delta)) &= \sum_{x \in \mathcal{X}} \sum_{a \in \mathcal{A}} \Bigg( \mathbb{IF}(\mu(x, a)) f_\delta(\pi(a \mid x)) \Big[\sum_t f_\delta(\pi(t \mid x)) \Big]^{-1} p(x)  \: + \tag{$i$}\\
    &\phantom{{}={\sum_{x \in \mathcal{X}} \sum_{a \in \mathcal{A}} \Bigg(}} \mu(x, a) \mathbb{IF}(f_\delta(\pi(a \mid x))) \Big[\sum_t f_\delta(\pi(t \mid x)) \Big]^{-1} p(x) \: + \tag{$ii$}\\
    &\phantom{{}={\sum_{x \in \mathcal{X}} \sum_{a \in \mathcal{A}} \Bigg(}} \mu(x, a) f_\delta(\pi(a \mid x))  \mathbb{IF}\Big(\Big[\sum_t f_\delta(\pi(t \mid x)) \Big]^{-1}\Big) p(x) \: + \tag{$iii$}\\
    &\phantom{{}={\sum_{x \in \mathcal{X}} \sum_{a \in \mathcal{A}} \Bigg(}} \mu(x, a) f_\delta(\pi(a \mid x)) \Big[\sum_t f_\delta(\pi(t \mid x)) \Big]^{-1} \mathbb{IF}(p(x)) \Bigg). \tag{$iv$}
\end{align*}
From here, we can use established derivations of influence functions as building blocks. Recall that
\begin{align*}
    \mathbb{IF}(\mu(x, a)) &= \frac{\mathbbm{1}(A = a, X=x)}{\pi(a \mid x) p(x)} \big( Y - \mu(x, a) \big) \\
    \mathbb{IF}(p(x)) &= \mathbbm{1}(X = x) - p(x) \\
    \mathbb{IF}(\pi(a \mid x)) &= \frac{\mathbbm{1}(X = x)}{p(x)} \big( \mathbbm{1}(A = a) - \pi(a \mid x) \big).
\end{align*}
Next, we derive the influence functions for $f_\delta(\pi(t \mid x))$ and $[\sum_t f_\delta(\pi(t \mid x))]^{-1}$. These are given by
\begin{align*}
    \mathbb{IF}(f_\delta(\pi(a \mid x))) &= f^{\prime}_\delta(\pi(a \mid x)) \frac{\mathbbm{1}(X = x)}{p(x)} \left( \mathbbm{1}(A = a) - \pi(a \mid x) \right) \\
    \mathbb{IF}\Big(\Big[\sum_a f_\delta(\pi(a \mid x)) \Big]^{-1}\Big) &= - \frac{\sum_a f^{\prime}_\delta(\pi(a \mid x))\frac{\mathbbm{1}(X = x)}{p(x)} \left( \mathbbm{1}(A = a) - \pi(a \mid x) \right) }{\left(\sum_a f_\delta(\pi(a \mid x)) \right)^2}.
\end{align*}
Finally, we plug these intermediate steps back into the influence function calculation. Observe that the first term is equal to
\begin{align*}
    (i)  &= \sum_{x \in \mathcal{X}} \sum_{a \in \mathcal{A}}  \frac{\mathbbm{1}(A = a, X=x)}{\pi(a \mid x)} \Big( Y - \mu(x, a) \Big)  \frac{f_\delta(\pi(a \mid x))}{\sum_t f_\delta(\pi(t \mid x))} \\
    &=  \frac{f_\delta(\pi(A \mid X))}{\pi(A \mid X) \sum_t f_\delta(\pi(t \mid X))} \Big( Y - \mu(X, A) \Big),
\end{align*}
the second is given by,
\begin{align*}
    (ii)  &= \sum_{x \in \mathcal{X}} \sum_{a \in \mathcal{A}} \frac{\mu(x, a) f^{\prime}_\delta(\pi(a \mid x))}{\sum_t f_\delta(\pi(t \mid x))} \mathbbm{1}(X = x) \left( \mathbbm{1}(A = a) - \pi(a \mid x) \right) \\
    &=  \frac{\mu(X, A) f^{\prime}_\delta(\pi(A \mid X))}{\sum_t f_\delta(\pi(t \mid X))}  -   \frac{\sum_{a \in \mathcal{A}} \mu(X, a) f^{\prime}_\delta(\pi(a \mid X)) \pi(a \mid X) }{\sum_t f_\delta(\pi(t \mid X))}   
\end{align*}
the third by
\begin{align*}
    (iii) &=  - \sum_{x \in \mathcal{X}} \sum_{a \in \mathcal{A}} \mu(x, a) f_\delta(\pi(a \mid x)) \frac{\sum_t f^{\prime}_\delta(\pi(t \mid x))\mathbbm{1}(X = x) \left( \mathbbm{1}(A = t) - \pi(t \mid x) \right) }{\left(\sum_t f_\delta(\pi(t \mid x)) \right)^2} \\
    &= \left(\frac{\sum_t f^{\prime}_\delta(\pi(t \mid X)) \pi(t \mid X) }{\left(\sum_t f_\delta(\pi(t \mid X)) \right)^2} - \frac{f^{\prime}_\delta(\pi(A \mid X))}{\left(\sum_t f_\delta(\pi(t \mid X)) \right)^2}\right)\sum_{a \in \mathcal{A}} \mu(X, a) f_\delta(\pi(a \mid X)),
\end{align*}
and finally the fourth term by
\begin{align*}
    (iv) &= \sum_{x \in \mathcal{X}} \sum_{a \in \mathcal{A}} \frac{\mu(x, a) f_\delta(\pi(a \mid x))}{\sum_t f_\delta(\pi(t \mid x))} (\mathbbm{1}(X = x) - p(x)) \\
    &= \sum_{a \in \mathcal{A}} \frac{\mu(X, a) f_\delta(\pi(a \mid X))}{\sum_t f_\delta(\pi(t \mid X))}  - \psi(\delta).
\end{align*}
Putting everything together, it follows that our candidate influence function is given by
\begin{align*}
    \mathbb{IF}(\psi(\delta)) &= \Bigg\{ \frac{f_\delta(\pi(A \mid X))}{\pi(A \mid X) \sum_a f_\delta(\pi(a \mid X))}\Big(Y - \mu(X, A) \Big) \: + \\
    &\phantom{{}={\Bigg\{}} \left( \frac{f^{\prime}_\delta(\pi(A \mid X)) \mu(X, A)}{\sum_t f_\delta(\pi(t \mid X))} - \frac{\sum_a f^{\prime}_\delta(\pi(a \mid X)) \pi(a \mid X) \mu(X, a)}{\sum_t f_\delta(\pi(t \mid X))} \right) \: - \\
    &\phantom{{}={\Bigg\{}}  \frac{\sum_a f_\delta( \pi( a \mid X)) \mu(X, a)}{\left[\sum_t f_\delta(\pi(t \mid X)) \right]^2} \left(f^\prime_\delta( \pi(A \mid X)) - \sum_t f^\prime_\delta( \pi(t \mid X)) \pi(t \mid X) \right) \: + \\
    &\phantom{{}={\Bigg\{}} \frac{\sum_a  \mu(X, a) f_\delta(\pi(a \mid X))}{\sum_t f_\delta(\pi(t \mid X))} - \psi(\delta) \Bigg\}.
\end{align*}
If we specifically use the exponential tilt $f_\delta( \pi(a \mid x)) = \exp(\delta a) \pi(a \mid x)$ and note that $f^\prime_\delta(\pi_a) = \exp(\delta a)$, we can see that the candidate influence function simplifies to
\begin{align*}
    \mathbb{IF}(\psi(\delta)) &= \Bigg\{ \frac{q_\delta(A \mid X)}{\pi(A \mid X)} \Big(Y - \mu(X, A) \Big) \: + \\
    &\phantom{{}={\Bigg\{}} \frac{q_\delta(A \mid X)}{\pi(A \mid X)}\left( \mu(X, A) - \sum_a q_\delta(a \mid X) \mu(X, a) \right) \: + \\
    &\phantom{{}={\Bigg\{}} \sum_a q_\delta(a \mid X) \mu(X, a)  - \psi(\delta) \Bigg\}.
\end{align*}

The remainder of the proof verifying that the von Mises expansion is a second-order product of errors is given in \cref{remainder_term_proof}. Note that in order for the derivation of \cref{eif_theorem} to be valid, it must be the case that $\delta$ is finite. Otherwise, as shown in \cref{sigma_bounds}, the variance of the efficient influence function will be infinite. However, as we demonstrate in \cref{estimation_inference_section}, it is still possible to obtain valid statistical inference with unbounded $\delta$.

\end{proof}

\subsection{Proof of \texorpdfstring{\cref{nonpar_efficiency_bound}}{Theorem 1}} \label{nonpar_efficiency_bound_proof}
\begin{proof}[\textbf{Proof:}]
    To derive the nonparametric efficiency bound, we must evaluate the variance of the efficient influence function, which we present in its compact form, i.e. 
    \begin{align*}
        \varphi(Z; \delta) = \frac{q_\delta(A \mid X)}{\pi(A \mid X)} \Big(Y - \mu_Q(X) \Big) + \mu_Q(X) - \psi(\delta).
    \end{align*}
where we define $\mu_Q(X) = \mathbb{E}_Q[ \mu(X, A) \mid X]$. Notably, since $\varphi(Z; \delta)$ is mean-zero, we simply need to evaluate $\mathbb{E}[\varphi(Z; \delta)^2]$. To that end, we define $U = \frac{q_\delta(A \mid X)}{\pi(A \mid X)} (Y - \mu_Q(X))$ and $V =  \mu_Q(X) - \psi(\delta)$. Next, we will show that the cross term $\mathbb{E}[UV] = 0$. To see this, first note that by the law of iterated expectations,
\begin{align*}
    \mathbb{E}[UV] &= \mathbb{E}\left[\left(\frac{q_\delta(A \mid X)}{\pi(A \mid X)} (Y - \mu_Q(X)) \right) \left(\mu_Q(X) - \psi(\delta) \right) \right] \\
    &= \mathbb{E}_X\left[ \mathbb{E}_{A \mid X}\left[ \mathbb{E}_{Y \mid X, A}\left[\left(\frac{q_\delta(A \mid X)}{\pi(A \mid X)} (Y - \mu_Q(X)) \right) \left(\mu_Q(X) - \psi(\delta) \right)  \right]\right] \right] \\
    &= \mathbb{E}_X\left[\left(\mu_Q(X) - \psi(\delta) \right)  \mathbb{E}_{A \mid X}\left[ \frac{q_\delta(A \mid X)}{\pi(A \mid X)}\mathbb{E}_{Y \mid X, A}\left[\left(Y - \mu_Q(X) \right)   \right]\right] \right] \\
    &= \mathbb{E}_X\left[\left(\mu_Q(X) - \psi(\delta) \right)  \mathbb{E}_{A \mid X}\left[ \frac{q_\delta(A \mid X)}{\pi(A \mid X)} \left( \mu(X, A) - \mu_Q(X) \right) \right] \right].
\end{align*}
Then, observe that
\begin{align*}
    \mathbb{E}_{A \mid X}\left[ \frac{q_\delta(A \mid X)}{\pi(A \mid X)} \left( \mu(X, A) - \mu_Q(X) \right) \right] = \int_a \left(\mu(X, a) - \mu_Q(X) \right) q_\delta(a \mid X) da = 0,
\end{align*}
which thereby implies that $\mathbb{E}[UV] = 0$. We can further simplify our variance calculation by noting that 
\begin{align*}
    \mathbb{E}[V^2] &= \mathbb{E}\left[\Big( \mathbb{E}_Q\left[\mu(X, A) \mid X \right]  - \psi(\delta) \Big)^2 \right] \\
    &= \mathbb{E}\left[\Big( \mathbb{E}_Q\left[\mu(X, A) \mid X \right]  - \int_x \int_a \mu(x, a) q_\delta(a \mid x) \, da \, d\mathbb{P}(x) \Big)^2 \right] \\
    &= \mathbb{E}\left[\Big( \mathbb{E}_Q\left[\mu(X, A) \mid X \right]  - \mathbb{E}\Big[ \mathbb{E}_Q \big[ \mu(X, A) \mid X\big] \Big] \Big)^2 \right] \\
    &= \var\Big( \mathbb{E}_Q\left[\mu(X, A) \mid X \right] \Big).
\end{align*}
Consequently, it follows that
\begin{align*}
    \var \left( \varphi(Z; \delta) \right) = \mathbb{E}\left[\left(\frac{q_\delta(A \mid X)}{\pi(A \mid X)} \right)^2 (Y - \mathbb{E}_Q\left[\mu(X, A) \mid X \right])^2 \right] + \var\Big( \mathbb{E}_Q\left[\mu(X, A) \mid X \right] \Big).
\end{align*}
For some later results it will be convenient to isolate individual components of $\var \left( \varphi(Z; \delta) \right)$, so we note that after applying the identity $Y - \mu_Q(X) = Y - \mu(X, A) + \mu(X, A) - \mu_Q(X)$, then we can further decompose $\var \left( \varphi(Z; \delta) \right)$ into the three term expansion,
 \begin{align*}
     \mathbb{E}\left[\left(\frac{q_\delta(A \mid X)}{\pi(A \mid X)} \right)^2 \Big( \var\left( Y \mid X, A \right) +  \big[ \mu(X, A) - \mu_Q(X) \big]^2 \Big) \right] + \var\big( \mu_Q(X) \big).
\end{align*}
\end{proof}

\subsection{Proof of \texorpdfstring{\cref{sigma_bounds}}{Lemma 1}} \label{sigma_bounds_proof}

\begin{proof}[\textbf{Proof:}] 

Recall that the three term expression for the nonparametric efficiency bound $\sigma^2_\delta$ is given by
\begin{align*}
     \mathbb{E}\left[\left(\frac{q_\delta(A \mid X)}{\pi(A \mid X)} \right)^2 \Big( \var\left( Y \mid X, A \right) +  \big[ \mu(X, A) - \mu_Q(X) \big]^2 \Big) \right] + \var\big( \mu_Q(X) \big)
\end{align*}
In what follows, we derive lower and upper bounds for $\sigma^2_\delta$. Note that we first assume that $\delta > 0$. Later in the proof, we discuss the setting where $\delta < 0$.

\noindent \textbf{Lower Bound:} First, we derive the lower bound. Using the three term decomposition of $\sigma^2_\delta$, we can immediately see that
 \begin{align*}
     \sigma^2_\delta \geq \mathbb{E}\left[\left(\frac{q_\delta(A \mid X)}{\pi(A \mid X)}\right)^2 \var\left( Y \mid X, A \right)   \right]
 \end{align*}
 since $\var( \mu_Q(X)) \geq 0$ and $(\mu(x, a) - \mu_Q(x))^2 \geq 0$ for all $a$ and $x$. Now, under the assumption that $\sigma^2_{\min} \leq \var\left( Y \mid X, A \right)$ all that remains is to lower bound the expected squared likelihood ratio. To do so, suppose that there exists some $\eta \in[0, 1)$ and $\pi_{\min} > 0$ such that $\pi(a \mid x) \geq \pi_{\min}$ for all $a \in [\eta, 1]$. Further assume that $\pi(a \mid x) \leq \pi_{\max} < \infty$ for all $a$ and $x$ and $\delta > 0$. Then, it follows that
\begin{align*}
    \mathbb{E}\left[\left(\frac{q_\delta(A \mid X)}{\pi(A \mid X)}\right)^2 \right] &= \mathbb{E}\left[\frac{\int^1_0 \exp(2 \delta a) \pi(a \mid X) da}{\left(\int^1_0 \exp(\delta a) \pi(a \mid X) da \right)^2} \right] \geq \frac{\pi_{\min}}{\pi^2_{\max}} \frac{\int^1_\eta \exp(2 \delta a)da}{\left(\int^1_0 \exp(\delta a) da \right)^2}. 
\end{align*}
Next, after integrating both terms we have
\begin{align*}
   \frac{\int^1_\eta \exp(2 \delta a)da}{\left(\int^1_0 \exp(\delta a) da \right)^2} = \frac{\delta}{2} \cdot \frac{1 - \exp(- 2 \delta(1 - \eta))}{(1 - \exp(- \delta))^2}.
\end{align*}
From here, observe that for any $u > 0$ and $v \geq 0$ it follows that
\begin{align*}
    1 - \exp(-uv) \geq (1 - \exp(-u))(1 - \exp(-v)).
\end{align*}
Thus, letting $u = 2(1 - \eta)$ and $v = \delta$ we can see that
\begin{align*}
    \frac{\delta}{2} \cdot \frac{1 - \exp(- 2 \delta(1 - \eta))}{(1 - \exp(- \delta))^2} &\geq \frac{\delta}{2} \cdot \frac{(1 - \exp(-2(1- \eta)))(1 - \exp(-\delta))}{(1 - \exp(- \delta))^2} \\
    &\geq \frac{\delta}{2} \cdot \frac{1 - \exp(-2(1- \eta))}{1 - \exp(- \delta)} \\
    &\geq  \delta \left( \frac{1 - \exp(-2(1- \eta))}{2} \right)
\end{align*}
where the last inequality follows since $1 - \exp(-\delta) \leq 1$ for all $\delta > 0$. Putting everything together, it follows that
\begin{align*}
    \sigma^2_\delta \geq \delta \left( \frac{\pi_{\min} \sigma^2_{\min}}{\pi^2_{\max}} \cdot  \frac{1 - \exp(-2(1- \eta))}{2} \right) \geq \delta \left( \frac{\pi_{\min} \sigma^2_{\min}}{\pi^2_{\max}} \cdot \frac{2}{5}(1 - \eta) \right).
\end{align*}
In the case that $\delta < 0$, we can follow analogous steps to obtain a lower bound that is linear in $|\delta|$. In this case, we now must assume there exists some $\eta \in (0, 1]$ where positivity holds. Then, we have that
\begin{align*}
     \mathbb{E}\left[\left(\frac{q_\delta(A \mid X)}{\pi(A \mid X)}\right)^2 \right] &\geq \frac{\pi_{\min}}{\pi^2_{\max}} \frac{\int^\eta_0 \exp(2 \delta a)da}{\left(\int^1_0 \exp(\delta a) da \right)^2}. 
\end{align*}
Following the same steps as before yields the bound
\begin{align*}
     \sigma^2_\delta \geq |\delta| \left( \frac{\pi_{\min} \sigma^2_{\min}}{\pi^2_{\max}} \cdot  \frac{1 - \exp(-2\eta )}{2} \right).
\end{align*}
\noindent \textbf{Upper Bound:} Next, for the upper bound observe that
\begin{align*}
    \var\left( Y \mid X, A \right) = \mathbb{E}\left[ Y^2 \mid X, A \right] - \mathbb{E}\left[ Y \mid X, A \right]^2 \leq \mathbb{E}\left[ Y^2 \mid X, A \right] \leq B^2.
\end{align*}
Following the same logic, since $|\mu(X, A)| = |\mathbb{E}[Y \mid X, A]|\leq B$, we have that
\begin{align*}
    \Big| \mathbb{E}_Q(\mu(X, A) \mid X) \Big| = \left|\int_a \mu(X, a) q_\delta(a \mid X)da \right| \leq  B \left| \int_a q_\delta(a \mid X)da \right| = B.
\end{align*}
Applying these calculations again, it is clear that
\begin{align*}
    \var\Big( \mathbb{E}_Q\left[\mu(X, A) \mid X \right] \Big) &\leq \mathbb{E}\left[ \mathbb{E}_Q\left[\mu(X, A) \mid X \right]^2 \right]  \leq B^2
\end{align*}
and
\begin{align*}
    \Big( \mu(X, A) - \mathbb{E}_Q(\mu(X, A) \mid X) \Big)^2 \leq  \Big(\left|\mu(X, A)\right| + \left|\mathbb{E}_Q(\mu(X, A) \mid X)\right|  \Big)^2 \leq 4B^2.
\end{align*}
Thus, we arrive at an initial upper bound of.
\begin{align*}
    \sigma^2_\delta &\leq \mathbb{E}\left[\left(\frac{q_\delta(A \mid X)}{\pi(A \mid X)}\right)^2 \Big( B^2 + 4B^2 \Big) \right] + B^2.
\end{align*}
Now, all that remains is to upper bound the expected squared likelihood ratio. Following similar arguments to the lower bound, we can see that
\begin{align*}
    \mathbb{E}\left[\left(\frac{q_\delta(A \mid X)}{\pi(A \mid X)}\right)^2 \right] &= \mathbb{E}\left[\frac{\int^1_0 \exp(2 \delta a) \pi(a \mid X) da}{\left(\int^1_0 \exp(\delta a) \pi(a \mid X) da \right)^2} \right] \\
    &\leq \frac{\pi_{\max}}{\pi^2_{\min}}\frac{\int^1_0 \exp(2 \delta a) da}{\left(\int^1_\eta \exp(\delta a) da \right)^2} \\
    &\leq \frac{\pi_{\max}}{\pi^2_{\min}} \cdot \frac{\delta}{2} \cdot \frac{\exp(2 \delta) - 1}{\left(\exp(\delta) - \exp(\delta \eta) \right)^2}.
\end{align*}
From here let $\ell = \delta(1 - \eta)$. Then, we can see that
\begin{align*}
    \exp(2 \delta) - 1 &= \exp(2 \delta \eta)( \exp(2 \ell) - 1) + ( \exp(2 \delta \eta) - 1) \quad \text{and} \\
    \left(\exp(\delta) - \exp(\delta \eta) \right)^2 &= \exp(2 \delta \eta) (\exp(\ell) - 1)^2,
\end{align*}
which allows us show that
\begin{align*}
    \frac{\exp(2 \delta) - 1}{\left(\exp(\delta) - \exp(\delta \eta) \right)^2} &= \frac{\exp(2 \ell) - 1}{(\exp(\ell) - 1)^2} + \frac{1 - \exp(- 2 \delta \eta )}{(\exp(\ell) - 1)^2} \\
    &= \text{coth}\left( \frac{\ell}{2} \right) +  \frac{1 - \exp(- 2 \delta \eta )}{(\exp(\ell) - 1)^2} \\
    &\overset{(i)}{\leq} \text{coth}\left( \frac{\ell}{2} \right) + \frac{2 \delta \eta}{\ell^2} \\
    &= \text{coth}\left( \frac{\ell}{2} \right) + \frac{2 \eta}{\delta(1 - \eta)^2}
\end{align*}
where $(i)$ follows since $1 - \exp(-2\delta \eta) \leq 2 \delta \eta$ and $\exp(-\ell) - 1 \geq \ell$. Finally, note that since  $x \cdot \text{coth}(cx) \leq 1 / c + x$ for any scalar $c$, it follows that
\begin{align*}
    \frac{\delta}{2}\text{coth}\left( \frac{\ell}{2} \right) =  \frac{\delta}{2}\text{coth}\left( \frac{\delta(1 - \eta)}{2} \right) \leq \frac{\delta}{2} + \frac{1}{1 - \eta}.
\end{align*}
Putting everything together, we arrive at a final upper bound of
\begin{align*}
    \sigma^2_\delta &\leq  \frac{\pi_{\max}}{\pi^2_{\min}}\left(\frac{\delta}{2} +  \frac{1}{(1 - \eta)^2}\right) \Big( B^2 + 4B^2 \Big) + B^2 \\
    &\leq B^2 \left[1 + \frac{5}{2} \cdot \frac{\pi_{\max}}{\pi^2_{\min}} \left(\delta + \frac{2}{1 - \eta} \right) \right].
\end{align*}
Again, in the setting where $\delta < 0$ we can repeat this same analysis by starting with the inequality
\begin{align*}
     \mathbb{E}\left[\left(\frac{q_\delta(A \mid X)}{\pi(A \mid X)}\right)^2 \right] &\leq \frac{\pi_{\max}}{\pi^2_{\min}}\frac{\int^1_0 \exp(2 \delta a) da}{\left(\int^\eta_0 \exp(\delta a) da \right)^2}.
\end{align*}
which (after repeating calculations) yields an analogous upper bound in terms of $|\delta|$.
\end{proof}

\subsection{Proof of \texorpdfstring{\cref{kl_divergence_lemma}}{Proposition 2}} \label{kl_divergence_lemma_proof}

\begin{proof}[\textbf{Proof:}]
    Suppose $\delta > 0$ and observe that
    \begin{align*}
        D_{\text{KL}}\Big(q_\delta(a \mid X) \: || \: \pi(a \mid X) \Big) &= \int_a q_\delta(a \mid X) \text{log}\left(\frac{q_\delta(a \mid X)}{\pi(a \mid X)} \right) da \\
        &= \int_a q_\delta(a \mid X) \text{log}\left(\frac{\exp(\delta a)}{\int_t \exp( \delta t) \pi(t \mid X) dt}  \right) da \\
        &= \int_a q_\delta(a \mid X) \left( \delta a - \text{log}\left(\int_t \exp( \delta t) \pi(t \mid X) dt \right)  \right) da \\
        &= \delta \int_a a q_\delta(a \mid X)da -  \text{log}\left(\int_t \exp( \delta t) \pi(t \mid X) dt \right) \int_a q_\delta(a \mid X) da \\
        &= \delta \int_a a q_\delta(a \mid X)da -  \text{log}\left(\int_t \exp( \delta t) \pi(t \mid X) dt \right),
    \end{align*}
    where the last step follows since $\int_a q_\delta(a \mid X) da = 1$. Note that we can think of $\int_a a q_\delta(a \mid X)da$ as $\mathbb{E}_Q[A \mid X]$ for some $A \sim q_{\delta}( \cdot \mid X)$. From here, we will differentiate with respect to $\delta$. Now, under the assumption that $\pi(a \mid X)$ is continuous, it follows that
    \begin{align*}
        \frac{\partial }{\partial \delta} \left\{ \text{log}\left(\int_t \exp(\delta t) \pi(t \mid X) dt \right) \right\} &= \frac{1}{\int_t \exp(\delta t) \pi(t \mid X) dt}\left( \frac{\partial }{\partial \delta} \int_t \exp(\delta t) \pi(t \mid X) dt\right) \\
        &= \frac{1}{\int_t \exp(\delta t) \pi(t \mid X) dt}\left(  \int_t \frac{\partial }{\partial \delta}  \exp(\delta t) \pi(t \mid X) dt \right) \\
        &= \frac{1}{\int_t \exp(\delta t) \pi(t \mid X) dt} \left(\int_t t \exp(\delta t) \pi(t \mid X) dt \right) \\
        &= \int_a a q_\delta(a \mid X)  da
    \end{align*}
    and
    \begin{align*}
        \frac{\partial}{\partial \delta} \left\{ \delta \int_a a q_\delta(a \mid X)da \right\} &= \int_a a q_\delta(a \mid X)da + \delta \left( \frac{\partial}{\partial \delta} \int_a a q_\delta(a \mid X)da \right) \\
        &= \int_a a q_\delta(a \mid X)da + \delta \left(  \int_a \frac{\partial}{\partial \delta} a q_\delta(a \mid X)da \right). 
    \end{align*}
    Finally, to complete the proof, observe that
    \begin{align*}
        \frac{\partial}{\partial \delta} \left\{ a q_\delta(a \mid X) \right\} &= \frac{\partial}{\partial \delta} \left\{ \frac{a \exp(\delta a) \pi(a \mid X)}{\int_t \exp(\delta t) \pi(t \mid X) dt} \right\} \\
        &= \frac{a^2 \exp(\delta a) \pi(a \mid X)}{\int_t \exp(\delta t) \pi(t \mid X) dt} - \frac{\left(a \exp(\delta a) \pi(a \mid X)\right)\left(\int_t t \exp(\delta t) \pi(t \mid X) dt \right)}{\left(\int_t \exp(\delta t) \pi(t \mid X) dt \right)^2},
    \end{align*}
    so it follows that $D^\prime_{KL}(q_\delta, \pi_a) := \frac{\partial}{\partial \delta} \left\{ D_{\text{KL}}\Big(q_\delta(a \mid X) \: || \: \pi(a \mid X) \Big) \right\}$ is given by
    \begin{align*}
        D^\prime_{KL}(q_\delta, \pi_a) &= \delta \left(\frac{\int_a a^2 \exp( \delta a) \pi(a \mid X) da}{\int_t \exp(\delta t) \pi(t \mid X) dt} - \left(\frac{\left(\int_a a \exp(\delta a) \pi(a \mid X) da\right)}{\left(\int_t \exp(\delta t) \pi(t \mid X) dt \right)} \right)^2 \right) \\
        &= \delta \left( \int_a a^2 q_\delta(a \mid X) da - \left( \int_a a q_\delta(a \mid X) da \right)^2 \right)
    \end{align*}
    which is equal to $\delta \var_Q(A \mid X)$. Now, we evaluate the second derivative with respect to $\delta$. Here, it follows that
    \begin{align*}
        \frac{\partial}{\partial \delta} \left\{ \delta \var_Q(A \mid X) \right\} = \var_Q(A \mid X) + \delta \left(\frac{\partial}{\partial \delta}\var_Q(A \mid X) \right).
    \end{align*}
    Following similar steps from before, we can see that $\frac{\partial}{\partial \delta} \left\{ \int_a a^2 q_\delta(a \mid X) da \right\}$ is given by
    \begin{align*}
         \frac{\int_a a^3 \exp(\delta a) \pi(a \mid X) da }{\int_t \exp(\delta t) \pi(t \mid X) dt} - \frac{\left(\int_a a^2 \exp(\delta a) \pi(a \mid X) da \right)\left(\int_t t \exp(\delta t) \pi(t \mid X) dt \right)}{\left(\int_t \exp(\delta t) \pi(t \mid X) dt \right)^2}
    \end{align*}
    which reduces to $\mathbb{E}_Q[A^3 \mid X] - \mathbb{E}_Q[A^2 \mid X] \mathbb{E}_Q[A \mid X]$, and
    \begin{align*}
        \frac{\partial}{\partial \delta} \left\{ \left( \int_a a q_\delta(a \mid X) da \right)^2  \right\} &= 2 \int_a a q_\delta(a \mid X) da \left(\int_a a^2 q_\delta(a \mid X) da - \left( \int_a a q_\delta(a \mid X) da \right)^2 \right) \\
        &= 2 \mathbb{E}_Q[A \mid X] \Big(\mathbb{E}_Q[A^2 \mid X] - \mathbb{E}_Q[A \mid X]^2 \Big).
    \end{align*}
    Then, using the fact that
    \begin{align*}
        \mathbb{E}_Q[A^3 \mid X] - 3\mathbb{E}_Q[A^2 \mid X] \mathbb{E}_Q[A 
        \mid X]  + 2 \mathbb{E}_Q[A \mid X]^3 = \mathbb{E}_Q\left[(A - \mathbb{E}_Q[A \mid X])^3 \mid X \right]
    \end{align*}
    it follows that
    \begin{align*}
        \frac{\partial^2}{\partial \delta^2}  \left\{ D_{\text{KL}}\Big(q_\delta(a \mid X) \: || \: \pi(a \mid X) \Big) \right\} = \var_Q(A \mid X) + \delta \mathbb{E}_Q\left[(A - \mathbb{E}_Q[A \mid X])^3 \mid X \right].
    \end{align*}
    If we instead were interested about the Kullback-Leibler divergence between $\pi(a \mid X)$ and $q_\delta(a \mid X)$, we can see that
    \begin{align*}
         D_{\text{KL}}\Big(\pi(a \mid X) \: || \: q_\delta(a \mid X) \Big) &= \int_a \pi(a \mid X) \text{log}\left(\frac{\pi(a \mid X)}{q_\delta(a \mid X)} \right) da \\
         &= \int_a \pi(a \mid X) \left(\text{log}\left(\int_t \exp(\delta t) \pi(t \mid X) dt \right) - \delta a \right) da \\
         &= \text{log}\left(\int_t \exp(\delta t) \pi(t \mid X) dt \right)   - \delta \int_a a \pi(a \mid X) da,
    \end{align*}
    which yields a first and second derivative of
    \begin{align*}
        \frac{\partial}{\partial \delta} \left\{ D_{\text{KL}}\Big(\pi(a \mid X) \: || \: q_\delta(a \mid X) \Big) \right\} &= \mathbb{E}_Q[A \mid X] - \mathbb{E}[A \mid X] \\
        \frac{\partial^2}{\partial \delta^2} \left\{ D_{\text{KL}}\Big(\pi(a \mid X) \: || \: q_\delta(a \mid X) \Big) \right\} &= \var_Q(A \mid X).
    \end{align*}
\end{proof}

\subsection{Proof of \texorpdfstring{\cref{minimax_lb}}{Theorem 2}} \label{minimax_lb_proof}
\begin{proof}[\textbf{Proof:}]
Let $P_0$ and $P_1$ denote distributions in $\mathcal{P}$ such that $P_0 = \otimes^n_{i=1} P_{0i}$ and $P_1 = \otimes^n_{i=1} P_{1i}$. Then, recall from  \cite{tsbakov2009} that if
\begin{align*}
    H^2(P_0, P_1) \leq \alpha < 2
\end{align*}
and $\psi(p_{0i}) - \psi(p_{1i}) \geq s > 0$ for a functional $\psi : \mathcal{P} \to \mathbb{R}$ for all $i = 1, \ldots, n$ then,
\begin{align*}
    \underset{\widehat{\psi}}{\text{inf}} \: \underset{P \in \mathcal{P}}{\text{sup}} \mathbb{E}_P \left[\ell\left(\widehat{\psi} - \psi(P) \right) \right] \geq \ell(s/2)\left(\frac{1 - \sqrt{\alpha(1 - \alpha/4)}}{2} \right)
\end{align*}
for any monotonic non-negative loss function $\ell$. Thus, for some $\varepsilon > 0$, consider the following null density and fluctuated alternative:
\begin{align*}
    p_0(z) &= p(y \mid x, a) p(a \mid x) p(x) \\
    p_1(z) &= \Big[p(y \mid x, a)(1 + \varepsilon \phi_y(z; p)) \Big]p(a \mid x) p(x)
\end{align*}
where we define
\begin{align*}
    \phi_y(z; p) = \frac{q_\delta(a \mid x)}{p(a \mid x)}(y - \mu(x, a)).
\end{align*}
Additionally, let 
\begin{align*}
    \phi_{\psi}(z; p) &= \frac{q_\delta(a \mid x)}{p(a \mid x)}\left( \mu(x, a) - \int q_\delta(a \mid x) \mu(x, a)da \right) \\
    \phi_{q, \mu}(z; p) &= \int \mu(x, a) q(a \mid x) da - \psi 
\end{align*}
such that $\varphi(z; p) = \phi_y(z; p) + \phi_{q, \mu}(z; p) + \phi_{\psi}(z ;p)$. With these definitions in place, we will first evaluate the functional separation. Applying the von Mises expansion, we can see that
\begin{align*}
    \psi(p_0) - \psi(p_1) - R_2(p_0, p_1) = - \int \varphi(z; p_0) p_1(z)dz.
\end{align*}
Then, plugging in our definitions of $p_1(z)$ we find
\begin{align*}
    - \int \varphi(z; p_0) p_1(z)dz &=  - \int \varphi(z; p_0) \Big[p(y \mid x, a)(1 + \varepsilon \phi_y(z; p)) \Big]p(a \mid x) p(x) dz \\
    &=  - \int \varphi(z; p) p(z) dz - \varepsilon \int \Big(\varphi(z; p_0) \phi_y(z; p)\Big) p(y \mid x, a) p(a \mid x) p(x) dz.
\end{align*}
From here, we can see that the first integral is zero, since $\varphi(\cdot)$ is mean-zero. Next, by expanding out $\varphi(z; p_0) \phi_y(z; p)$ it follows that
\begin{align*}
      - \varepsilon \int \Big(\varphi(z; p_0) \phi_y(z; p)\Big) p(y \mid x, a) p(a \mid x) p(x) dz &=  - \varepsilon \int  \phi_y(z; p)^2 p(z) dz,
\end{align*}
as each of the cross-terms are zero. That is,
\begin{align*}
    \int \phi_{y}(z; p) \phi_{q, \mu} p(z) dz &= 0 \\
    \int \phi_{y}(z; p) \phi_{\psi} p(z) dz &= 0
\end{align*}
as shown in the proof of \cref{nonpar_efficiency_bound}. Thus, we are left with
\begin{align*}
    \psi(p_0) - \psi(p_1) - R_2(p_0, p_1) = -\varepsilon \int  \phi_y(z; p)^2 p(z) dz.
\end{align*}
Furthermore, it can be shown that the remainder term $R_2(p_0, p_1)$ under the null and fluctuated alternative densities is also equal to zero. Recall that by \cref{remainder_term} it follows that $R_2(p_0, p_1) = R_1 - R_2$ where
\begin{align*}
    R_1 &= \mathbb{E} \left[\int_a \left(\frac{q_0(a \mid X)}{p_0(a \mid X)} - \frac{q_1(a \mid X)}{p_1(a \mid X)} \right)\Big((p_1(a \mid X)\mu_1(X, a) - p_0(a \mid x)\mu_0(X, a) \Big)  da \right] \\
    R_2 &= \mathbb{E} \left[\left(\int_a \frac{q_1(a \mid X)}{p_1(a \mid X)} \mu_0(X, a) p_0(a \mid X) da \right)\left(\int_a \frac{q_0(a \mid X)}{p_0(a \mid X)} (p_1(a \mid X) - p_0(a \mid X)) da \right)^2\right],
\end{align*}
which evaluates to zero, because $p_0(a \mid X) = p_1(a \mid X)$ under the null and fluctuated densities. This simplifies our functional separation to
\begin{align*}
    |\psi(p_0) - \psi(p_1)|  = \varepsilon \left| \int  \phi_y(z; p)^2 p(z) dz \right|.
\end{align*}
Thus, all that remains is to bound the integral of $\phi_y(z;p)^2$. By \cref{sigma_bounds}, it follows that 
\begin{align*}
   |\psi(p_0) - \psi(p_1)| \geq \varepsilon  \left( \frac{2}{5} \cdot \frac{p_{\min} \sigma^2_{\min}}{p^2_{\max}} \right) \delta 
\end{align*}
where we assume that $p_{\min} \leq p(a \mid x) \leq p_{\max}$ for all $a,x$ and $\sigma^2_{\min} \leq \var(Y \mid X, A)$. Note that this corresponds to the setting under \cref{sigma_bounds} in which $\eta = 0$. This satisfies the functional separation requirement. Next, we evaluate the $\chi^2$-distance between the null density and fluctuated alternative. Observe that
\begin{align*}
    \chi^2(p_1, p_0) &= \int \left( \frac{p_1(z)}{p_0(z)} \right)^2 p_0(z)dz - 1 \\
    &= \int  \frac{p^2_1(z)}{p_0(z)}  - 1 \\
    &= \int  \frac{\left(\big[p(y \mid x, a)(1 + \varepsilon \phi_y(z; p)) \big]p(a \mid x) p(x)\right)^2}{p_0(z)}dz  - 1 \\
    &= \int  \frac{\left(p_0(z) + \varepsilon \phi_y(z; p) p_0(z)\right)^2}{p_0(z)} dz  - 1 \\
    &= \int  \left(p_0(z) + 2\varepsilon \phi_y(z; p) p_0(z) dz + \varepsilon^2 \phi_y(z; p)^2 p_0(z)\right)  - 1
\end{align*}
Then, using the fact that $\int p_0(z) dz = 1$ and
\begin{align*}
    \int \phi_y(z; p) p_0(z) dz &= \int \frac{q_\delta(a \mid x)}{p(a \mid x)}\Big(y - \mu(x, a)\Big) p_0(z) dz \\
    &= \int \frac{q_\delta(a \mid x)}{p(a \mid x)}\Big(\mu(x, a) - \mu(x, a)\Big) p(y \mid x, a) p(a \mid x) p(x) dz = 0
\end{align*}
it follows that
\begin{align*}
    \chi^2(p_1, p_0) = \varepsilon^2 \int \phi_y(z; p)^2 p_0(z) dz.
\end{align*}
Now, we need to upper bound $\int \phi_y(z; p)^2 p_0(z) dz$. Again, by \cref{sigma_bounds} it follows that
\begin{align*}
    \chi^2(p_1, p_0) \leq \varepsilon^2 B^2 \left[1 + \frac{5}{2} \cdot \frac{p_{\max}}{p^2_{\min}} \left(2 + \delta \right) \right]
\end{align*}
under the additional assumption that $|Y|\leq B$ with probability one. To complete the proof, we make use of the fact that the $\chi^2$ distance upper bounds the Hellinger distance in the sense that for distributions $P_0$ and $P_1$, $H^2(P_0, P_1) \leq \sqrt{\chi^2(P_0, P_1)}$. Therefore, if we show that $\sqrt{\chi^2(P_0, P_1)} < \alpha$, for some $\alpha < 2$, it must also hold for $H^2(P_0, P_1)$. From here, recall that for product measures $P_0 = \otimes^n_{i=1} p_{0i}$ and $P_1 = \otimes^n_{i=1} p_{1i}$, then
\begin{align*}
    \chi^2(P_0, P_1) = \prod^n_{i=1} \Big(1 + \chi^2(p_{0i}, p_{1i})\Big) - 1,
\end{align*}
so we must show that $\prod^n_{i=1} \Big(1 + \chi^2(p_{0i}, p_{1i})\Big) < \alpha^2 + 1$. After a log transform
\begin{align*}
   \text{log}\left( \prod^n_{i=1} \Big(1 + \chi^2(p_{0i}, p_{1i})\Big) \right) &= \sum^n_{i=1} \text{log}\Big(1 + \chi^2(p_{0i}, p_{1i})\Big) \\
   &\leq n \text{log}\left(1 + \varepsilon^2  B^2 \left(1 + \frac{5 p_{\max}}{2 p_{\min}^2} \left(2 + \delta  \right) \right)\right) \\
   &\leq n  \varepsilon^2  B^2 \left(1 + \frac{5 p_{\max}}{2 p_{\min}^2} \left(2 + \delta  \right) \right).
\end{align*}
using $\log(1+x) \leq x$ in the last line. 
Thus, it suffices to show that
\begin{align*}
    n  \varepsilon^2  B^2 \left(1 + \frac{5 p_{\max}}{2 p_{\min}^2} \left(2 + \delta  \right) \right) < \text{log}(\alpha^2 + 1)
\end{align*}
for some $\alpha < 2$, which clearly holds when
\begin{align*}
    \varepsilon^2 < \frac{\text{log}(\alpha^2 + 1)}{n  B^2 \left( 1 + \frac{5 p_{\max}}{2 p_{min}^2} (2+\delta) \right)}.
\end{align*}
If we assume $\delta \geq 2$, then it follows that
\begin{align*}
    \frac{\text{log}(\alpha^2 + 1)}{n  B^2 \left( 1 + \frac{5 p_{\max}}{2 p_{min}^2} (2+\delta) \right)} &\geq \frac{\text{log}(\alpha^2 + 1)}{n  B^2 \left( 1 + \frac{5 p_{\max}}{p_{min}^2} \delta  \right)} \geq \frac{\text{log}(\alpha^2 + 1)}{n  B^2 \left( 1 + \frac{5 p^2_{\max}}{ p_{min}^2} \delta  \right)} \geq \frac{\text{log}(\alpha^2 + 1)}{ 10 n  B^2 ( p^2_{\max} / p_{min}^2)\delta }
\end{align*}
where the last inequality follows since $p^2_{\max} / p^2_{\min} \geq 1$. Thus, setting
\begin{align*}
    \varepsilon^2 =\frac{\text{log}(\alpha^2 + 1)}{ 10 n  B^2 ( p^2_{\max} / p_{\min}^2 )\delta }
\end{align*}
ensures that $H^2(P_0, P_1) \leq \alpha < 2$. Putting everything together, since
\begin{align*}
    s &= \varepsilon \left( \frac{2}{5} \cdot  \frac{p_{\min} \sigma^2_{\min}}{p_{\max}^2} \right)  \delta \\
    &= \sqrt{\frac{\text{log}(\alpha^2 + 1)}{ 10 n  B^2 ( p^2_{\max} / p_{min}^2 )\delta }} \left( \frac{2}{5} \cdot \frac{p_{\min} \sigma^2_{\min}}{p_{\max}^2} \right) \delta  \\
    &= \sqrt{\frac{\delta}{n}} \left(\frac{\sqrt{10 \cdot \text{log}(\alpha^2 + 1)}p^2_{\min} \sigma^2_{\min}}{25 B p^3_{\max}} \right),
\end{align*}
we can see that 
\begin{align*}
    \underset{\widehat{\psi}}{\text{inf}} \ \underset{P \in \mathcal{P}}{\text{sup}} \ \mathbb{E}_P \left[\ell\left(\widehat{\psi} - \psi(\delta) \right) \right] &\geq \ell\left(\sqrt{\frac{\delta}{n}} \left(\frac{\sqrt{10 \cdot \text{log}(\alpha^2 + 1)}p^2_{\min} \sigma^2_{\min}}{50 B p^3_{\max}} \right) \right)\left(\frac{1 - \sqrt{\alpha(1 - \frac{\alpha}{4})}}{2} \right).
\end{align*}
From here, if we are interested in the root mean squared error as our loss function, it follows that
\begin{align*}
    \underset{\widehat{\psi}}{\text{inf}} \ \underset{P \in \mathcal{P}}{\text{sup}} \ \mathbb{E}_P \left|\widehat{\psi} - \psi(\delta) \right| &\geq \sqrt{\frac{C}{n/\delta}}
\end{align*}
where
\begin{align*}
    C = \left(\frac{\sqrt{10 \cdot \text{log}(\alpha^2 + 1)}p^2_{\min} \sigma^2_{\min}}{50 B p^3_{\max}} \right)\left(\frac{1 - \sqrt{\alpha(1 - \alpha/4)}}{2} \right).
\end{align*}
Finally, note that although we have considered the set of models $\mathcal{P}$ lower bounded densities, we can immediately extend this result to the set of models that only have upper bounded densities, $\mathcal{P}^\prime$. This follows since
\begin{align*}
    \underset{\widehat{\psi}}{\text{inf}} \underset{P \in \left\{ \mathcal{P} \cup \mathcal{P}^\prime \right\} }{\text{sup}} \mathbb{E}_P \left|\widehat{\psi} - \psi(\delta) \right| \geq \underset{\widehat{\psi}}{\text{inf}} \ \underset{P \in \mathcal{P}}{\text{sup}} \ \mathbb{E}_P \left|\widehat{\psi} - \psi(\delta) \right| &\geq \sqrt{\frac{C}{n/\delta}}.
\end{align*}

\end{proof}

\subsection{Proof of \texorpdfstring{\cref{remainder_term}}{Proposition 3}} \label{remainder_term_proof}
\begin{proof}[\textbf{Proof:}]
Recall that for two distributions, $\mathbb{P}$ and $\widehat{\mathbb{P}}$, the remainder term in the von Mises expansion is given by
\begin{align*}
    R_2(\widehat{\mathbb{P}}, \mathbb{P}) &= \psi(\widehat{\mathbb{P}}) - \psi(\mathbb{P}) - \int \varphi (z; \widehat{\mathbb{P}}) d(\widehat{\mathbb{P}} - \mathbb{P})(z) \\
    &= \psi(\widehat{\mathbb{P}}) - \psi(\mathbb{P}) + \int \varphi (z; \widehat{\mathbb{P}}) d\mathbb{P}(z) - \int \varphi (z; \widehat{\mathbb{P}}) d\widehat{\mathbb{P}}(z) \\
    &= \psi(\widehat{\mathbb{P}}) - \psi(\mathbb{P}) + \int \varphi (z; \widehat{\mathbb{P}}) d\mathbb{P}(z) \tag{Since $\varphi$ is mean-zero}\\
    &= \cancel{\psi(\widehat{\mathbb{P}})} - \psi(\mathbb{P}) + \int\phi(z; \widehat{\mathbb{P}}) d\mathbb{P}(z) - \cancel{\psi(\widehat{\mathbb{P}})} \\
    &= \mathbb{E}_\mathbb{P}\left[\phi(Z; \widehat{\mathbb{P}}) - \psi(\mathbb{P})\right]
\end{align*}
where $\phi(Z; \mathbb{P}) = \varphi(Z; \mathbb{P}) + \psi(\mathbb{P})$. For notational convenience, we use the shorthand script $\pi_a = \pi(a \mid X)$, $\mu_a = \mu(X, a)$, and $f_\delta(\pi_a) = f_\delta( \pi(a \mid X))$. Then, leveraging our results from \cref{eif_theorem} it follows that
\begin{align*}
    R_2(\widehat{\mathbb{P}}, \mathbb{P}) &= \mathbb{E}_\mathbb{P}\Bigg[ \frac{f_\delta(\widehat{\pi}(A \mid X))}{\widehat{\pi}(A \mid X) \sum_t f_\delta(\widehat{\pi}_t)}\Big(Y - \widehat{\mu}(X, A) \Big) \: + \tag{$i$}\\
    &\phantom{{}={\mathbb{E}_\mathbb{P}\Bigg[}} \left( \frac{f^{\prime}_\delta(\widehat{\pi}(A \mid X)) \widehat{\mu}(X, A)}{\sum_t f_\delta(\widehat{\pi}_t)} - \frac{\sum_t f^{\prime}_\delta(\widehat{\pi}_t) \widehat{\pi}_t \widehat{\mu}_t}{\sum_t f_\delta(\widehat{\pi}_t)} \right) \: - \tag{$ii$} \\
    &\phantom{{}={\mathbb{E}_\mathbb{P}\Bigg[}}  \frac{\sum_t f_\delta( \widehat{\pi}_t) \widehat{\mu}_t}{\left[\sum_t f_\delta(\widehat{\pi}_t) \right]^2} \left(f^\prime_\delta( \widehat{\pi}(A \mid X)) - \sum_t f^\prime_\delta( \widehat{\pi}_t) \widehat{\pi}_t \right) \: + \tag{$iii$}\\
    &\phantom{{}={\mathbb{E}_\mathbb{P}\Bigg[}} \frac{\sum_t  \mu_t f_\delta(\widehat{\pi}_t )}{\sum_t f_\delta(\widehat{\pi}_t)} - \psi(\mathbb{P}) \Bigg]. \tag{$iv$}
\end{align*}
From here, we simplify the remainder by applying the law of iterated expectations. Observe that
\begin{align*}
    (i) &= \mathbb{E}\left[ \mathbb{E}\left[ \frac{f_\delta(\widehat{\pi}(A \mid X))}{\widehat{\pi}(A \mid X) \sum_t f_\delta(\widehat{\pi}_t)}\Big(Y - \widehat{\mu}(X, A) \Big) \mid X, A \right]\right] \\
    &= \mathbb{E}\left[  \frac{f_\delta(\widehat{\pi}(A \mid X))}{\widehat{\pi}(A \mid X) \sum_t f_\delta(\widehat{\pi}_t)}\Big(\mu(X, A) - \widehat{\mu}(X, A) \Big) \right] \\
    &= \mathbb{E}\left[ \mathbb{E}\left[  \frac{f_\delta(\widehat{\pi}(A \mid X))}{\widehat{\pi}(A \mid X) \sum_t f_\delta(\widehat{\pi}_t)}\Big(\mu(X, A) - \widehat{\mu}(X, A) \Big) \mid X \right] \right] \\
    &= \mathbb{E}\left[ \sum_a   \frac{f_\delta(\widehat{\pi}_a) \pi_a}{\widehat{\pi}_a \sum_t f_\delta(\widehat{\pi}_t)}\Big(\mu_a - \widehat{\mu}_a \Big)  \right]. 
\end{align*}
Similarly, the other terms in the remainder can be expressed as
\begin{align*}
    (ii) &= \mathbb{E}\left[\left( \frac{f^{\prime}_\delta(\widehat{\pi}(A \mid X)) \widehat{\mu}(X, A)}{\sum_t f_\delta(\widehat{\pi}_t)} - \frac{\sum_t f^{\prime}_\delta(\widehat{\pi}_t) \widehat{\pi}_t \widehat{\mu}_t}{\sum_t f_\delta(\widehat{\pi}_t)} \right) \right] \\
    &= \mathbb{E}\left[ \frac{1}{\sum_t f_\delta(\widehat{\pi}_t)}\sum_a \left( f^{\prime}_\delta(\widehat{\pi}_a) \widehat{\mu}_a - \sum_t f^{\prime}_\delta(\widehat{\pi}_t) \widehat{\pi}_t \widehat{\mu}_t \right)\pi_a  \right] \\
    &= \mathbb{E}\left[ \frac{1}{\sum_t f_\delta(\widehat{\pi}_t)}\sum_a  f^{\prime}_\delta(\widehat{\pi}_a) \widehat{\mu}_a \pi_a - \frac{1}{\sum_t f_\delta(\widehat{\pi}_t)}\sum_t f^{\prime}_\delta(\widehat{\pi}_t) \widehat{\pi}_t \widehat{\mu}_t  \right] \\
    &= \mathbb{E}\left[ \frac{1}{\sum_t f_\delta(\widehat{\pi}_t)}\sum_a  f^{\prime}_\delta(\widehat{\pi}_a) \widehat{\mu}_a (\pi_a -  \widehat{\pi}_a)  \right] 
\end{align*}
and
\begin{align*}
    (iii) &= \mathbb{E}\left[\frac{\sum_t f_\delta( \widehat{\pi}_t) \widehat{\mu}_t}{\left[\sum_t f_\delta(\widehat{\pi}_t) \right]^2} \left(f^\prime_\delta( \widehat{\pi}(A \mid X)) - \sum_t f^\prime_\delta( \widehat{\pi}_t) \widehat{\pi}_t \right) \right] \\
    &= \mathbb{E} \left[\frac{\sum_t f_\delta( \widehat{\pi}_t) \widehat{\mu}_t}{\left[\sum_t f_\delta(\widehat{\pi}_t) \right]^2} \left( \sum_a f^\prime_\delta( \widehat{\pi}_a) \pi_a  - \sum_t f^\prime_\delta( \widehat{\pi}_t) \widehat{\pi}_t \right) \right] \\
    &= \mathbb{E} \left[\frac{\sum_t f_\delta( \widehat{\pi}_t) \widehat{\mu}_t}{\left[\sum_t f_\delta(\widehat{\pi}_t) \right]^2} \left( \sum_a f^\prime_\delta( \widehat{\pi}_a) (\pi_a - \widehat{\pi}_a) \right) \right]
\end{align*}
and finally
\begin{align*}
    (iv) &= \mathbb{E}\left[\frac{\sum_t  \mu_t f_\delta(\widehat{\pi}_t )}{\sum_t f_\delta(\widehat{\pi}_t)} - \psi(\mathbb{P}) \right].
\end{align*}
Then, putting everything together, it follows that
\begin{align*}
    R_2(\widehat{\mathbb{P}}, \mathbb{P}) &= \mathbb{E}\Bigg[ \frac{1}{\sum_t f_\delta(\widehat{\pi}_t)} \sum_a   \frac{f_\delta(\widehat{\pi}_a) \pi_a}{\widehat{\pi}_a }\Big(\mu_a - \widehat{\mu}_a \Big) - \frac{\sum_t f_\delta( \widehat{\pi}_t) \widehat{\mu}_t}{\left[\sum_t f_\delta(\widehat{\pi}_t) \right]^2} \left( \sum_a f^\prime_\delta( \widehat{\pi}_a) (\pi_a - \widehat{\pi}_a) \right) \: + \\
    &\phantom{{}={\mathbb{E}\Bigg[}}  \frac{1}{\sum_t f_\delta(\widehat{\pi}_t)}\sum_a  f^{\prime}_\delta(\widehat{\pi}_a) \widehat{\mu}_a (\pi_a -  \widehat{\pi}_a) + \frac{\sum_t  \mu_t f_\delta(\widehat{\pi}_t )}{\sum_t f_\delta(\widehat{\pi}_t)} - \psi(\mathbb{P}) \Bigg].
\end{align*}
From here, we will proceed with a Taylor series expansion argument. Observe that
\begin{align*}
    f_\delta(\pi_a) - f_\delta(\widehat{\pi}_a) = f^\prime (\widehat{\pi}_a) \big(\pi_a - \widehat{\pi}_a \big) + \underbrace{f^{\prime \prime} (\pi^*_a) \big(\pi_a - \widehat{\pi}_a \big)^2}_{r(a, X)}
\end{align*}
where $\pi^*_a$ is between $\pi_a$ and $\widehat{\pi}_a$. Substituting in $r(a, X)$ and adding and subtracting the term $f_\delta(\pi_a)$, we can see that
\begin{align*}
    R_2(\widehat{\mathbb{P}}, \mathbb{P}) &= \mathbb{E}\Bigg[ \frac{1}{\sum_t f_\delta(\widehat{\pi}_t)} \sum_a \left(\frac{f_\delta(\widehat{\pi}_a) \pi_a}{\widehat{\pi}_a} + \textcolor{blue}{f_\delta(\pi_a)} - \textcolor{red}{f_\delta(\pi_a)} \right)\big(\mu_a - \widehat{\mu}_a \big) \: - \\
    &\phantom{{}={\mathbb{E} \Bigg[}}  \frac{\sum_a f_\delta( \widehat{\pi}_a) \widehat{\mu}_a}{\left[\sum_t f_\delta(\widehat{\pi}_t) \right]^2}  \sum_a \Big(  f_\delta(\pi_a) - f_\delta(\widehat{\pi}_a) - r(a, X) \Big) \: + \\
    &\phantom{{}={\mathbb{E} \Bigg[ }} \frac{1}{\sum_t f_\delta(\widehat{\pi}_t)} \sum_a \Big(  f_\delta(\pi_a) - f_\delta(\widehat{\pi}_a) - r(a, X) \Big)\widehat{\mu}_a   \: + \\
    &\phantom{{}={\mathbb{E} \Bigg[ }} \frac{1}{\sum_t f_\delta(\widehat{\pi}_t)} \sum_a f_\delta(\widehat{\pi}_a) \widehat{\mu}_a  - \psi \Bigg].
\end{align*}
From here, expanding out this expression and making cancellations reduces the remainder term to
\begin{align*}
     R_2(\widehat{\mathbb{P}}, \mathbb{P}) &=  \mathbb{E} \Bigg[ \frac{1}{\sum_t f_\delta(\widehat{\pi}_t)}\sum_a \left(\frac{f_\delta(\widehat{\pi}_a) \pi_a}{\widehat{\pi}_a} - f_\delta(\pi_a) \right)\big(\mu_a - \widehat{\mu}_a \big) \: - \\
    &\phantom{{}={\mathbb{E} \Bigg[}}  \frac{\sum_a f_\delta( \widehat{\pi}_a) \widehat{\mu}_a}{\left[\sum_a f_\delta(\widehat{\pi}_t) \right]^2}  \sum_t \Big(  f_\delta(\pi_a) - f_\delta(\widehat{\pi}_a) \Big) +   \sum_a \frac{f_\delta( \widehat{\pi}_a) \widehat{\mu}_a}{\left[\sum_t f_\delta(\widehat{\pi}_t) \right]^2}  \sum_a r(a, X)   \: - \\
    &\phantom{{}={\mathbb{E} \Bigg[ }} \frac{1}{\sum_t f_\delta(\widehat{\pi}_t)} \sum_a  r(a, X) \widehat{\mu}_a  + \sum_a f_\delta(\pi_a ) \mu_a \left(\frac{1}{\sum_t f_\delta(\widehat{\pi}_t)} - \frac{1}{\sum_t f_\delta(\pi_t )} \right) \Bigg].
\end{align*}
Finally, to algebraically manipulate our expression into second-order terms, let $u = \sum_a f_\delta (\pi_a) \mu_a$ and $v = \sum_a f_\delta(\widehat{\pi}_a)$. Then, it is clear that
\begin{align*}
    u\left( \frac{1}{\widehat{v}} - \frac{1}{v}\right) - \widehat{u}\left( \frac{v - \widehat{v}}{\widehat{v}^2}\right) &= \Bigg\{ \sum_a f_\delta(\pi_a ) \mu_a \left(\frac{1}{\sum_t f_\delta(\widehat{\pi}_t)} - \frac{1}{\sum_t f_\delta(\pi_t )} \right) \: - \\
    &\phantom{{}={\Bigg\{}} \frac{\sum_a f_\delta( \widehat{\pi}_a) \widehat{\mu}_a}{\left[\sum_t f_\delta(\widehat{\pi}_t) \right]^2}  \sum_t \Big(  f_\delta(\pi_a) - f_\delta(\widehat{\pi}_a) \Big) \Bigg\}.
\end{align*}
Furthermore, we can show that this expression can be written in terms of products of differences between distributions by noting that
\begin{align*}
    u\left( \frac{1}{\widehat{v}} - \frac{1}{v}\right) - \widehat{u}\left( \frac{v - \widehat{v}}{\widehat{v}^2}\right) &= (u - \textcolor{red}{\widehat{u}} + \textcolor{blue}{\widehat{u}}) \left( \frac{1}{\widehat{v}} - \frac{1}{v}\right) - \widehat{u}\left( \frac{v - \widehat{v}}{\widehat{v}^2}\right) \\
    &= (u - \widehat{u})\left( \frac{1}{\widehat{v}} - \frac{1}{v}\right) - \widehat{u}\left(\frac{(\widehat{v} - v)^2}{v \widehat{v}^2} \right).
\end{align*}
This allows us to obtain the final expression,
\begin{align*}
     R_2(\widehat{\mathbb{P}}, \mathbb{P}) &= \mathbb{E} \Bigg[ \frac{1}{\sum_t f_\delta(\widehat{\pi}_t)} \sum_a \left(\frac{f_\delta(\widehat{\pi}_a) \pi_a}{\widehat{\pi}_a} - f_\delta(\pi_a) \right)\big(\mu_a - \widehat{\mu}_a \big) \: - \\
    &\phantom{{}={\mathbb{E} \Bigg[}} \frac{\sum_a f_\delta( \widehat{\pi}_a) \widehat{\mu}_a}{\sum_t f_\delta(\pi_t)\left[\sum_t f_\delta(\widehat{\pi}_t) \right]^2} \left( \sum_t  f_\delta(\pi_a) - f_\delta(\widehat{\pi}_a)   \right)^2 \: + \\
    &\phantom{{}={\mathbb{E} \Bigg[}} \sum_a \frac{f_\delta( \widehat{\pi}_a) \widehat{\mu}_a}{\left[\sum_t f_\delta(\widehat{\pi}_t) \right]^2}  \sum_a r(a, X)  -  \frac{1}{\sum_t f_\delta(\widehat{\pi}_t) }\sum_a  r(a, X) \widehat{\mu}_a  \: + \\
    &\phantom{{}={\mathbb{E} \Bigg[}} \sum_a \Big( f_\delta(\pi_a ) \mu_a - f_\delta(\widehat{\pi}_a) \widehat{\mu}_a \Big) \left(\frac{1}{\sum_t f_\delta(\widehat{\pi}_t)} - \frac{1}{\sum_t f_\delta(\pi_t )} \right) \Bigg],
\end{align*}
which confirms that the remainder term from the von Mises expansion is a second-order product of errors. In the case that we are using the exponentially tilted intervention distribution, using the fact that
\begin{align*}
    r(a, X) &= f_\delta(\pi_a) - f_\delta(\widehat{\pi}_a) - f^\prime (\widehat{\pi}_a) \big(\pi_a - \widehat{\pi}_a \big) \\
    &= \exp(\delta a)( \pi_a -  \widehat{\pi}_a) - \exp(\delta a) (\pi_a - \widehat{\pi}_a) \\
    &= 0,
\end{align*}
it follows that the remainder term $R_2(\widehat{\mathbb{P}}, \mathbb{P})$ simplifies to
\begin{align*}
    R_2(\widehat{\mathbb{P}}, \mathbb{P}) &=  \mathbb{E} \Bigg[ \sum_a \Big( \exp(\delta a) (\pi_a \mu_a - \widehat{\pi}_a \widehat{\mu}_a)\Big) \left(\frac{1}{\sum_t \exp(\delta t) \widehat{\pi}_t} - \frac{1}{\sum_t \exp(\delta t) \pi_t}\right) \: - \\
    &\phantom{{}={\mathbb{E} \Bigg[}}  \frac{\sum_a \exp(\delta a) \widehat{\pi}_a \widehat{\mu}_a}{\sum_t \exp(\delta t) \pi_t \left[ \sum_t \exp(\delta t) \widehat{\pi}_t \right]^2}\left(\sum_a \exp(\delta a) (\pi_a - \widehat{\pi}_a) \right)^2 \Bigg] \\
    &= \mathbb{E} \Bigg[\sum_a \left(\frac{\widehat{q}_\delta}{\widehat{\pi}_a} - \frac{q_\delta}{\pi} \right) \big( \pi_a \mu_a - \widehat{\pi}_a \widehat{\mu}_a)\big)  - \sum_a \frac{q_\delta}{\pi_a} \widehat{\pi}_a \widehat{\mu}_a \left(\sum_a \frac{\widehat{q}_\delta}{\widehat{\pi}_a }(\pi_a - \widehat{\pi}_a) \right)^2\Bigg].
\end{align*}
Then, after adding and subtracting $\pi_a \widehat{\mu}_a$ in the second term, we arrive at the final expression of $R_2(\widehat{\mathbb{P}}, \mathbb{P}) = R_1 - R_2$ where
    \begin{align*}
        R_1 &= \mathbb{E} \left[\int_a \left(\frac{\widehat{q}_\delta}{\widehat{\pi}_a} - \frac{q_\delta}{\pi_a} \right)\Big((\pi_a - \widehat{\pi}_a) \widehat{\mu}_a + (\mu_a - \widehat{\mu}_a) \pi_a \Big) da\right]  \\
        R_2 &= \mathbb{E} \left[\left(\int_a \frac{q_\delta}{\pi_a} \widehat{\mu}_a\widehat{\pi}_a da \right)\left(\int_a \frac{\widehat{q}_\delta}{\widehat{\pi}_a} (\pi_a - \widehat{\pi}_a) da \right)^2\right].
    \end{align*}
\end{proof}

\subsection{Proof of \texorpdfstring{\cref{l2_upper_bound}}{Equation (2)}} \label{l2_upper_bound_proof}
\begin{proof}[\textbf{Proof:}]
In this proof we demonstrate how global positivity is used (and required) to obtain a bound on the remainder of the von Mises expansion, as in the work of \cite{diaz2020causal}. We begin by bounding the term $R_2$ from \cref{remainder_term}. Suppose that $\pi_{\min} \leq \text{min}\left\{\pi(a \mid x), \widehat{\pi}(a \mid x) \right\}$ for all $a, x$. Then,
\begin{align*}
    R_2 &= \mathbb{E} \left[\left(\int_a \frac{q_\delta}{\pi_a} \widehat{\mu}_a\widehat{\pi}_a da \right)\left(\int_a \frac{\widehat{q}_\delta}{\widehat{\pi}_a} (\pi_a - \widehat{\pi}_a) da \right)^2\right] \\
    &= \mathbb{E} \left[\left(\int_a \frac{q_\delta}{\pi_a} \widehat{\mu}_a\widehat{\pi}_a da \right)\left(\int_a \frac{\widehat{q}_\delta}{\widehat{\pi}_a} (\pi_a - \widehat{\pi}_a) \frac{\sqrt{\pi_a}}{\sqrt{\pi_a} } da \right)^2\right] \\
    &\leq \frac{1}{\pi_{\min}} \left| \mathbb{E} \left[\left(\int_a \frac{q_\delta}{\pi_a} \widehat{\mu}_a\widehat{\pi}_a da \right)\left(\int_a \frac{\widehat{q}_\delta}{\widehat{\pi}_a} (\pi_a - \widehat{\pi}_a) \sqrt{\pi}_a da \right)^2\right] \right| \\
    &\leq \frac{1}{\pi_{\min}} \left| \mathbb{E} \left[\left(\int_a \frac{q_\delta}{\pi_a} \widehat{\mu}_a\widehat{\pi}_a da \right)\left(\int_a \left(\frac{\widehat{q}_\delta}{\widehat{\pi}_a}\right)^2  da \right) \left( \int_a (\pi_a - \widehat{\pi}_a)^2 \pi_a da \right)\right] \right| \\
    &\leq \frac{1}{\pi_{\min}} \left| \left| \left(\int_a \frac{q_\delta}{\pi_a} \widehat{\mu}_a\widehat{\pi}_a da \right)\left(\int_a \left(\frac{\widehat{q}_\delta}{\widehat{\pi}_a}\right)^2  da \right) \right| \right|_\infty \left| \left|\pi - \widehat{\pi} \right| \right|^2_2
\end{align*}
where the second inequality follows by applying Cauchy-Schwarz, and the third by applying Holder's inequality. From here, we can see that
\begin{align*}
    \int_a \frac{q_\delta}{\pi_a} \widehat{\mu}_a\widehat{\pi}_a da = \frac{\int_a \exp(\delta a) \widehat{\mu}_a \widehat{\pi}_a da }{\int_a \exp(\delta a) \widehat{\pi}_a da } \leq B\left(\frac{\pi_{\max}}{\pi_{\min}} \right) 
\end{align*}
and
\begin{align*}
    \int_a \left(\frac{\widehat{q}_\delta}{\widehat{\pi}_a}\right)^2 da =   \frac{\int_a \exp(2 \delta a)  da}{\left(\int_a \exp(\delta a) \widehat{\pi}_a  da \right)^2} \leq \frac{1}{\pi_{\min}^2} \frac{\int_a \exp(2 \delta a) da}{\left(\int_a \exp(\delta a) da \right)^2} = \frac{\delta  \coth(\delta / 2)}{2 \pi_{\min}^2}
\end{align*}
so it follows that
\begin{align*}
    R_2 \leq   \left(\frac{\pi_{\max} B  \delta  \coth(\delta / 2)}{2 \pi^4_{\min}} \right)   \left| \left|\pi - \widehat{\pi} \right| \right|^2_2.
\end{align*}
Now, let us consider the $R_1$ term, which we split into two pieces $R_1 = S_1 + S_2$, which are given by
\begin{align*}
    R_1 &= \mathbb{E} \left[\int_a \left(\frac{\widehat{q}_\delta}{\widehat{\pi}_a} - \frac{q_\delta}{\pi_a} \right)(\pi_a - \widehat{\pi}_a) \widehat{\mu}_a da \right] + \mathbb{E} \left[\int_a \left(\frac{\widehat{q}_\delta}{\widehat{\pi}_a} - \frac{q_\delta}{\pi_a} \right) (\mu_a - \widehat{\mu}_a) \pi_a  da\right].
\end{align*}
Again, using the positivity assumption and then applying Cauchy-Schwarz and Holder's inequality, we can see that
\begin{align*}
    S_1 &= \mathbb{E} \left[\int_a \left(\frac{\widehat{q}_\delta}{\widehat{\pi}_a} - \frac{q_\delta}{\pi_a} \right) \frac{\sqrt{\widehat{\mu}_a \pi_a}}{\sqrt{\pi_a}} (\pi_a - \widehat{\pi}_a)  \frac{\sqrt{\widehat{\mu}_a \pi_a}}{\sqrt{\pi_a}} da \right] \\
    &\leq \frac{1}{\pi_{\min}}\mathbb{E} \left[\left(\int_a \left(\frac{\widehat{q}_\delta}{\widehat{\pi}_a} - \frac{q_\delta}{\pi_a} \right)^2 \widehat{\mu}_a \pi_a da \right)^{1/2} \left( \int_a (\pi_a - \widehat{\pi}_a)^2   \widehat{\mu}_a \pi_a da \right)^{1/2} \right] \\
    &\leq \frac{1}{\pi_{\min}}\mathbb{E} \left[\left(\int_a \left(\frac{\widehat{q}_\delta}{\widehat{\pi}_a} - \frac{q_\delta}{\pi_a} \right)^2 \widehat{\mu}_a \pi_a da \right)\right]^{1/2} \mathbb{E} \left[ \left( \int_a (\pi_a - \widehat{\pi}_a)^2   \widehat{\mu}_a \pi_a da \right) \right]^{1/2} \\
    &\leq \frac{B^2}{\pi_{\min}} \left| \left|\frac{\widehat{q}_\delta}{\widehat{\pi}} - \frac{q_\delta}{\pi} \right| \right|_2 \left| \left|\pi - \widehat{\pi} \right| \right|_2.
\end{align*}
Similarly, the second term can be bounded by
\begin{align*}
    S_2 &= \mathbb{E} \left[\int_a \left(\frac{\widehat{q}_\delta}{\widehat{\pi}_a} - \frac{q_\delta}{\pi_a} \right) (\mu_a - \widehat{\mu}_a) \pi_a  da\right]\\
    &= \mathbb{E} \left[\int_a \left[\left(\frac{\widehat{q}_\delta}{\widehat{\pi}_a} - \frac{q_\delta}{\pi_a} \right)\sqrt{\pi_a} \right] \Big[(\mu_a - \widehat{\mu}_a) \sqrt{\pi_a} \Big]  da\right]\\
    &\leq \mathbb{E} \left[\left(\int_a \left(\frac{\widehat{q}_\delta}{\widehat{\pi}_a} - \frac{q_\delta}{\pi_a} \right)^2 \pi_a  da \right)^{1/2} \left(\int_a (\mu_a - \widehat{\mu}_a)^2 \pi_a  da \right)^{1/2}\right] \\ 
    &\leq \mathbb{E} \left[\int_a \left(\frac{\widehat{q}_\delta}{\widehat{\pi}_a} - \frac{q_\delta}{\pi_a} \right)^2 \pi_a  da \right]^{1/2} \mathbb{E} \left[ \int_a (\mu_a - \widehat{\mu}_a)^2 \pi_a  da \right]^{1/2} \\
    &= \left| \left|\frac{\widehat{q}_\delta}{\widehat{\pi}} - \frac{q_\delta}{\pi} \right| \right|_2 \left| \left|\mu - \widehat{\mu} \right| \right|_2
\end{align*}
where both inequalities follow by applying the Cauchy-Schwarz inequality. Putting everything together, it follows that
\begin{align*}
    R_2(\widehat{\mathbb{P}}, \mathbb{P}) &\lesssim \left| \left| \frac{\widehat{q}_\delta}{\widehat{\pi}} - \frac{q_\delta}{\pi} \right| \right|_2 \Big( \left| \left| \widehat{\pi} - \pi \right| \right|_2 + \left| \left| \widehat{\mu} - \mu \right| \right|_2\Big) + \left| \left| \widehat{\pi} - \pi \right| \right|^2_2.
\end{align*}
\end{proof}

\subsection{Proof of \texorpdfstring{\cref{l2_lower_bound}}{Equation (3)}} \label{l2_lower_bound_proof}
\begin{proof}[\textbf{Proof:}] Suppose the conditions of \cref{sigma_bounds} hold. Then,
\begin{align*}
    \left|\left| \frac{\widehat{q}_a}{\widehat{\pi}_a} - \frac{q_a}{\pi_a} \right|\right|^2_2 &= \mathbb{E}\left[\left( \frac{\int_a \exp(\delta a) (\pi_a - \widehat{\pi}_a) da }{\left(\int_a \exp(\delta a) \widehat{\pi}_a da \right)\left(\int_a \exp(\delta a) \pi_a da \right)}\right)^2 \int_a \exp(2 \delta a)\pi_a da \right] \\
    &= \mathbb{E}\left[\left( \frac{\int_a \exp(\delta a) (\pi_a - \widehat{\pi}_a) da }{\left(\int_a \exp(\delta a) \widehat{\pi}_a da \right)}\right)^2 \left(\frac{\int_a \exp(2 \delta a)\pi_a da}{\left(\int_a \exp(\delta a) \pi_a da \right)^2 } \right)\right] \\
    &\geq \frac{\delta}{2} \left( \frac{\pi_{\min}}{\pi^2_{\max}}\right)\mathbb{E}\left[\left( \frac{\int_a \exp(\delta a) (\pi_a - \widehat{\pi}_a) da }{\left(\int_a \exp(\delta a) \widehat{\pi}_a da \right)}\right)^2\right] \\
    &= \frac{\delta}{2} \left( \frac{\pi_{\min}}{\pi^2_{\max}}\right) \left| \left| \frac{\gamma(\delta) }{\widehat{\gamma}(\delta)} - 1 \right| \right|^2_2
\end{align*}
where the inequality follows by \cref{sigma_bounds} and we define $\gamma(\delta) = \int_a \exp(\delta a)\pi_a da$.
\end{proof}

\subsection{Proof of \texorpdfstring{\cref{remainder_bound}}{Theorem 3}} \label{remainder_bound_proof}

Before we proceed with the proof of \cref{remainder_bound}, we make use of a technical lemma. The remainder in the von Mises expansion involves terms like 
\begin{align*}
    \int_a \frac{q_\delta}{\pi_a} f_a \ da \ , \ \int_a \frac{\widehat{q}_\delta}{\widehat\pi_a} f_a \ da \ , \ \int_a \left(  \frac{\widehat{q}_\delta}{\widehat\pi_a} - \frac{q_\delta}{\pi_a} \right) f_a \ da 
\end{align*}
for some bounded function $f_a$, and we want to show these are not growing with $\delta$. The following lemma gives sufficient conditions for this. 

\begin{lemma} \label{sup_norm_lemma}
Suppose that $\pi_{\min} \leq \text{min}\left\{ \pi(a \mid X), \widehat{\pi}(a \mid X) \right\}$ for some $\eta \in [0, 1)$ and all $a \in  [\eta, 1]$. Then, for some bounded function $f(a \mid x)$,
\begin{align*}
\left| \int_a \frac{q(a \mid x)}{\pi(a \mid x)} f(a \mid x) da \right| &\leq  \frac{ \sup_a | f(a \mid x) |}{\pi_{\min}(1 - \eta)} \\
\left| \int_a \frac{\widehat{q}(a \mid x)}{\widehat\pi(a \mid x)} f(a \mid x)  da \right| &\leq  \frac{ \sup_a | f(a \mid x) |}{\pi_{\min}(1 - \eta)} \\
\left| \int_a \left\{ \frac{\widehat{q}(a \mid x)}{\widehat\pi(a \mid x)} - \frac{q(a \mid x)}{\pi(a \mid x)} \right\} f(a \mid x)  da \right| &\leq \frac{ \sup_a | f(a \mid x) |}{\pi_{\min}(1 - \eta)} \left( \frac{ \sup_a | \pi(a \mid x) - \widehat\pi(a \mid x) |}{\pi_{\min}(1 - \eta)}\right) .
\end{align*}
\end{lemma}

\bigskip

\begin{proof}[\textbf{Proof:}]
Let $\| f(\cdot \mid x) \|_\infty = \sup_a | f(a \mid x)|$. Then, it follows that
\begin{align*}
    \int_a \frac{q(a \mid x)}{\pi(a \mid x)} f(a \mid x) da &\leq  \| f(\cdot \mid x) \|_\infty \left(\frac{\int^1_0\exp(\delta a)  da }{\int^1_0 \exp(\delta t) \pi(t \mid x)  dt} \right) \\
    &\leq  \| f(\cdot \mid x) \|_\infty \left(\frac{\int^1_0\exp(\delta a)  da }{\int^1_\eta \exp(\delta t) \pi(t \mid x)  dt} \right) \\
    &\leq  \frac{\| f(\cdot \mid x) \|_\infty}{\pi_{\min}} \left(\frac{\int^1_0\exp(\delta a)  da }{\int^1_\eta \exp(\delta t) dt} \right).
\end{align*}
From here, after integration we can see that for all $\delta > 0$
\begin{align*}
    \frac{\int^1_0\exp(\delta a)  da }{\int^1_\eta \exp(\delta t) dt} = \frac{\exp(\delta) - 1}{\exp(\delta) - \exp(\eta \delta)} = \frac{1 - \exp(-\delta)}{1 - \exp(-\delta(1 - \eta))} \leq \frac{1}{1 - \eta},
\end{align*}
and so this yields the first result,
\begin{align*}
    \left| \int_a \frac{q(a \mid x)}{\pi(a \mid x)} f(a \mid x) da \right| \leq \frac{ \sup_a | f(a \mid x) |}{\pi_{\min}(1 - \eta)}.
\end{align*}
The logic for the second result is exactly the same since we have assumed $\widehat{\pi}(a \mid x)$ is lower bounded just like $\pi(a \mid x)$. For the third result we have that $S := \int^{1}_{0} \left(  \frac{\widehat{q}_\delta}{\widehat\pi_a} - \frac{q_\delta}{\pi_a} \right) f(a \mid x)  da$ is bounded as
\begin{align*}
    S &= \int^{1}_{0} \left( \frac{\exp(\delta a)}{\int^{1}_{0} \exp(\delta t) \widehat{\pi}_t dt} - \frac{\exp(\delta a)}{\int^{1}_{0} \exp(\delta t) \pi_t  dt}  \right) f(a \mid x)  da \\
    &\leq \| f(\cdot \mid x) \|_\infty \int^1_0 \left( \frac{\exp(\delta a)}{\int^1_0 \exp(\delta t) \pi_t  dt} \left|  \frac{\int^1_0 \exp(\delta t) \pi_t dt}{\int^1_0\exp(\delta t) \widehat{\pi}_t  dt}  - 1 \right| \right) da \\
    &\leq \frac{\| f(\cdot \mid x) \|_\infty}{\pi_{\min}(1 - \eta)}  \left|  \frac{\int^1_0 \exp(\delta t) |\pi_t - \widehat{\pi}_t|  dt}{\int^1_0 \exp(\delta t) \widehat{\pi}_t  dt} \right|  \\
    &\leq \frac{\| f(\cdot \mid x) \|_\infty}{[\pi_{\min}(1-\eta)]^2}  \| \pi(\cdot \mid x) - \widehat\pi(\cdot \mid x) \|_\infty
\end{align*}
following the same steps as the first two results.
\end{proof}
Now that we have established this technical lemma, we can complete the proof of \cref{remainder_bound}. 
\begin{proof}[\textbf{Proof:}]
    Recall that by \cref{remainder_term} it follows that  $R_2(\widehat{\mathbb{P}}, \mathbb{P}) = R_1 - R_2$ where
    \begin{align*}
        R_1 &= \mathbb{E} \left[\int_a \left(\frac{\widehat{q}_\delta}{\widehat{\pi}_a} - \frac{q_\delta}{\pi_a} \right)\Big((\pi_a - \widehat{\pi}_a) \widehat{\mu}_a + (\mu_a - \widehat{\mu}_a) \pi_a \Big) da\right]  \\
        R_2 &= \mathbb{E} \left[\left(\int_a \frac{q_\delta}{\pi_a} \widehat{\mu}_a\widehat{\pi}_a da \right)\left(\int_a \frac{\widehat{q}_\delta}{\widehat{\pi}_a} (\pi_a - \widehat{\pi}_a) da \right)^2\right].
    \end{align*}
    We begin by considering the $R_2$ term. By applying \cref{sup_norm_lemma} followed by Holder's inequality, we can see that
    \begin{align*}
        R_2 &\leq \mathbb{E}\left[\left( \frac{\text{sup}_a \ |\widehat{\mu}_a \widehat{\pi}_a |}{\pi_{\min}(1 - \eta)}   \right)\left( \frac{\text{sup}_a \ | \pi_a - \widehat{\pi}_a |}{\pi_{\min}(1 - \eta)} \right)^2\right] \leq \frac{B \pi_{\max}}{[\pi_{\min}(1 - \eta)]^3} \int \text{sup}_a \ | \pi_a - \widehat{\pi}_a |^2 d\mathbb{P}(x)
    \end{align*}
    where we assume $|\widehat{\mu}(x, a)| \leq B$ for all $x, a$. Similarly, applying \cref{sup_norm_lemma} to the first term yields
    \begin{align*}
        R_1 &\leq \mathbb{E}\left[\left(\frac{\text{sup}_a | \widehat{\pi} - \pi_a| }{[\pi_{\min}(1-\eta)]^2} \right)\text{sup}_a   \Big| (\pi_a - \widehat{\pi}_a) \widehat{\mu}_a + (\mu_a - \widehat{\mu}_a) \pi_a \Big|   \right] \\
        &\leq \underbrace{\mathbb{E} \left[\left(\frac{\text{sup}_a | \widehat{\pi} - \pi_a| }{[\pi_{\min}(1-\eta)]^2} \right)\text{sup}_a   \big| (\pi_a - \widehat{\pi}_a) \widehat{\mu}_a \big|  \right]}_{(i)} + \underbrace{\mathbb{E} \left[\left(\frac{\text{sup}_a | \widehat{\pi} - \pi_a| }{[\pi_{\min}(1-\eta)]^2} \right) \text{sup}_a \big| (\mu_a - \widehat{\mu}_a) \pi_a \big|  \right]}_{(ii)}. 
\end{align*}
From here, we can complete the upper bound by first observing that
\begin{align*}
    (i) &\leq B \mathbb{E} \left[\left(\frac{\text{sup}_a | \widehat{\pi} - \pi_a| }{[\pi_{\min}(1-\eta)]^2} \right)\text{sup}_a  | \pi_a - \widehat{\pi}_a  |  \right] \\
    &= \frac{B}{[\pi_{\min}(1-\eta)]^2}  \int  \text{sup}_a \  | \pi_a - \widehat{\pi}_a |^2  d \mathbb{P}(x),
\end{align*}
and then by applying the Cauchy-Schwarz inequality, which yields
\begin{align*}
    (ii) &\leq \pi_{\max}  \mathbb{E} \left[\left(\frac{\text{sup}_a | \widehat{\pi} - \pi_a| }{[\pi_{\min}(1-\eta)]^2} \right) \text{sup}_a | \mu_a - \widehat{\mu}_a|   \right] \\
    &\leq \pi_{\max}  \mathbb{E} \left[\left(\frac{\text{sup}_a | \widehat{\pi} - \pi_a|^2 }{[\pi_{\min}(1-\eta)]^4} \right) \right]^{1/2} \mathbb{E} \left[\underset{a}{\text{sup}} \ |\mu_a - \widehat{\mu}_a|^2   \right]^{1/2} \\
    &\leq \frac{\pi_{\max}}{[\pi_{\min}(1-\eta)]^2}  \left( \int \text{sup}_a | \widehat{\pi} - \pi_a|^2 d\mathbb{P}(x) \right)^{1/2} \left( \int \text{sup}_a \ |\mu_a - \widehat{\mu}_a|^2  d \mathbb{P}(x) \right)^{1/2}.
\end{align*}
Therefore, putting everything together it follows that
\begin{align*}
        R_2(\widehat{\mathbb{P}}, \mathbb{P}) \lesssim  || \widehat{\pi} - \pi||_{L^2_x, L^\infty_a} \cdot   || \widehat{\mu} - \mu ||_{L^2_x, L^\infty_a}  +  || \widehat{\pi} - \pi||^2_{L^2_x, L^\infty_a}
    \end{align*}
    where $||f||^2_{L^2_x, L^\infty_a} =  \int \big( \underset{a}{\text{sup}} \, |f(x, a) | \big)^2 d \mathbb{P}(x)$.
\end{proof}

\subsection{Proof of \texorpdfstring{\cref{asymptotic_normality}}{Theorem 4}} \label{asymptotic_normality_proof}

We will split up the proof of \cref{asymptotic_normality} into two proofs: one, showing that Lindeberg's condition holds and therefore $(\mathbb{P}_n - \mathbb{P})\{\varphi(Z; \mathbb{P})\}$ converges to a normal distribution, and a second proof showing that the empirical process term is of order $o_\mathbb{P}(\sqrt{\delta})$. We then combine these results and that of \cref{remainder_bound} to complete the proof of \cref{asymptotic_normality}.

\begin{proof}[\textbf{Proof:}]
First, we show that $\varphi(Z; \mathbb{P})$ satisfies Lindeberg's condition, which then implies that $(\mathbb{P}_n - \mathbb{P})\{\varphi(Z; \mathbb{P})\}$ converges to a normal distribution. Note that
   \begin{align*}
       \varphi_i = \frac{q_\delta(A \mid X_i)}{\pi(A \mid X_i)}\left(Y_i - \mu(X_i, A) \right)  +  \mathbb{E}_Q(\mu(X_i, A) \mid X_i)
   \end{align*}
where $\mathbb{E}_Q(\mu(X_i, A) \mid X_i) =  \sum_{a \in \mathcal{A}} \mu(a, X_i) q_\delta(a \mid X_i)$. We wish to show that for all $\varepsilon > 0$,
\begin{align*}
    \lim_{n\to\infty}\frac{1}{s_n^2}\sum_{k=1}^n\mathbb{E}\left[
    \varphi^2_i  \mathbbm{1}\{|\varphi_i|>\varepsilon s_n\}
    \right] = 0
\end{align*}
where $s^2_n = \sum^n_{i=1} \var(\varphi_i(Z_i))$. Note that by applying \cref{sigma_bounds} we may lower bound $\var(\varphi_i(Z_i))$ as
\begin{align*}
    \sum^n_{i=1} \var(\varphi_i(Z_i)) \geq  n\delta \left( \frac{\pi_{\min} \sigma^2_{\min}}{\pi^2_{\max}} \cdot \frac{2}{5}(1 - \eta) \right).
\end{align*}
Thus, it follows that 
\begin{align*}
    \mathbb{E}\left[
    \varphi^2_i \mathbbm{1}\{|\varphi_i|>\varepsilon s_n\}
    \right] \leq \mathbb{E}\left[
    \varphi^2_i \mathbbm{1}\left\{ |\varphi_i|> \varepsilon \sqrt{n \delta} \left(\frac{\pi^{1/2}_{\min} \sigma_{\min}}{\pi_{\max}} \sqrt{\frac{2(1 - \eta)}{5}} \right)\right\} 
    \right].
\end{align*}
From here, we upper bound each term in $|\varphi_i|$ individually. Assume that $|Y| \leq B$ with probability one. Then, it follows that
\begin{align*}
    \left|\frac{q_\delta(A \mid X_i)}{\pi(A \mid X_i)}\left(Y_i - \mu(X_i, A) \right) \right| 
    &\leq 2B\left|\frac{q_\delta(A \mid X_i)}{\pi(A \mid X_i)} \right| \\
    &\leq 2B\left| \frac{\exp(\delta a)}{\int^1_\eta \exp(\delta a) \pi(a \mid X_i) da } \right| \\
    &\leq \frac{2B}{\pi_{\min}} \left(\frac{\delta \exp(\delta) }{\exp(\delta) - \exp(\eta \delta)} \right) \\
    &= \delta \left( \frac{2B}{\pi_{\min}} \cdot \frac{1}{1 - \exp(-\delta(1 - \eta))}\right)
\end{align*}
Next, we can see that
\begin{align*}
    \left|\mathbb{E}_Q(\mu(X_i, A) \mid X_i) \right| =  \left|\sum_{a \in \mathcal{A}} \mu(a, X_i) q_\delta(a \mid X_i) \right| \leq B \sum_{a \in \mathcal{A}} q_\delta(a \mid X_i) = B.
\end{align*}
Putting both results together, it follows that $\mathbb{E}\left[
    \varphi^2_i \mathbbm{1}\{|\varphi_i|>\varepsilon s_n\}
    \right]$ is bounded above by
\begin{align*}
     \mathbb{E}\left[
    \varphi^2_i \mathbbm{1}\left\{\frac{2B}{\pi_{\min}} \cdot \frac{1}{1 - \exp(-\delta(1 - \eta))} + \frac{B}{\delta} > \varepsilon \sqrt{\frac{n}{\delta}} \left(\frac{\pi^{1/2}_{\min} \sigma_{\min}}{\pi_{\max}} \sqrt{\frac{2(1 - \eta)}{5}} \right)\right\} 
    \right]
\end{align*}
and furthermore that
\begin{align*}
     \mathbbm{1}\left\{\frac{2B}{\pi_{\min}} \cdot \frac{1}{1 - \exp(-\delta(1 - \eta))} + \frac{B}{\delta} > \varepsilon \sqrt{\frac{n}{\delta}} \left(\frac{\pi^{1/2}_{\min} \sigma_{\min}}{\pi_{\max}} \sqrt{\frac{2(1 - \eta)}{5}} \right)\right\}  \to 0
\end{align*}
as long as $n / \delta \to \infty$. To complete the proof, we use the dominated convergence theorem. Observe that for all $n \in \mathbb{N}$,
\begin{align*}
    \frac{1}{s_n^2}\sum_{k=1}^n
    \varphi^2_i \mathbbm{1}\{|\varphi_i|>\varepsilon s_n\}
     \leq  \frac{1}{s_n^2}\sum_{k=1}^n
    \varphi^2_i
\end{align*}
and furthermore that
\begin{align*}
     \mathbb{E}\left[\frac{1}{s_n^2}\sum_{k=1}^n
    \varphi^2_i \right] = \frac{1}{\sum^n_{i=1} \mathbb{E}(\varphi^2_i(Z_i))} \sum^n_{i=1} \mathbb{E}(\varphi^2_i(Z_i)) < \infty.
\end{align*}
Therefore, by the dominated convergence theorem,
\begin{align*}
 \underset{n \to \infty}{\text{lim}} \left\{ \frac{1}{s_n^2}\sum_{k=1}^n\mathbb{E}\left[
    \varphi^2_i  \mathbbm{1}\{|\varphi_i|>\varepsilon s_n\}
    \right] \right\} = 0.
\end{align*}
\end{proof}

Next, we prove that the empirical process term is of order $o_\mathbb{P}(\sqrt{\delta})$.

\begin{proof}[\textbf{Proof:}]
   Here, we will show that the $L_2(\mathbb{P})$ norm of $\varphi(Z; \hat{\mathbb{P}}) - \varphi(Z; \mathbb{P})$ is $o_\mathbb{P}(\sqrt{\delta})$. To proceed, first recall that
   \begin{align*}
       \varphi(Z; \mathbb{P}) = \frac{q_\delta(A \mid X)}{\pi(A \mid X)} Y - \frac{q_\delta(A \mid X)}{\pi(A \mid X)}\mathbb{E}_Q(\mu(X, A) \mid X)  + \mathbb{E}_Q(\mu(X, A) \mid X).
   \end{align*}
As a result, we can see that $\varphi(Z; \hat{\mathbb{P}}) - \varphi(Z; \mathbb{P}) = D_1 + D_2 + D_3$  after grouping together like terms, where
\begin{align*}
    D_1 &= \left(\frac{\hat{q}_\delta(A \mid X)}{\hat{\pi}(A \mid X)} - \frac{q_\delta(A \mid X)}{\pi(A \mid X)} \right) Y \\
    D_2 &= \mathbb{E}_{\hat{Q}}(\hat{\mu}(X, A) \mid X) - \mathbb{E}_Q(\mu(X, A) \mid X) \\
    D_3 &= \frac{q_\delta(A \mid X)}{\pi(A \mid X)}\mathbb{E}_Q(\mu(X, A) \mid X) - \frac{\hat{q}_\delta(A \mid X)}{\hat{\pi}(A \mid X)}\mathbb{E}_{\hat{Q}}(\hat{\mu}(X, A) \mid X).
\end{align*}
Then, applying the triangle inequality, it follows that we can bound the $L_2(\mathbb{P})$ norm of $\varphi(Z; \hat{\mathbb{P}}) - \varphi(Z; \mathbb{P})$ by
\begin{align*}
    ||\varphi(Z; \hat{\mathbb{P}}) - \varphi(Z; \mathbb{P})||_2 \leq ||D_1||_2 + ||D_2||_2 + ||D_3||_2.
\end{align*}
From here, we bound each of these norms individually. First, observe that under the assumption that $|Y| \leq B$ with probability one,
\begin{align*}
    ||D_1||_2 &\leq B \cdot \mathbb{E} \left[\int_a \left(\frac{\hat{q}_\delta(a \mid X)}{\hat{\pi}(a \mid X)} - \frac{q_\delta(a \mid X)}{\pi(a \mid X)} \right)^2 \pi_a da\right]^{1/2}.
\end{align*}
Next, following a similar approach to the proof of \cref{sup_norm_lemma}, we can see that
\begin{align*}
    \frac{\hat{q}_\delta}{\hat{\pi}_a} - \frac{q_\delta}{\pi_a} &= \frac{\exp(\delta a)}{\int_t \exp(\delta t) \hat{\pi}_t dt} - \frac{\exp(\delta a)}{\int_t \exp(\delta t) \pi_t dt} \\
    &= \frac{\exp(\delta a)}{\int_t \exp(\delta t) \pi_t dt}\left(\frac{\int_t \exp(\delta t) \pi_t dt}{\int_t \exp(\delta t) \hat{\pi}_t dt} - 1 \right) \\
    &= \frac{\exp(\delta a)}{\int_t \exp(\delta t) \pi_t dt}\left(\frac{\int_t \exp(\delta t) (\pi_t - \hat{\pi}_t) dt}{\int_t \exp(\delta t) \hat{\pi}_t dt} \right) \\
&\leq \frac{\text{sup}_a |\pi_a - \hat{\pi}_a|}{\pi_{\min}} \frac{\exp(\delta a)}{\int_t \exp(\delta t) \pi_t dt}.
\end{align*}
Consequently, we may bound the $L_2(\mathbb{P})$ norm of $D_1$ by
\begin{align*}
    ||D_1||_2 &\leq B \cdot \mathbb{E} \left[\frac{\text{sup}_a |\pi_a - \hat{\pi}_a|^2}{\pi^2_{\min}}\left(\frac{\int_a \exp(2 \delta a) \pi_a da}{\left(\int_t \exp(\delta t) \pi_t dt\right)^2}  \right)\right]^{1/2} \\
    &\leq \frac{B}{\pi_{\min}} \left( \int \underset{a}{\text{sup}} \ |\pi_a - \hat{\pi}_a|^2 d\mathbb{P}(x) \right)^{1/2} \left( \left| \left| \frac{\int_a \exp(2 \delta a) \pi_a da}{\left(\int_t \exp(\delta t) \pi_t dt\right)^2} \right| \right|_\infty  \right)^{1/2} \\
    &\leq \frac{B}{\pi_{\min}} \left( \int \underset{a}{\text{sup}} \ |\pi_a - \hat{\pi}_a|^2 d\mathbb{P}(x) \right)^{1/2}  \frac{\pi_{\max}}{\pi^2_{\min}}\left(\frac{\delta}{2} +  \frac{1}{(1 - \eta)^2}\right)^{1/2} \\
    &\leq \frac{B}{\pi_{\min}} \left( \int \underset{a}{\text{sup}} \ |\pi_a - \hat{\pi}_a|^2 d\mathbb{P}(x) \right)^{1/2} \left(\frac{\sqrt{\pi_{\max}}}{\pi_{\min}}\left(\frac{\sqrt{\delta}}{\sqrt{2}} + \frac{1}{1 - \eta}  \right)  \right),
\end{align*}
where the second inequality follows by Holder's inequality, the third inequality by applying \cref{sigma_bounds}, and the fourth by using $\sqrt{a + b} \leq \sqrt{a} + \sqrt{b}$ for positive real numbers $a$ and $b$. Therefore, as long as $( \int \text{sup}_a  |\pi_a - \hat{\pi}_a|^2 d\mathbb{P}(x))^{1/2} = o_\mathbb{P}(1)$ it follows that $||D_1||_2 = o_\mathbb{P}(\sqrt{\delta})$. Next, to bound the $L_2(\mathbb{P})$ norm of $D_2$ we rewrite the expression as simple differences of nuisance functions. Observe that
\begin{align*}
   D_2 &=  \int_a \hat{\mu}_a \hat{q}_\delta(a \mid X) \ da - \int_a \mu_a q_\delta(a \mid X) \ da \\
   &= \int_a \hat{\mu}_a \hat{q}_\delta(a \mid X) \ da - \int_a \mu_a \hat{q}_\delta(a \mid X) \ da + \int_a \mu_a \hat{q}_\delta(a \mid X) \ da - \int_a \mu_a q_\delta(a \mid X) \ da \\
   &= \int_a (\hat{\mu}_a - \mu_a) \hat{q}_\delta(a \mid X) \ da  + \int_a  \big(\hat{q}_\delta(a \mid X) \ da -  q_\delta(a \mid X)\big)\mu_a da \\
   &\leq \text{sup}_a \ |\hat{\mu}_a - \mu_a| + \text{sup}_a \ | \mu_a| \int_a  \big(\hat{q}_\delta(a \mid X) \ da -  q_\delta(a \mid X)\big) \ da \tag{$i$}
\end{align*}
where we use the shorthand $\mu_a = \mu(X, a)$ and in the last inequality we use the fact that $\int_a \hat{q}_\delta(a \mid X)da = 1$. Therefore, it follows that
\begin{align*}
    ||D_2||_2 \leq \left(\int  \underset{a}{\text{sup}} \ |\hat{\mu}_a - \mu_a|^2 d\mathbb{P}(x) \right)^{1/2} + \underbrace{\left| \left|\text{sup}_a \ | \mu_a| \int_a  \big(\hat{q}_\delta(a \mid X) \ da -  q_\delta(a \mid X)\big) \ da \right| \right|_2}_{D_{22}}.
\end{align*}
Note that
\begin{align*}
     D_{22} &= \mathbb{E} \left[\left(\text{sup}_a \ | \mu_a| \int_a  \big(\hat{q}_\delta(a \mid X) -  q_\delta(a \mid X)\big) \ da \right)^2 \right]^{1/2} \\
     &\leq  ||\mu^2(X, A)||_\infty \mathbb{E} \left[\left( \int_a  \big(\hat{q}_\delta(a \mid X) \ da -  q_\delta(a \mid X)\big) \ da \right)^2 \right]^{1/2} \\
    &\leq  ||\mu^2(X, A)||_\infty \mathbb{E} \left[\left(\frac{\text{sup}_a |\pi_a - \widehat\pi_a|}{\pi_{\min}^2}  \right)^2 \right]^{1/2} \\
    &= \left|\left|\frac{\mu^2(X, A)}{\pi^2_{\min}}\right|\right|_\infty \left( \int \underset{a}{\text{sup}} \ |\pi_a - \widehat\pi_a|^2 d \mathbb{P}(x) \right)^{1/2}.
\end{align*}
Thus, it follows that $||D_2||_2 = o_\mathbb{P}(1)$ assuming that $(\int \text{sup}_a \ |\hat{\mu}_a - \mu_a|^2 d\mathbb{P}(x))^{1/2} = o_\mathbb{P}(1)$ and $( \int \text{sup}_a \ |\pi_a - \widehat\pi_a|^2 d \mathbb{P}(x))^{1/2} = o_\mathbb{P}(1)$. Clearly, under these assumptions $||D_2||_2 = o_\mathbb{P}(\sqrt{\delta})$ as well. Following the same logic for the bounds on $D_1$ and $D_2$, we may write $D_3$ as $E_1 + E_2$ where
\begin{align*}
    E_1 &= \frac{q_\delta(A \mid X)}{\pi(A \mid X)}\Big(\mathbb{E}_Q(\mu(X, A) \mid X) - \mathbb{E}_{\hat{Q}}(\hat{\mu}(X, A) \mid X) \Big), \\
    E_2 &=\left(\frac{q_\delta(A \mid X)}{\pi(A \mid X)} - \frac{\hat{q}_\delta(A \mid X)}{\hat{\pi}(A \mid X)}\right)\mathbb{E}_{\hat{Q}}(\hat{\mu}(X, A) \mid X).
\end{align*}
Therefore, it follows that $||D_3||_2 \leq ||E_1||_2 + ||E_2||_2$, and furthermore, that
\begin{align*}
    ||E_1||_2 &= \mathbb{E}\left[\frac{q^2_\delta(A \mid X)}{\pi^2(A \mid X)}\Big(\mathbb{E}_Q(\mu(X, A) \mid X) - \mathbb{E}_{\hat{Q}}(\hat{\mu}(X, A) \mid X) \Big)^2\right]^{1/2} \\
    &= \mathbb{E}\left[\Big(\mathbb{E}_Q(\mu(X, A) \mid X) - \mathbb{E}_{\hat{Q}}(\hat{\mu}(X, A) \mid X) \Big)^2 \int_a \frac{q^2_\delta(a \mid X)}{\pi(a \mid X)} da\right]^{1/2} \\
    &\leq  \mathbb{E}\left[\Big(\mathbb{E}_Q(\mu(X, A) \mid X) - \mathbb{E}_{\hat{Q}}(\hat{\mu}(X, A) \mid X) \Big)^2\right]^{1/2} \left( \left| \left| \int_a \frac{q^2_\delta(a \mid X)}{\pi(a \mid X)} da \right| \right|_\infty \right)^{1/2} \\
    &\leq \frac{\sqrt{\pi_{\max}}}{\pi_{\min}}\left(\frac{\sqrt{\delta}}{\sqrt{2}} + \frac{1}{1 - \eta}  \right)  \Big| \Big|\underbrace{\mathbb{E}_Q(\mu(X, A) \mid X) - \mathbb{E}_{\hat{Q}}(\hat{\mu}(X, A) \mid X)}_{D_2} \Big| \Big|_2
\end{align*}
where the first inequality follows by Holder's inequality and the second by applying \cref{sigma_bounds}. Thus, under the assumption that $(\int \text{sup}_a \ |\hat{\mu}_a - \mu_a|^2 d\mathbb{P}(x))^{1/2} = o_\mathbb{P}(1)$ and that $( \int \text{sup}_a \ |\pi_a - \widehat\pi_a|^2 d \mathbb{P}(x))^{1/2} = o_\mathbb{P}(1)$ it follows that  $||E_1||_2 = o_\mathbb{P}(\sqrt{\delta})$. Next, note that
\begin{align*}
    ||E_2||_2 &= \mathbb{E}\left[\left(\frac{q_\delta(A \mid X)}{\pi(A \mid X)} - \frac{\hat{q}_\delta(A \mid X)}{\hat{\pi}(A \mid X)}\right)^2\mathbb{E}_{\hat{Q}}(\hat{\mu}(X, A) \mid X)^2 \right]^{1/2} \\
    &\leq \mathbb{E}\left[\left(\frac{q_\delta(A \mid X)}{\pi(A \mid X)} - \frac{\hat{q}_\delta(A \mid X)}{\hat{\pi}(A \mid X)}\right)^2 \right]^{1/2} \left| \left| \mathbb{E}_{\hat{Q}}(\hat{\mu}(X, A) \mid X)^2 \right| \right|_\infty.
\end{align*}
The first term is identical to the bound on $D_1$, and for the second term it follows that
\begin{align*}
    \left| \left| \mathbb{E}_{\hat{Q}}(\hat{\mu}(X, A) \mid X)^2 \right| \right|_\infty = \left| \left| \left(\int_a \hat{\mu}_a \hat{q}_\delta(a \mid X) \ da\right)^2 \right| \right|_\infty \leq ||\hat{\mu}^2(X, A)||_\infty 
\end{align*}
by again using that $\int_a \hat{q}_\delta(a \mid X) \ da = 1$. Therefore, $||E_2||_2 = o_\mathbb{P}(\sqrt{\delta})$ under the assumption that $( \int \text{sup}_a  |\pi_a - \hat{\pi}_a|^2 d\mathbb{P}(x))^{1/2} = o_\mathbb{P}(1)$. Putting everything together, we can see that
\begin{align*}
    ||\varphi(Z; \hat{\mathbb{P}}) - \varphi(Z; \mathbb{P})||_2 = o_\mathbb{P}(\sqrt{\delta})
\end{align*}
\end{proof}

With these two technical proofs complete, we can now finalize the proof of \cref{asymptotic_normality}. Recall, the decomposition of the bias-corrected estimator is given by
\begin{align*}
   \frac{\sqrt{n}}{\sigma_\delta}\left(\widehat{\psi}(\delta) - \psi(\delta)\right) &= \frac{\sqrt{n}}{\sigma_\delta} \left[(\mathbb{P}_n - \mathbb{P})\{\varphi(Z; \mathbb{P})\} + (\mathbb{P}_n - \mathbb{P})\{\varphi(Z; \widehat{\mathbb{P}}) - \varphi(Z; \mathbb{P})\} + R_2(\widehat{\mathbb{P}}, \mathbb{P}) \right]
\end{align*}
where $\sigma^2_\delta = \var(\varphi_\delta)$ is the nonparametric efficiency bound defined in \cref{nonpar_efficiency_bound}. After showing Lindeberg's condition holds, it follows that the first term converges in distribution to a standard Normal,
\begin{align*}
    \frac{\sqrt{n}}{\sigma_\delta} \Big[(\mathbb{P}_n - \mathbb{P})\{\varphi(Z; \mathbb{P})\} \Big] \overset{d}{\longrightarrow} N(0, 1).
\end{align*}
For the second term, i.e., the empirical process term, recall by \cref{sigma_bounds} that $\sigma_\delta = O(\sqrt{\delta})$. Then, recall that we proved that the $L_2(\mathbb{P})$ norm of  $\varphi(Z; \widehat{\mathbb{P}}) - \varphi(Z; \mathbb{P})$ is $o_\mathbb{P}(\sqrt{\delta})$, under the assumption that $|| \widehat{\pi} - \pi||_{L^2_x, L^\infty_a} = o_\mathbb{P}(1)$ and $ || \widehat{\mu} - \mu ||_{L^2_x, L^\infty_a} = o_\mathbb{P}(1)$. Thus, it follows that
\begin{align*}
    \frac{\sqrt{n}}{\sigma_\delta} \left[(\mathbb{P}_n - \mathbb{P})\{\varphi(Z; \widehat{\mathbb{P}}) - \varphi(Z; \mathbb{P})\} \right] = o_{\mathbb{P}}(1).
\end{align*}
Finally, in \cref{remainder_bound} we proved that  
\begin{align*}
       R_2(\widehat{\mathbb{P}}, \mathbb{P}) \lesssim  || \widehat{\pi} - \pi||_{L^2_x, L^\infty_a} \Big(  || \widehat{\mu} - \mu ||_{L^2_x, L^\infty_a}  +  || \widehat{\pi} - \pi||_{L^2_x, L^\infty_a} \Big).
 \end{align*}
Therefore, if we assume $ || \widehat{\pi} - \pi||_{L^2_x, L^\infty_a} \cdot   || \widehat{\mu} - \mu ||_{L^2_x, L^\infty_a}  +  || \widehat{\pi} - \pi||^2_{L^2_x, L^\infty_a} = o_{\mathbb{P}}(\sqrt{\delta / n})$, then it follows that
\begin{align*}
    \frac{\sqrt{n}}{\sigma_\delta} \left[R_2(\widehat{\mathbb{P}}, \mathbb{P}) \right] = o_{\mathbb{P}}(1),
\end{align*}
and consequently that
\begin{align*}
    \frac{\sqrt{n}}{\sigma_\delta}\left(\widehat{\psi}(\delta) - \psi(\delta)\right)\overset{d}{\longrightarrow}N(0,1),
\end{align*}
thereby completing the proof of \cref{asymptotic_normality}.

\subsection{Proof of \cref{reflected_eif}} \label{reflected_eif_proof} 

\begin{proof}[\textbf{Proof:}] First, recall that the reflected intervention distribution, $r_\delta(a \mid x)$, is defined as
    \begin{align*}
     \frac{\text{exp}(\delta a) \pi(a \mid x) \mathbbm{1}(a \leq a^\prime)\left\{ \int^{a^\prime}_0 \pi(t \mid x) dt \right\} }{\int^{a^\prime}_0 \text{exp}(\delta t) \pi(t \mid x) dt} +  \frac{\text{exp}(-\delta a) \pi(a \mid x) \mathbbm{1}(a > a^\prime) \left\{ \int^{1}_{a^\prime} \pi(t \mid x) dt \right\}}{\int^{1}_{a^\prime} \text{exp}(-\delta t) \pi(t \mid x) dt}.
\end{align*}
We again follow the ``derivative rules'' of \cite{kennedy2023semiparametric} where we assume the data are discrete and use simple influence functions as building blocks to derive a candidate influence function, then confirm its validity by showing the von Mises expansion has a remainder that is a second order product of errors. Then, the candidate influence function for the reflected incremental effect $\mathbb{IF}(\psi_R(\delta, a^\prime))$ is given by
\begin{align*}
     \sum_x \sum_a \Big\{\mathbb{IF}( \mu(x, a) ) r_\delta(a \mid x)  p(x) + \mu(x, a) \mathbb{IF}( r_\delta(a \mid x))  p(x) +  \mu(x, a) r_\delta(a \mid x) \mathbb{IF}( p(x)) \Big\}.
\end{align*}
From here, again using
\begin{align*}
    \mathbb{IF}(\mu(x, a)) &= \frac{\mathbbm{1}(A = a, X = x)}{\pi(a \mid x) p(x)}(Y - \mu(x, a)) \\
    \mathbb{IF}(p(x)) &= \mathbbm{1}(X = x) - p(x)
\end{align*}
it follows that the first term reduces to
\begin{align*}
     \sum_x \sum_a \frac{\mathbbm{1}(A = a, X = x)}{\pi(a \mid x) p(x)}(Y - \mu(x, a)) r_\delta(a \mid x)  p(x) = \frac{r_\delta(A \mid X)}{\pi(A \mid X)} \Big(Y - \mu(X, A) \Big)
\end{align*}
and the third term to
\begin{align*}
   \sum_x \sum_a  \mu(x, a) r_\delta(a \mid x) \Big( \mathbbm{1}(X = x) - p(x) \Big) = \sum_a  \mu(X, a) r_\delta(a \mid X) - \psi_R(\delta, a^\prime).
\end{align*}
Now, in order to evaluate the second term, we decompose $\sum_x \sum_a \mu(x, a) \mathbb{IF}( r_\delta(a \mid x)) p(x)$ into the summation $T_1 + T_2 + T_3 + T_4$, where
\begin{align*}
    T_1 &=   \sum_x \sum_a  \mu(x, a) \left( \frac{\mathbb{IF}(f(\pi(a \mid x)))}{\sum_{t \leq a^\prime} \text{exp}(\delta t) \pi(t \mid x)} \right) p(x) \\
    T_2 &= \sum_x \sum_a  \mu(x, a) f(\pi(a \mid x)) \mathbb{IF}\Bigg(\Bigg[\sum_{t \leq a^\prime} \text{exp}(\delta t) \pi(t \mid x) \Bigg]^{-1} \Bigg) p(x) \\
    T_3 &=  \sum_x \sum_a  \mu(x, a)\left(\frac{\mathbb{IF}(g(\pi(a \mid x)))}{ \sum_{t > a^\prime} \text{exp}(-\delta t) \pi(t \mid x)}\right) p(x) \\
    T_4 &=  \sum_x \sum_a  \mu(x, a) g(\pi(a \mid x)) \mathbb{IF}\Bigg(\Bigg[\sum_{t > a^\prime} \text{exp}(-\delta t) \pi(t \mid x) \Bigg]^{-1} \Bigg) p(x)
\end{align*}
where we define the expressions $f(\pi(a \mid x)) = \text{exp}(\delta a) \pi(a \mid x) \mathbbm{1}(a \leq a^\prime)\{ \int^{a^\prime}_0 \pi(t \mid x) dt \}$ and $g(\pi(a \mid x)) = \text{exp}(-\delta a) \pi(a \mid x) \mathbbm{1}(a > a^\prime) \{ \int^{1}_{a^\prime} \pi(t \mid x) dt\}$. We now evaluate the candidate influence function for each term in $\mathbb{IF}(r_\delta(a \mid x))$. Using the fact that
\begin{align*}
    \mathbb{IF}(\pi(a\mid x)) = \frac{\mathbbm{1}(X = x)}{p(x)} \Big( \mathbbm{1}(A = a) - \pi(a \mid x) \Big)
\end{align*}
we can see that
\begin{align*}
    \mathbb{IF}(f(\pi(a \mid x))) = \underbrace{\exp(\delta a) \mathbb{IF}( \pi(a \mid x)) \mathbbm{1}(a \leq a^\prime) F_{\pi}(a^\prime)}_{(i)} + \underbrace{\exp(\delta a) \pi(a \mid x) \mathbbm{1}(a \leq a^\prime) \mathbb{IF}(F_{\pi}(a^\prime) )}_{(ii)},
\end{align*}
where for notational convenience, we define $F_{\pi}(a^\prime) := \sum_{t \leq a^\prime} \pi(t \mid x)$. Next, it follows that
\begin{align*}
    (i) &= \frac{\exp(\delta a) \mathbbm{1}(X = x)}{p(x)} \Big( \mathbbm{1}(A = a) - \pi(a \mid x)\Big)\mathbbm{1}(a \leq a^\prime)F_{\pi}(a^\prime) \\
    &= \frac{\exp(\delta a) \mathbbm{1}(X = x) \mathbbm{1}(A = a) \mathbbm{1}(a \leq a^\prime) F_{\pi}(a^\prime) }{p(x)}  - \frac{ \exp(\delta a) \pi(a \mid x)  \mathbbm{1}(X = x)\mathbbm{1}(a \leq a^\prime) F_{\pi}(a^\prime)}{p(x)} 
\end{align*}
and
\begin{align*}
    (ii) &= \exp(\delta a) \pi(a \mid x) \mathbbm{1}(a \leq a^\prime) \sum_{t \leq a^\prime} \frac{\mathbbm{1}(X = x)}{p(x)} \Big( \mathbbm{1}(A = t) - \pi(t \mid x) \Big) \\
    &= \exp(\delta a) \pi(a \mid x) \mathbbm{1}(a \leq a^\prime)  \frac{\mathbbm{1}(X = x)}{p(x)} \Big( \mathbbm{1}(A \leq a^\prime) -  F_{\pi}(a^\prime) \Big),
\end{align*}
so, combining terms, $ \mathbb{IF}(f(\pi(a \mid x)))$ simplifies to
\begin{align*}
     \frac{\exp(\delta a) \mathbbm{1}(X = x) \mathbbm{1}(a \leq a^\prime)}{p(x)} ( (\mathbbm{1}(A = a) - \pi(a \mid x)) F_{\pi}(a^\prime) + \pi(a \mid x) (\mathbbm{1}(A \leq a^\prime) -  F_{\pi}(a^\prime) ) ).
\end{align*}
Similarly, we can see that $\mathbb{IF}(g(\pi(a \mid x)))$ is equal to
\begin{align*}
     \frac{\exp(-\delta a) \mathbbm{1}(X = x) \mathbbm{1}(a > a^\prime)}{p(x)} ((\mathbbm{1}(A = a) - \pi(a \mid x)) F^c_{\pi}(a^\prime) + \pi(a \mid x) (\mathbbm{1}(A > a^\prime) - F^c_{\pi}(a^\prime)))
\end{align*}
where $F^c_{\pi}(a^\prime) := 1 - F_{\pi}(a^\prime) = \sum_{t > a^\prime} \pi(t \mid x)$. Finally, note that
\begin{align*}
     \mathbb{IF}\left( \left[\sum_{t \leq a^\prime} \text{exp}(\delta t) \pi(t \mid x) \right]^{-1} \right) = -\frac{\sum_{t \leq a^\prime} \text{exp}(\delta t) \frac{\mathbbm{1}(X=x)}{p(x)} \left( \mathbbm{1}(A = t) - \pi(t \mid x) \right)}{\left(\sum_{t \leq a^\prime} \text{exp}(\delta t) \pi(t \mid x) \right)^2}
\end{align*}
and
\begin{align*}
    \mathbb{IF}\left( \left[\sum_{t > a^\prime} \text{exp}(- \delta t) \pi(t \mid x) \right]^{-1} \right) = -\frac{\sum_{t > a^\prime} \text{exp}(-\delta t) \frac{\mathbbm{1}(X=x)}{p(x)} \left( \mathbbm{1}(A = t) - \pi(t \mid x) \right)}{\left(\sum_{t > a^\prime} \text{exp}(- \delta t) \pi(t \mid x) \right)^2}.
\end{align*}
With these derivations in place, we can now evaluate each term in the decomposition of $\sum_x \sum_a  \mu(x, a)\mathbb{IF}(r_\delta(a \mid x)) p(x)$. First, observe that $T_1$ reduces to
\begin{align*}
    T_1 &= \Bigg\{\frac{\mu(X, A) \exp(\delta A) \mathbbm{1}(A \leq a^{\prime})  F_{\pi}(a^{\prime})}{\sum_{t \leq a^{\prime}} \text{exp}(\delta t) \pi(t \mid X)} - \frac{F_{\pi}(a^{\prime}) \sum_{a \leq a^{\prime}} \mu(X, a) \exp(\delta a) \pi(a \mid X)  }{\sum_{t \leq a^{\prime}} \text{exp}(\delta t) \pi(t \mid X)} \ +\\
    &\phantom{{}={\Bigg\{}}\frac{ \mathbbm{1}(A \leq a^{\prime})\sum_{a \leq a^{\prime}} \mu(X, a) \exp(\delta a) \pi(a \mid X)}{\sum_{t \leq a^{\prime}} \text{exp}(\delta t) \pi(t \mid X)} - \frac{F_{\pi}(a^{\prime}) \sum_{a \leq a^{\prime}} \mu(X, a) \exp(\delta a) \pi(a \mid X)   }{\sum_{t \leq a^{\prime}} \text{exp}(\delta t) \pi(t \mid X)} \Bigg\}
\end{align*}
or equivalently, to
\begin{align*}
    \frac{\mu(X, A) \exp(\delta A) \mathbbm{1}(A \leq a^{\prime})  F_{\pi}(a^{\prime})}{\sum_{t \leq a^{\prime}} \text{exp}(\delta t) \pi(t \mid X)} + \frac{\mathbbm{1}(A \leq a^{\prime})}{F_{\pi}(a^{\prime})} \sum_{a \leq a^{\prime}} \mu(X, a) r_\delta(a \mid X) - 2\sum_{a \leq a^{\prime}} \mu(X, a) r_\delta(a \mid X)
\end{align*}
where we have used the fact that
\begin{align*}
    \frac{F_\pi(a^{\prime}) \sum_{a \leq a^{\prime}}\mu(X, a) \text{exp}(\delta a) \pi(a \mid X)}{\left(\sum_{t \leq a^{\prime}} \text{exp}(\delta t) \pi(t \mid X) \right)} = \sum_{a \leq a^{\prime}} \mu(X, a) r_\delta(a \mid X).
\end{align*}
and
\begin{align*}
    \frac{\mathbbm{1}(A \leq a^{\prime})\sum_{a \leq a^{\prime}} \mu(X, a) \exp(\delta a) \pi(a \mid X) }{\sum_{t \leq a^{\prime}} \text{exp}(\delta t) \pi(t \mid X)} = \frac{\mathbbm{1}(A \leq a^{\prime})}{F_{\pi}(a^{\prime})} \sum_{a \leq a^{\prime}} \mu(X, a) r_\delta(a \mid X).
\end{align*}
Next, observe that
\begin{align*}
   T_2 &= -  \frac{F_\pi(a^{\prime}) \sum_{a \leq a^{\prime}}\mu(X, a) \text{exp}(\delta a) \pi(a \mid X) }{\left(\sum_{t \leq a^{\prime}} \text{exp}(\delta t) \pi(t \mid X) \right)^2}\left( \text{exp}(\delta A) \mathbbm{1}(A \leq a^{\prime}) - \sum_{t \leq a^{\prime}} \text{exp}(\delta t) \pi(t \mid X)  \right) \\
   &= -\frac{ \text{exp}(\delta A) \mathbbm{1}(A \leq a^{\prime}) }{\left(\sum_{t \leq a^{\prime}} \text{exp}(\delta t) \pi(t \mid X) \right)} \cdot  \sum_{a \leq a^{\prime}} \mu(X, a) r_\delta(a \mid X)  +  \sum_{a \leq a^{\prime}} \mu(X, a) r_\delta(a \mid X)
\end{align*}
Similarly, it follows that $T_3$ simplifies to
\begin{align*}
    \frac{\mu(X, A) \exp(- \delta A) \mathbbm{1}(A > a^{\prime})  F^c_{\pi}(a^{\prime})}{\sum_{t > a^{\prime}} \text{exp}(- \delta t) \pi(t \mid X)} + \frac{\mathbbm{1}(A > a^{\prime})}{F^c_{\pi}(a^{\prime})} \sum_{a > a^{\prime}} \mu(X, a) r_\delta(a \mid X) - 2\sum_{a > a^{\prime}} \mu(X, a) r_\delta(a \mid X)
\end{align*}
and that
\begin{align*}
   T_4 = -\frac{ \text{exp}(- \delta A) \mathbbm{1}(A > a^{\prime}) }{\left(\sum_{t > a^{\prime}} \text{exp}(- \delta t) \pi(t \mid X) \right)} \cdot  \sum_{a > a^{\prime}} \mu(X, a) r_\delta(a \mid X)  +  \sum_{a > a^{\prime}} \mu(X, a) r_\delta(a \mid X).
\end{align*}
Taking stock of what we have, it follows that $\sum_x \sum_a \mu(x, a) \mathbb{IF}( r_\delta(a \mid x))  p(x)$ is given by
\begin{align*}
    &\phantom{{}={}}\Bigg\{\frac{r_\delta(A \mid X)}{\pi(A \mid X)} \mu(X, A) -  \sum_{a} \mu(X, a) r_\delta(a \mid X) \ + \\
    &\phantom{{}={\Bigg\{}}\frac{\mathbbm{1}(A \leq a^{\prime})}{F_{\pi}(a^{\prime})} \sum_{a \leq a^{\prime}} \mu(X, a) r_\delta(a \mid X) -\frac{ \text{exp}(\delta A) \mathbbm{1}(A \leq a^{\prime}) }{\left(\sum_{t \leq a^{\prime}} \text{exp}(\delta t) \pi(t \mid X) \right)} \sum_{a \leq a^{\prime}} \mu(X, a) r_\delta(a \mid X) \ + \\
      &\phantom{{}={\Bigg\{}} \frac{\mathbbm{1}(A > a^{\prime})}{F^c_{\pi}(a^{\prime})} \sum_{a > a^{\prime}} \mu(X, a) r_\delta(a \mid X) -\frac{ \text{exp}(- \delta A) \mathbbm{1}(A > a^{\prime}) }{\left(\sum_{t > a^{\prime}} \text{exp}(- \delta t) \pi(t \mid X) \right)}  \sum_{a > a^{\prime}} \mu(X, a) r_\delta(a \mid X) \Bigg\}
\end{align*}
where we have grouped together the terms
\begin{align*}
     \frac{F_\pi(a^{\prime})\mu(X, A) \exp(\delta a) \mathbbm{1}(A \leq a^{\prime})}{\sum_{t \leq a^{\prime}} \text{exp}(\delta t) \pi(t \mid X)} +  \frac{F^c_\pi(a^{\prime})\mu(X, A) \exp(- \delta A) \mathbbm{1}(A > a^{\prime})}{\sum_{t > a^{\prime}} \text{exp}(-\delta t) \pi(t \mid X)} = \frac{r_\delta(A \mid X)}{\pi(A \mid X)} \mu(X, A).
\end{align*}
Finally, observe that
\begin{align*}
    \frac{\mathbbm{1}(A \leq a^{\prime})}{F_{\pi}(a^{\prime})} \sum_{a \leq a^{\prime}} \mu(X, a) r_\delta(a \mid X) &= \mathbbm{1}(A \leq a^{\prime}) \left( \frac{\sum_{a} \mu(X, a) r_\delta(a \mid X) \mathbbm{1}(a \leq a^{\prime})}{\mathbb{P}(A \leq a^{\prime} \mid X)} \right) \\
    &= \mathbbm{1}(A \leq a^{\prime}) \mathbb{E}_R\left( \mu(X, A) \mid X, A \leq a^{\prime} \right)
\end{align*}
and that
\begin{align*}
     \frac{\mathbbm{1}(A \leq a^{\prime})}{F_{\pi}(a^{\prime})} \frac{r_\delta(A \mid X)}{\pi(A \mid X)} \sum_{a \leq a^{\prime}} \mu(X, a) r_\delta(a \mid X) &= \mathbbm{1}(A \leq a^{\prime}) \frac{r_\delta(A \mid X)}{\pi(A \mid X)} \mathbb{E}_R\left( \mu(X, a) \mid X, A \leq a^{\prime} \right)
\end{align*}
where we say $\mathbb{E}_R(\cdot)$ represents an expectation with respect to the reflected distribution $r_\delta$. Thus, combining complementary terms and putting everything together, we have that our candidate influence function is given by $\varphi_R(Z;\delta) = D_Y + D_{r, \mu} + D_\psi$, where
\begin{align*}
    D_Y &= \frac{r_\delta(A \mid X)}{\pi(A \mid X)} \Big(Y - \mu(X, A) \Big) \\
    D_{r, \mu} &= \frac{r_\delta(A \mid X)}{\pi(A \mid X)} \mu(X, A) + \Big(\mathbbm{1}(A \leq a^{\prime})L(X, A) + \mathbbm{1}(A > a^{\prime}) U(X, A) \Big) - \mathbb{E}_R(\mu(X, A) \mid X) \\
    D_\psi &= \mathbb{E}_R(\mu(X, A) \mid X)  - \psi_R(\delta)
\end{align*}
and $\mathbb{E}_R(\mu(X, A) \mid X) = \int_a \mu(X, a) r_\delta(a \mid X) \ da$ is the conditional mean of $\mu(X, A)$ under the reflected exponentially tilted distribution such that
\begin{align*}
    L(X, A) &= \mathbb{E}_R(\mu(X, A) \mid X, A \leq a^{\prime}) - \frac{r_\delta(A \mid X)}{\pi(A \mid X)} \mathbb{E}_R(\mu(X, A) \mid X, A \leq a^{\prime}) \\
    U(X, A) &= \mathbb{E}_R(\mu(X, A) \mid X, A > a^{\prime}) - \frac{r_\delta(A \mid X)}{\pi(A \mid X)} \mathbb{E}_R(\mu(X, A) \mid X, A > a^{\prime}).
\end{align*}
Now, we show that the remainder from the von Mises expansion of our candidate influence function can be written as a second order product of errors. Observe that the remainder $\psi_R(\widehat{\mathbb{P}}) - \psi_R(\mathbb{P}) + \int \varphi_R(z; \widehat{\mathbb{P}}) \ d\mathbb{P}(z)$ can be split up into the following terms:
\begin{align*}
    (i) &= \mathbb{E}\left[\frac{\widehat{r}_\delta(A \mid X)}{\widehat{\pi}(A \mid X)} \Big(Y - \widehat{\mu}(X, A) \Big) \right] \\
    (ii) &=  \mathbb{E}\Bigg[\frac{\widehat{r}_\delta(A \mid X)}{\widehat{\pi}(A \mid X)} \widehat{\mu}(X, A) + \mathbbm{1}(A \leq a^{\prime}) \widehat{L}(X, A) + \mathbbm{1}(A > a^{\prime}) \widehat{U}(X, A)  -  \int_a \widehat{\mu}(X, a) \widehat{r}_\delta(a \mid X)da \Bigg] \\
    (iii) &= \mathbb{E}\left[\int_a \widehat{\mu}(X, a) \widehat{r}_\delta(a \mid X)da  - \psi_R(\mathbb{P}) \right].
\end{align*}
Evaluating the first expectation we can see that
\begin{align*}
    (i) &= \mathbb{E}\left[ \mathbb{E}\left[\frac{\widehat{r}_\delta(A \mid X)}{\widehat{\pi}(A \mid X)} \Big(Y - \widehat{\mu}(X, A) \Big) \mid X, A \right] \right] \\
    &= \mathbb{E}\left[\frac{\widehat{r}_\delta(A \mid X)}{\widehat{\pi}(A \mid X)} \Big(\mu(X, A) - \widehat{\mu}(X, A) \Big) \right] \\
    &= \mathbb{E}\left[\int_a \frac{\widehat{r}_\delta}{\widehat{\pi}_a} \Big(\mu_a - \widehat{\mu}_a \Big) \pi_a da\right]
\end{align*}
where we use the shorthand $\mu_a = \mu(X, a)$, $\pi_a = \pi(a \mid X)$ and $r_\delta = r_\delta(a \mid X)$. Next, within $(ii)$ observe that by applying the law of iterated expectations,
\begin{align*}
    \mathbb{E}\left[\mathbbm{1}(A \leq a^{\prime}) \widehat{L}(X, A) \right] &= \mathbb{E}\left[ \frac{\mathbbm{1}(A \leq a^{\prime}) }{\int_{a \leq a^{\prime}}  \widehat{\pi}(a \mid X) da} \left( \int_{a \leq a^{\prime}}  \widehat{\mu}_a \widehat{r}_\delta da - \frac{\widehat{r}(A \mid X)}{\widehat{\pi}(A \mid X)} \int_{a \leq a^{\prime}} \widehat{\mu}_a \widehat{r}_\delta da \right) \right] \\
    &=  \mathbb{E}\left[ \frac{\int_{a \leq a^{\prime}}  \pi_a da}{\int_{a \leq a^{\prime}}  \widehat{\pi}_a da}  \int_{a \leq a^{\prime}}  \widehat{\mu}_a \widehat{r}_\delta da - \frac{\int_{a \leq a^{\prime}} \frac{\widehat{r}  }{\widehat{\pi}_a} \pi_a da \int_{a \leq a^{\prime}} \widehat{\mu}_a \widehat{r}_\delta da}{\int_{a \leq a^{\prime}}  \widehat{\pi}_a da} \right] 
\end{align*}
and similarly,
\begin{align*}
     \mathbb{E}\left[\mathbbm{1}(A > a^{\prime}) \widehat{U}(X, A) \right] &= \mathbb{E}\left[ \frac{\int_{a > a^{\prime}}  \pi_a da}{\int_{a > a^{\prime}}  \widehat{\pi}_a da}  \int_{a > a^{\prime}}  \widehat{\mu}_a \widehat{r}_\delta da - \frac{\int_{a > a^{\prime}} \frac{\widehat{r}  }{\widehat{\pi}_a} \pi_a da \int_{a > a^{\prime}} \widehat{\mu}_a \widehat{r}_\delta da}{\int_{a > a^{\prime}}  \widehat{\pi}_a da} \right] .
\end{align*}
Putting together both terms, it follows that
\begin{align*}
    (ii) &= \Bigg\{ \mathbb{E} \left[ \int_a \frac{\widehat{r}}{\widehat{\pi}_a} \widehat{\mu}_a \pi_a da + \frac{ \int_{a \leq a^{\prime}}  \pi_a da}{\int_{a \leq a^{\prime}}  \widehat{\pi}_a da}  \int_{a \leq a^{\prime}}  \widehat{\mu}_a \widehat{r}_\delta da -  \frac{\int_{a \leq a^{\prime}}  \widehat{\mu}_a \widehat{r}_\delta da}{\int_{a \leq a^{\prime}}  \widehat{\pi}_a da} \int_{a \leq a^{\prime}} \frac{\widehat{r}  }{\widehat{\pi}_a} \pi_a da \right]\ + \\
    &\phantom{{}={\Bigg\{ }} \mathbb{E}\left[\frac{ \int_{a > a^{\prime}}  \pi_a da}{\int_{a > a^{\prime}}  \widehat{\pi}_a da}  \int_{a > a^{\prime}}  \widehat{\mu}_a \widehat{r}_\delta da -  \frac{\int_{a > a^{\prime}}  \widehat{\mu}_a \widehat{r}_\delta da }{\int_{a > a^{\prime}}  \widehat{\pi}_a da} \int_{a > a^{\prime}} \frac{\widehat{r}  }{\widehat{\pi}_a} \pi_a da   -  \int_a \widehat{\mu}_a \widehat{r}_\delta da  \right] \Bigg\} .
\end{align*}
Then, combining all three terms, we can see that $R_2(\widehat{\mathbb{P}}, \mathbb{P})$ can be written as
\begin{align*}
    R_2(\widehat{\mathbb{P}}, \mathbb{P}) &=  \Bigg\{ \mathbb{E}\left[\int_a \frac{\widehat{r}_\delta}{\widehat{\pi}_a} \Big(\mu_a - \widehat{\mu}_a \Big) \pi_a da\right] + \\
    &\phantom{{}={\Bigg\{}} \mathbb{E} \left[ \int_a \frac{\widehat{r}_{\delta}}{\widehat{\pi}_a} \widehat{\mu}_a \pi_a da + \frac{ \int_{a \leq a^{\prime}}  \pi_a da}{\int_{a \leq a^{\prime}}  \widehat{\pi}_a da}  \int_{a \leq a^{\prime}}  \widehat{\mu}_a \widehat{r}_\delta da -  \frac{\int_{a \leq a^{\prime}}  \widehat{\mu}_a \widehat{r}_\delta da}{\int_{a \leq a^{\prime}}  \widehat{\pi}_a da} \int_{a \leq a^{\prime}} \frac{\widehat{r}_{\delta}  }{\widehat{\pi}_a} \pi_a da \right]\ + \\
    &\phantom{{}={\Bigg\{}} \mathbb{E}\left[\frac{ \int_{a > a^{\prime}}  \pi_a da}{\int_{a > a^{\prime}}  \widehat{\pi}_a da}  \int_{a > a^{\prime}}  \widehat{\mu}_a \widehat{r}_\delta da -  \frac{\int_{a > a^{\prime}}  \widehat{\mu}_a \widehat{r}_\delta da}{\int_{a > a^{\prime}}  \widehat{\pi}_a da} \int_{a > a^{\prime}} \frac{\widehat{r}_{\delta}  }{\widehat{\pi}_a} \pi_a da   -  \int_a \widehat{\mu}_a \widehat{r}_\delta da  \right] + \\
    &\phantom{{}={\Bigg\{}} \mathbb{E}\left[\int_a \widehat{\mu}_a \widehat{r}_\delta da  - \psi_R(\mathbb{P}) \right] \Bigg\} .
\end{align*}
To proceed, we split the remainder into terms such that $a \leq a^{\prime}$ and $a > a^{\prime}$. That is, by considering the following decompositions:
\begin{align*}
    \int_a \frac{\widehat{r}_\delta}{\widehat{\pi}_a} \Big(\mu_a - \widehat{\mu}_a \Big) \pi_a da &= \int_{a \leq a^{\prime}} \frac{\widehat{r}_\delta}{\widehat{\pi}_a} \Big(\mu_a - \widehat{\mu}_a \Big) \pi_a da + \int_{a > a^{\prime}} \frac{\widehat{r}_\delta}{\widehat{\pi}_a} \Big(\mu_a - \widehat{\mu}_a \Big) \pi_a da \\
     \int_a \frac{\widehat{r}_{\delta}}{\widehat{\pi}_a} \widehat{\mu}_a \pi_a da &=  \int_{a \leq a^{\prime}} \frac{\widehat{r}_{\delta}}{\widehat{\pi}_a} \widehat{\mu}_a \pi_a da +  \int_{a > a^{\prime}} \frac{\widehat{r}_{\delta}}{\widehat{\pi}_a} \widehat{\mu}_a \pi_a da \\
     \int_a \widehat{\mu}_a \widehat{r}_\delta da &= \int_{a \leq a^{\prime}} \widehat{\mu}_a \widehat{r}_\delta da + \int_{a > a^{\prime}} \widehat{\mu}_a \widehat{r}_\delta da \\
     \psi_R(\mathbb{P}) &= \mathbb{E}\left[ \int_{a \leq a^{\prime}} \mu_a r_{\delta} da \right] + \mathbb{E}\left[ \int_{a > a^{\prime}} \mu_a r_{\delta} da \right].
\end{align*}
We proceed first with the terms where $a \leq a^{\prime}$. Here, we have
\begin{align*}
   L_R &=  \mathbb{E} \Bigg[ \int_{a \leq a^{\prime}} \frac{\widehat{r}_\delta}{\widehat{\pi}_a} \mu_a \pi_a da  + \underbrace{\frac{ \int_{a \leq a^{\prime}}  \pi_a da \int_{a \leq a^{\prime}}  \widehat{\mu}_a \widehat{r}_\delta da}{\int_{a \leq a^{\prime}}  \widehat{\pi}_a da}   -  \frac{\int_{a \leq a^{\prime}}  \widehat{\mu}_a \widehat{r}_\delta da \int_{a \leq a^{\prime}} \frac{\widehat{r}_{\delta}  }{\widehat{\pi}_a} \pi_a da}{\int_{a \leq a^{\prime}}  \widehat{\pi}_a da} }_{(iv)} -\int_{a \leq a^{\prime}} \mu_a r_{\delta} da \Bigg]
\end{align*}
Next, we make a few algebraic manipulations. Observe that the difference between the first and last terms in $L_R$, $\int_{a \leq a^{\prime}} \frac{\widehat{r}_\delta}{\widehat{\pi}_a} \mu_a \pi_a da -\int_{a \leq a^{\prime}} \mu_a r_{\delta} da $, may be written as
\begin{align*}
     \int_{a \leq a^{\prime}} \exp(\delta a) \mu_a \pi_a \left( \frac{\int_{a \leq a^{\prime}} \widehat{\pi}_a da }{\int_{a \leq a^{\prime}} \exp(\delta a) \widehat{\pi}_a da } - \frac{\int_{a \leq a^{\prime}} \pi_a da}{\int_{a \leq a^{\prime}} \exp(\delta a) \pi_a da }\right) da.
\end{align*}
Furthermore, it follows that
\begin{align*}
     (iv) = \frac{\int_{a \leq a^{\prime}} \exp(\delta a) \widehat{\mu}_a \widehat{\pi}_a da}{\left(\int_{a \leq a^{\prime}} \exp(\delta a) \widehat{\pi}_a da \right)^2}\left[ \int_{a \leq a^{\prime}} \exp(\delta a) \widehat{\pi}_a da \int_{a \leq a^{\prime}} \pi_a da  -  \int_{a \leq a^{\prime}} \exp(\delta a) \pi_a da \int_{a \leq a^{\prime}} \widehat{\pi}_a da  \right].
\end{align*}
Now, let $u = \int_{a \leq a^{\prime}} \exp(\delta a) \mu_a \pi_a da$, $v = \int_{a \leq a^{\prime}} \exp(\delta a) \pi_a da$, and $w = \int_{a \leq a^{\prime}} \pi_a da$. Then, it follows that
\begin{align*}
    L_R &= u \left(\frac{\widehat{w}}{\widehat{v}} - \frac{w}{v} \right) + \frac{\widehat{u}}{\widehat{v}^2}\Big( \widehat{v} w - v \widehat{w} \Big) \\
    &= (u + \textcolor{blue}{\widehat{u}} - \textcolor{red}{\widehat{u}}) \left(\frac{\widehat{w}}{\widehat{v}} - \frac{w}{v} \right) + \frac{\widehat{u}}{\widehat{v}^2}\Big( \widehat{v} w - v \widehat{w} \Big) \\
    &= (u - \widehat{u}) \left(\frac{\widehat{w}}{\widehat{v}} - \frac{w}{v} \right) + \widehat{u} \left(\frac{\widehat{w}}{\widehat{v}} - \frac{w}{v} \right) + \frac{\widehat{u}}{\widehat{v}^2}\Big( \widehat{v} w - v \widehat{w} \Big) \\
    &= (u - \widehat{u}) \left(\frac{\widehat{w}}{\widehat{v}} - \frac{w}{v} \right) + \widehat{u} \left(\frac{v\widehat{w} - \widehat{v} w}{v\widehat{v}} \right) - \frac{\widehat{u}}{\widehat{v}^2}\Big(v\widehat{w} -\widehat{v} w \Big) \\
    &=  (u - \widehat{u}) \left(\frac{\widehat{w}}{\widehat{v}} - \frac{w}{v} \right) + \widehat{u} \Big( v\widehat{w} -\widehat{v} w \Big)
 \left(\frac{1}{v\widehat{v}} - \frac{1}{\widehat{v}^2} \right) \\
 &=  (u - \widehat{u}) \left(\frac{\widehat{w}}{\widehat{v}} - \frac{w}{v} \right) + \frac{\widehat{u}}{v \widehat{v}^2} \Big( v\widehat{w} -\widehat{v} w \Big)
 \left(\widehat{v} - v \right) \\
 &:= L_1 + L_2,
\end{align*}
thus yielding a second order product of errors. However, we may simplify this expression further by making a few more manipulations. First, note by adding and subtracting $\pi_a \widehat{\mu}_a$ it follows that
\begin{align*}
    L_1 &= \left(\int_{a \leq a^{\prime}} \exp(\delta a) \Big(\mu_a \pi_a  - \widehat{\mu}_a \widehat{\pi}_a \Big) da \right) \left( \frac{\int_{a \leq a^{\prime}} \widehat{\pi}_a da}{\int_{a \leq a^{\prime}} \exp(\delta a) \widehat{\pi}_a da} - \frac{\int_{a \leq a^{\prime}} \pi_a da}{\int_{a \leq a^{\prime}} \exp(\delta a) \pi_a da}  \right) \\
    &= \int_{a \leq a^{\prime}}\left(\frac{\widehat{r}_\delta}{\widehat{\pi}_a} - \frac{r_\delta}{\pi_a} \right)\Big( (\mu_a - \widehat{\mu}_a)\pi_a + (\pi_a - \widehat{\pi}_a) \widehat{\mu}_a \Big) da.
\end{align*}
Next, we can see that
\begin{align*}
    v\widehat{w} -\widehat{v} w &= \int_{a \leq a^{\prime}} \exp(\delta a) \pi_a da \left(\int_{a \leq a^{\prime}} \widehat{\pi}_a da\right) - \int_{a \leq a^{\prime}} \exp(\delta a) \widehat{\pi}_a da \left(\int_{a \leq a^{\prime}} \pi_a da\right) \\
    &= \int_{a \leq a^{\prime}} \exp(\delta a) \left[\pi_a \left(\int_{a \leq a^{\prime}} \widehat{\pi}_a da\right) - \widehat{\pi}_a \left(\int_{a \leq a^{\prime}} \pi_a da\right) \right] da \\
    &= \int_{a \leq a^{\prime}} \exp(\delta a) \left[\pi_a \int_{a \leq a^{\prime}} (\widehat{\pi}_a - \pi_a) da - (\widehat{\pi}_a - \pi_a) \left(\int_{a \leq a^{\prime}} \pi_a da\right) \right] da \\
    &=  \int_{a \leq a^{\prime}} \exp(\delta a) \pi_a da \int_{a \leq a^{\prime}} (\widehat{\pi}_a - \pi_a) da - \int_{a \leq a^{\prime}} \exp(\delta a) (\widehat{\pi}_a - \pi_a) \left(\int_{a \leq a^{\prime}} \pi_a da\right) da
\end{align*}
and
\begin{align*}
    \widehat{v} - v = \int_{a \leq a^{\prime}} \exp(\delta a) (\widehat{\pi}_a - \pi_a) da.
\end{align*}
Thus, combining both terms it follows that
\begin{align*}
    L_2 &= \frac{\widehat{\mathbb{E}}_R\left(\mu(X, A) \mid X, A \leq a^{\prime} \right)}{\widehat{\mathbb{P}}(A \leq a^{\prime} \mid X)}\left( \int_{a \leq a^{\prime}} (\widehat{\pi}_a - \pi_a)\left( 1 - \frac{r_\delta}{\pi_a} \right)  da \right)\int_{a \leq a^{\prime}} \frac{\widehat{r}_\delta}{\widehat{\pi}_a} (\widehat{\pi}_a - \pi_a) da.
\end{align*}
Following the same steps for $a > a^{\prime}$, we find analogous remainder terms of
\begin{align*}
    U_1 &= \int_{a > a^{\prime}}\left(\frac{\widehat{r}_\delta}{\widehat{\pi}_a} - \frac{r_\delta}{\pi_a} \right)\Big( (\mu_a - \widehat{\mu}_a)\pi_a + (\pi_a - \widehat{\pi}_a) \widehat{\mu}_a \Big) da
\end{align*}
and
\begin{align*}
    U_2 = \frac{\widehat{\mathbb{E}}_R\left(\mu(X, A) \mid X, A > a^{\prime} \right)}{\widehat{\mathbb{P}}(A > a^{\prime} \mid X)}\left( \int_{a > a^{\prime}} (\widehat{\pi}_a - \pi_a)\left( 1 - \frac{r_\delta}{\pi_a} \right)  da \right)\int_{a > a^{\prime}} \frac{\widehat{r}_\delta}{\widehat{\pi}_a} (\widehat{\pi}_a - \pi_a) da.
\end{align*}
Putting everything together, it follows that $R_2(\widehat{\mathbb{P}}, \mathbb{P}) = R_1 + R_2 + R_3$ where
\begin{align*}
     R_1 &= \mathbb{E} \left[ \int_{a}\left(\frac{\widehat{r}_\delta}{\widehat{\pi}_a} - \frac{r_\delta}{\pi_a} \right)\Big( (\mu_a - \widehat{\mu}_a)\pi_a + (\pi_a - \widehat{\pi}_a) \widehat{\mu}_a \Big) da \right] \\
    R_2 &= \mathbb{E}\left[\frac{\widehat{\mathbb{E}}_R\left(\mu(X, A) \mid X, A \leq a^{\prime} \right)}{\widehat{\mathbb{P}}(A \leq a^{\prime} \mid X)}\left( \int_{a \leq a^{\prime}} (\widehat{\pi}_a - \pi_a)\left( 1 - \frac{r_\delta}{\pi_a} \right)  da \right)\int_{a \leq a^{\prime}} \frac{\widehat{r}_\delta}{\widehat{\pi}_a} (\widehat{\pi}_a - \pi_a) da \right] \\
    R_3 &= \mathbb{E}\left[ \frac{\widehat{\mathbb{E}}_R\left(\mu(X, A) \mid X, A > a^{\prime} \right)}{\widehat{\mathbb{P}}(A > a^{\prime} \mid X)}\left( \int_{a > a^{\prime}} (\widehat{\pi}_a - \pi_a)\left( 1 - \frac{r_\delta}{\pi_a} \right)  da \right)\int_{a > a^{\prime}} \frac{\widehat{r}_\delta}{\widehat{\pi}_a} (\widehat{\pi}_a - \pi_a) da\right].
 \end{align*}
\end{proof}

\subsection{Proof of \texorpdfstring{\cref{dose_response_theorem}}{Theorem 5}} \label{dose_response_theorem_proof}

We split the proof of \cref{dose_response_theorem} into two parts. First, a technical lemma that establishes $\psi_R(\delta; a^\prime)$ is converging to $\mathbb{E}[Y^{a^\prime}]$ as $\delta \to \infty$. Then, after establishing this result we make use of \cref{asymptotic_normality} in combination with the technical lemma to complete the proof.

\begin{lemma} \label{dose_response_lemma}
Assume $A \in [0,1]$ and:
\begin{enumerate}
    \item[(i)] $\pi_{\min} \leq \pi(a \mid x)$ for all $a \in  [a^{\prime} - \eta, a^\prime + \eta]$ for some $\eta \in (0, \text{min}\{a^\prime, 1 - a^\prime\}]$.
    \item[(ii)] $\pi(a \mid x) \leq \pi_{\max}$ for all $a$, $x$.
    \item[(iii)] $|\mu(x,a) - \mu(x,a^{\prime})| \leq L|a-a^{\prime}|$ for any $x$. 
\end{enumerate}
Then, for any $a^\prime \in [0, 1]$,
\begin{align*}
    | \psi_R(\delta; a^\prime) - \E(Y^{a^\prime}) | \leq \frac{1}{\delta} \left( \frac{L \pi_{\max}}{\pi_{\min}(1 - \eta)} \right).
\end{align*}
\end{lemma}

\begin{proof}[\textbf{Proof:}]
First, let us consider the case where $a = 1$. Then, it follows that
\begin{align*}
     | \psi(\delta) - \E(Y^1)| &=  \left| \int_x \int_a \mu(x,a) q_\delta(a \mid x) \, da \, d\Pb(x) - \int_x \mu(x,1) \, d\Pb(x) \right| \\
     &\leq  \left| \int_x \int_a \left| \mu(x,a) - \mu(x, 1) \right| q_\delta(a \mid x) \, da \, d\Pb(x)\right| \\
     &\leq L \int_x \int_a |a-1| q_\delta(a \mid x) \, da \, d\Pb(x). 
\end{align*}
From here, note that
\begin{align*}
     \int_a (1-a) q_\delta(a \mid x) da &= \frac{\int^1_0 (1 - a) \exp(\delta a) \pi(a \mid x) da}{\int^1_0 \exp(\delta a) \pi(a \mid x) da } \\
     &\leq \frac{\int^1_0 (1 - a) \exp(\delta a) \pi(a \mid x) da}{\int^1_\eta \exp(\delta a) \pi(a \mid x) da } \\
     &\leq \frac{\pi_{\max}}{\pi_{\min}} \cdot  \frac{\int^1_0 (1 - a) \exp(\delta a)  da}{\int^1_\eta \exp(\delta a) da } \\
     &=  \frac{\pi_{\max}}{\pi_{\min}}\left( \frac{1}{\delta} \cdot \frac{\exp(\delta)-1-\delta}{\exp(\delta)- \exp(\eta \delta)} \right)
\end{align*}
and furthermore, that
\begin{align*}
    \frac{\exp(\delta)-1-\delta}{\exp(\delta)- \exp(\eta \delta)} = \frac{1 -\exp(-\delta)(\delta + 1)}{1 - \exp(-\delta(1 - \eta))} \leq \frac{1}{1 - \eta}.
\end{align*}
Putting everything together yields the upper bound
\begin{align*}
    | \psi(\delta) - \E(Y^1)| \leq \frac{1}{\delta} \left( \frac{L \pi_{\max}}{\pi_{\min}(1 - \eta)} \right)
\end{align*}
Note that the case where $a^\prime = 0$ follows identically for $\delta \to -\infty$ assume weak positivity holds in a neighborhood of zero. Next, we consider some interior point $a^\prime \in [0, 1]$ using the reflected incremental effect. Recall that the reflected exponential tilt is defined as
\begin{align*}
    r_\delta(a \mid x) =  \omega_{a^\prime}\frac{\text{exp}(\delta a) \pi(a \mid x) \mathbbm{1}(a \leq a^\prime) }{\int^{a^\prime}_0 \text{exp}(\delta t) \pi(t \mid x) dt} +  (1-\omega_{a^\prime})\frac{\text{exp}(-\delta a) \pi(a \mid x) \mathbbm{1}(a > a^\prime)}{\int^{1}_{a^\prime} \text{exp}(-\delta t) \pi(t \mid x) dt}
\end{align*}
where $\omega_{a^\prime} = \int^{a^{\prime}}_0 \pi(t \mid x) dt$ and the reflected incremental effect as
\begin{align*}
    \psi_R(\delta; a^\prime) =  \int_\mathcal{X} \int_\mathcal{A} \mu(x, a) r_\delta(a \mid x) \, da \, d\mathbb{P}(x).
\end{align*}
Following the same style argument as before, we can see that
\begin{align*}
     | \psi_R(\delta, a^\prime) - \E(Y^{a^\prime})| &=  \left| \int_x \int_a \mu(x,a) r_\delta(a \mid x) \, da \, d\Pb(x) - \int_x \mu(x,a^\prime) \, d\Pb(x) \right| \\
     &\leq  \left| \int_x \int_a \left  | \mu(x,a) - \mu(x, a^\prime) \right| r_\delta(a \mid x) \, da \, d\Pb(x)\right| \\
     &\leq L \int_x \int_a |a-a^\prime| r_\delta(a \mid x) \, da \, d\Pb(x) \\
     &= \Bigg\{ L \int_x \int_a \left( \mathbb{P}(A \leq a^\prime \mid x) \frac{\int^{a^\prime}_0 (a^\prime - a) \exp(\delta a) \pi(a \mid x) da}{\int^{a^\prime}_0 \exp(\delta t) \pi(t \mid x) dt} \ + \right. \\
     &\phantom{{}={\Bigg\{ }} \left. \mathbb{P}(A > a^\prime \mid x) \frac{\int^{a^\prime}_0 (a^\prime - a) \exp(\delta a) \pi(a \mid x) da}{\int^{a^\prime}_0 \exp(\delta t) \pi(t \mid x) dt} \right) \, da \, d\Pb(x).
\end{align*}
From here observe that after following identical bounds to the $a = 1$ case, it follows that
\begin{align*}
     \frac{\int^{a^\prime}_0 (a^\prime - a) \exp(\delta a) \pi(a \mid x) da}{\int^{a^\prime}_0 \exp(\delta t) \pi(t \mid x) dt} &\leq \frac{1}{\delta} \frac{\pi_{\max}}{\pi_{\min}(1 - \eta)} \quad \text{and} \\
     \frac{\int^{a^\prime}_0 (a^\prime - a) \exp(\delta a) \pi(a \mid x) da}{\int^{a^\prime}_0 \exp(\delta t) \pi(t \mid x) dt} &\leq \frac{1}{\delta} \frac{\pi_{\max}}{\pi_{\min}(1 - \eta)}.
\end{align*}
Therefore, in the case of the reflected incremental effect we again have that 
\begin{align*}
        | \psi_R(\delta, a^\prime) - \E(Y^{a^\prime})| \leq \frac{1}{\delta} \left( \frac{L \pi_{\max}}{\pi_{\min}(1 - \eta)} \right).
\end{align*}
\end{proof}

Now that we have established that $ | \psi_R(\delta, a^\prime) - \E(Y^{a^\prime})| = O_\mathbb{P}\left(1/ \delta \right)$ we proceed with the remainder of the proof. Let $\psi_R(\infty, a^\prime) = \E(Y^{a^\prime})$. Then, it follows that
\begin{align*}
    \left|\widehat{\psi}_R(\delta, a^\prime) - \psi_R(\infty, a^\prime) \right| &= \left|\widehat{\psi}_R(\delta, a^\prime) - \psi_R(\delta, a^\prime) + \psi_R(\delta, a^\prime) - \psi_R(\infty, a^\prime) \right| \\
    &\leq \left|\widehat{\psi}_R(\delta, a^\prime) - \psi_R(\delta, a^\prime) \right| + \Big|\psi_R(\delta, a^\prime) - \psi_R(\infty, a^\prime) \Big|. 
\end{align*}
By \cref{asymptotic_normality} (see the discussion in \cref{reflected_ie_section} for more details on how \cref{asymptotic_normality} applies for the reflected incremental effect estimator) it follows that $|\widehat{\psi}_R(\delta, a^\prime) - \psi_R(\delta, a^\prime)| = O_\mathbb{P}\left(\sqrt{1 / (n / \delta) } \right)$, so combining terms and applying \cref{dose_response_lemma} we have that
\begin{align*}
   \left|\widehat{\psi}_R(\delta, a^\prime) - \psi_R(\infty, a^\prime) \right| = O_\mathbb{P}\left(\frac{1}{\delta} + \sqrt{\frac{1}{n / \delta }} \right).
\end{align*}

\subsection{Proof of \texorpdfstring{\cref{sigma_delta_ratio}}{Corollary 1}} \label{sigma_delta_ratio_proof}
\begin{proof}[\textbf{Proof:}]
In this proof, we establish the limiting behavior of $\sigma^2_\delta / \delta$ as $\delta \to \infty$. In order to avoid writing a tediously long proof, we stick to just the incremental effects estimator. However, we note that the exact same result for a chosen $a^\prime$ holds (instead of for edge-of-support $a = 1$ case) when using the reflected incremental effect. This can be shown after considering the nonparametric efficiency bound for the reflected incremental effect in \cref{reflected_efficiency_bound}, which has a structure identical to that of \cref{nonpar_efficiency_bound}. Then, as in the proof of \cref{dose_response_theorem}, we can repeat analogous arguments for interior points. With that said, first recall that by \cref{nonpar_efficiency_bound} the nonparametric efficiency bound is given by
 \begin{align*}
     \mathbb{E}\left[\frac{q^2_\delta(A \mid X)}{\pi^2(A \mid X)} \Big( \var\left( Y \mid X, A \right) +  \big[ \mu(X, A) - \mathbb{E}_Q(\mu(X, A) \mid X) \big]^2 \Big) \right] + \var\Big( \mathbb{E}_Q\left[\mu(X, A) \mid X \right] \Big),
\end{align*}
which we split into three parts:
\begin{align*}
    \sigma_1^2 &= \int_a \frac{q_\delta(a \mid X)^2}{\pi(a \mid X)} \var\left( Y \mid X, a \right) da, \\
    \sigma_2^2 &= \int_a \frac{q_\delta(a \mid X)^2}{\pi(a \mid X)}\left( \mu(X, a) - \int_a q_\delta(a \mid X) \mu(X, a)da \right)^2 da, \quad \text{and}\\
    \sigma_3^2 &= \left(\int_a q_\delta(a \mid X) \mu(X, a)da  - \psi(\delta) \right)^2
\end{align*}
such that $\sigma^2_\delta = \mathbb{E}\left[\sigma_1^2 + \sigma_2^2 + \sigma_3^2 \right]$. To evaluate the limit as $\delta \to \infty$, we further decompose $\sigma^2_1$ into
\begin{align*}
    \frac{\sigma^2_1}{\delta} &= \frac{1}{\delta} \int_a \frac{q_\delta(a \mid X)^2}{\pi(a \mid X)} \var\left( Y \mid X, a \right) da \\
    &= \frac{1}{\delta} \frac{\int_a \exp(2 \delta a) \pi_a \var\left( Y \mid X, a \right) da}{\left(\int_a \exp(\delta a) \pi_a da \right)^2} \\
    &= \left[ \frac{\int_a \exp(2 \delta a) \pi_a\var(Y_a) da}{\int_a \exp(2 \delta a) \pi_a da }  \cdot \frac{\int_a \exp(2 \delta a) \pi_a  da}{ \int_a \exp(2 \delta a) da } \cdot \left( \frac{ \int_a \exp(\delta a) da }{\int_a  \exp(\delta a) \pi_a  da } \right)^2 \cdot \frac{1}{\delta} \frac{\int_a \exp(2 \delta a)  da}{\left( \int_a \exp(\delta a) da \right)^2} \right]
\end{align*}
where for notational simplicity we say $\var(Y_a) := \var(Y \mid X, a)$. We will evaluate each expression from left to right. Following a similar proof structure to that of \cref{dose_response_theorem} we can see that
\begin{align*}
     \left| \frac{\int_a \exp(2 \delta a) \pi_a \var(Y_a) da}{\int_a \exp(2 \delta a) \pi_a \ da } - \var(Y_1)\right| &= \left| \frac{\int_a \exp(2 \delta a) \pi_a (\var(Y_a) - \var(Y_1)) da}{\int_a \exp(2 \delta a) \pi_a da } \right| \\
     &\leq  \frac{\int_a \exp(2 \delta a) \pi_a \big|\var(Y_a) - \var(Y_1)\big| da}{\int_a \exp(2 \delta a) \pi_a da } \\
     &\leq \frac{L\int_a \exp(2 \delta a) \pi_a \big|a - 1 \big| da}{\int_a \exp(2 \delta a) \pi_a \ da } \\
     &\leq \frac{L \int_a \exp(2 \delta a) \pi_a \big|a - 1 \big| da}{\int^1_{1 - \gamma} \exp(2 \delta a) \pi_a da } \\
     &\leq \frac{L \pi_{\max}}{\pi_{\min}}\frac{ \int_a \exp(2 \delta a) \big|a - 1 \big| da}{\int^1_{1 - \gamma} \exp(2 \delta a) da } 
\end{align*}
where in the second inequality we have used that $\var(Y \mid X, A)$ is Lipschitz (with constant $L$), the third inequality holds for some $\gamma \in [0, 1)$, and the fourth inequality follows by the assumption that positivity holds in a neighborhood of $a = 1$, so we may use the lower bound $\pi_{\min} \leq \pi(a \mid x)$ for $a \in [\gamma, 1]$. From here, we may directly integrate to see that
\begin{align*}
    \frac{ \int_a \exp(2 \delta a)\big|a - 1 \big| da}{\int^1_{1 - \gamma} \exp(2 \delta a) da } &= \frac{(\exp(2 \delta) - 2\delta - 1) / 4 \delta^2}{[\exp(2 \delta) - \exp(2 \delta (1 - \gamma))] / 2 \delta } \\
    &= \frac{1}{2 \delta} \left( \frac{\exp(2 \delta) - 2\delta - 1}{\exp(2 \delta) - \exp(2 \delta (1 - \gamma))} \right) \\
    &\leq \frac{1}{2 \delta}.
\end{align*}
Therefore, it is clear that as $\delta \to \infty$ then
\begin{align*}
    \frac{\int_a \exp(2 \delta a) \pi_a \var(Y \mid X, a) \ da}{\int_a \exp(2 \delta a) \pi_a da } \longrightarrow \var(Y \mid X, 1).
\end{align*}
Now, we consider the other terms in our decomposition of $\sigma^2_1$. Observe that
\begin{align*}
    \left| \frac{ \int_a \exp(2 \delta a) \pi_a   da}{ \int_a \exp(2 \delta a)da } - \pi_1 \right| &= \left| \frac{\int_a  \exp(2 \delta a) (\pi_a - \pi_1)  da}{ \int_a \exp(2 \delta a) da }  \right| \\
    &\leq  \frac{\int_a  \exp(2 \delta a) \left|\pi_a - \pi_1 \right|  da}{ \int_a \exp(2 \delta a) da }   \\
    &\leq \frac{L \int_a  \exp(2 \delta a) |a - 1|  \ da}{ \int_a \exp(2 \delta a)da },
\end{align*}
which we have already shown converges to zero as $\delta \to \infty$. Following the same logic, it can also be shown that
\begin{align*}
    \left( \frac{ \int_a \exp(\delta a) da }{ \int_a  \exp(\delta a) \pi_a \ da } \right)^2 \longrightarrow \frac{1}{\pi(1 \mid X)^2}.
\end{align*}
Finally, for the last term we can see that 
\begin{align*}
    \underset{\delta \to \infty}{\text{lim}}\left\{ \frac{1}{\delta} \frac{\int_a \exp(2 \delta a)  da}{\left( \int_a \exp(\delta a)  da \right)^2} \right\} = \underset{\delta \to \infty}{\text{lim}}\left\{ \frac{1}{\delta} \left[ \frac{1/2\delta (\exp(2 \delta) 
- 1)}{1/\delta^2 (\exp(\delta) - 1)^2}\right] \right\} = \underset{\delta \to \infty}{\text{lim}}\left\{ \frac{1}{2} \frac{\exp(2 \delta) - 1}{(\exp(\delta) - 1)^2} \right\} = \frac{1}{2}.
\end{align*}
Thus, putting everything together, it follows that
\begin{align*}
    \underset{\delta \to \infty}{\text{lim}}\left( \frac{\sigma^2_1}{\delta} \right) = \frac{1}{2} \left( \frac{\var(Y\mid X,1)}{\pi(1\mid X)}\right).
\end{align*}
Similarly, it can be shown that
\begin{align*}
     \underset{\delta \to \infty}{\text{lim}}\left\{ \frac{1}{\delta}\int_a \frac{q_\delta(a\mid X)^2}{\pi(a\mid X)}\mu^2(X,a)da \right\} &= \frac{1}{2} \left(\frac{\mu(X,1)^2}{\pi(1\mid X)} \right), \\
     \underset{\delta \to \infty}{\text{lim}}\left\{\frac{1}{\delta}\int_a \frac{q_\delta(a\mid X)^2}{\pi(a\mid X)}\mu(X,a)da \right\} &= \frac{1}{2} \left( \frac{\mu(X,1)}{\pi(1\mid X)} \right) ,\\
     \underset{\delta \to \infty}{\text{lim}}\left\{\frac{1}{\delta}\int_a \frac{q_\delta(a\mid X)^2}{\pi(a\mid X)}da \right\} &= \frac{1}{2} \left( \frac{1}{\pi(1\mid X)} \right), \\
     \underset{\delta \to \infty}{\text{lim}}\left\{\frac{1}{\delta}\int_a q_\delta(a\mid X) \mu(X, a) da \right\} &= \mu(X, 1).
\end{align*}
Therefore, if we expand out the terms in $\sigma^2_2$ it follows that
\begin{align*}
    \underset{\delta \to \infty}{\text{lim}}\left( \frac{\sigma_2^2}{\delta}\right) = \frac{1}{2}\left(\frac{\mu(X,1)^2}{\pi(1\mid X)} \right) - \frac{\mu(X,1)}{\pi(1\mid X)}\cdot \mu(X,1)  + \frac{1}{2}\left( \frac{\mu(X,1)^2}{\pi(1\mid X)}\right) = 0.
\end{align*}
Finally, we can show that as $\delta \to \infty$ then $\sigma_3^2 / \delta \to 0$ using the probability bounds on $Y$. Observe that since $| \mu(X, A)| = |\mathbb{E}[Y \mid X, A]|  \leq B$, then
\begin{align*}
    \sigma^2_3 &=  \left(\int_a q_\delta(a \mid X) \mu(X, a)da  - \psi(\delta) \right)^2 \\
    &\leq \left(\left|\int_a q_\delta(a \mid X) \mu(X, a)da \right| + \left|\int_x \int_a \mu(x, a) q_\delta(a \mid x) \, da \, d\mathbb{P}(x) \right| \right)^2 \\
    &\leq 4B^2.
\end{align*}
Putting everything together, it follows that
\begin{align*}
     \underset{\delta \to \infty}{\text{lim}}\left( \frac{\sigma^2_\delta }{\delta}\right) = \mathbb{E} \left[\frac{1}{2} \left( \frac{\var(Y \mid X,1)}{\pi(1 \mid X)} \right) \right].
\end{align*}
\end{proof}

\subsection{Proof of \cref{reflected_efficiency_bound}} \label{reflected_efficiency_bound_proof}

\begin{proof}[\textbf{Proof:}]
To derive the non-parametric efficiency bound, we must evaluate the variance of the efficient influence function of the reflected incremental effect. Recall by \cref{reflected_eif} that the efficient influence function is given by $\varphi_R(Z;\delta, a^\prime) = D_Y + D_{r, \mu} + D_\psi$ where
\begin{align*}
    D_Y &= \frac{r_\delta(A \mid X)}{\pi(A \mid X)} \Big(Y - \mu(X, A) \Big) \\
    D_{r, \mu} &= \frac{r_\delta(A \mid X)}{\pi(A \mid X)} \mu(X, A) + \Big(\mathbbm{1}(A \leq a^\prime)L(X, A) + \mathbbm{1}(A > a^\prime) U(X, A) \Big) - \mathbb{E}_R(\mu(X, A) \mid X) \\
    D_\psi &= \mathbb{E}_R(\mu(X, A) \mid X)  - \psi_R
\end{align*}
where $\mathbb{E}_R(\mu(X, A) \mid X) = \int_a \mu(X, a) r_\delta(a \mid X) \ da$ is the conditional mean of $\mu(X, A)$ under the reflected exponentially tilted distribution and
\begin{align*}
    L(X, A) &= \mathbb{E}_R(\mu(X, A) \mid X, A \leq a^\prime) - \frac{r_\delta(A \mid X)}{\pi(A \mid X)} \mathbb{E}_R(\mu(X, A) \mid X, A \leq a^\prime) \\
    U(X, A) &= \mathbb{E}_R(\mu(X, A) \mid X, A > a^\prime) - \frac{r_\delta(A \mid X)}{\pi(A \mid X)} \mathbb{E}_R(\mu(X, A) \mid X, A > a^\prime).
\end{align*}

Since $\varphi_R(Z;\delta, a^\prime)$ is mean zero, to obtain an expression for $\var(\varphi_R(Z;\delta, a^\prime))$ we simply need to evaluate $\mathbb{E} [(D_Y + D_{r, \mu} + D_{\psi} )^2]$. Thus, we first show that each of the cross-terms are equal to zero. Using the fact that $\mu(X, A) = \mathbb{E}[Y \mid X, A]$ and applying the law of iterated expectations, we can see that $\mathbb{E}\left[ D_Y D_{r, \mu}\right] = 0$ and $\mathbb{E}\left[D_Y D_\psi \right] = 0$ since for any function $f(X, A)$, it is clear that
\begin{align*}
    \mathbb{E}\left[f(X, A)(Y - \mu(X, A) \right] = \mathbb{E}\left[f(X, A) \mathbb{E}\left[(Y - \mu(X, A) \mid X, A \right] \right] = 0.
\end{align*}
Next, we consider the cross term $\mathbb{E}[D_{r,\mu} D_\psi]$. Within this cross term, observe that
\begin{align*}
    (i) &= \mathbb{E}\left[ \mathbbm{1}(A \leq a^{\prime}) L(X, A)\left(\int_a \mu(X, a) r_\delta(a \mid X) da  - \psi_R \right) \right] \\
    &= \mathbb{E}\left[ \mathbb{E}\left[ \mathbbm{1}(A \leq a^{\prime}) L(X, A)\left(\int_a \mu(X, a) r_\delta(a \mid X) da  - \psi_R \right) \mid X \right] \right] \\
    &= \mathbb{E}\Bigg[ \left(\int_a \mu(X, a) r_\delta(a \mid X) da  - \psi_R \right) \underbrace{\mathbb{E}\Big[ \mathbbm{1}(A \leq a^{\prime}) L(X, A) \mid X \Big]}_{(ii)} \Bigg] 
\end{align*}
and furthermore, using the shorthand $\mu_a = \mu(X, a)$, $\pi_a = \pi(a \mid X)$, and $r_\delta = r_\delta(a \mid X)$, note that
\begin{align*}
   (ii) &= \mathbb{E}\left[\left( \frac{\mathbbm{1}(A \leq a^{\prime})}{\int_{a \leq a^{\prime}} \pi_a da} \left(\int_{a \leq a^{\prime}}\mu_a r_\delta da - \frac{r_\delta(A \mid X) }{\pi(A \mid X)} \int_{a \leq a^{\prime}} \mu_a r_\delta da\right) \right)\mid X \right] \\
   &= \int_a \left( \frac{\mathbbm{1}(a \leq a^{\prime})}{\int_{a \leq a^{\prime}} \pi_a da} \left(\int_{a \leq a^{\prime}}\mu_a r_\delta da - \frac{r_\delta }{\pi_a} \int_{a \leq a^{\prime}} \mu_a r_\delta da\right) \right) \pi_a da \\
   &= \int_{a \leq a^{\prime}}\mu_a r_\delta da - \left( \int_{a \leq a^{\prime}} \mu_a r_\delta da\right) \frac{1}{\int_{a \leq a^{\prime}} \pi_a da}\int_{a \leq a^{\prime}}  r_\delta da \\
   &= 0
\end{align*}
since $\int_{a \leq a^{\prime}}  r_\delta da  = \int_{a \leq a^{\prime}} \pi_a da$. Similarly, it can be shown that
\begin{align*}
     \mathbb{E}\left[ \mathbbm{1}(A > a^{\prime})U(X, A)\left(\int_a \mu(X, a) r_\delta(a \mid X) da  - \psi_R \right) \right] = 0.
\end{align*}
Finally, observe that we may show the last cross term in $\mathbb{E}[D_{r,\mu} D_\psi]$,
\begin{align*}
    \mathbb{E}\left[\left(  \frac{r_\delta(A \mid X)}{\pi(A \mid X)} \mu(X, A) -  \int_a \mu(X, a) r_\delta(a \mid X) da\right)\left(\int_a \mu(X, a) r_\delta(a \mid X) da  - \psi_R \right) \right]
\end{align*}
evaluates to zero, since $ \mathbb{E}\left[\frac{r_\delta(A \mid X)}{\pi(A \mid X)}\left( \mu(X, A) - \int_a r_\delta(a \mid X) \mu(X, a)da \right) \mid X \right]$ can be written as 
\begin{align*}
   \mathbb{E}\left[ \int_a \mu(X, a) r_\delta(a \mid X) da - \int_a  \mu(X, a) r_\delta(a \mid X) da \right]  &= 0.
\end{align*}
Thus, it follows that $\var(\varphi_R(Z;\delta, a^\prime)) = \mathbb{E}[D^2_Y + D^2_{r, \mu} + D^2_\psi]$. From here, we show that $\var(\varphi_R(Z;\delta, a^\prime))$ is still governed by a linear dependence on $\delta$. To see this, we consider the dominating term $ \mathbb{E}\left[\left(r_\delta(A \mid X) / \pi(A \mid X) \right)^2 \right]$. Expanding out the expectation, it follows that $\mathbb{E}\left[\left(r_\delta(A \mid X) / \pi(A \mid X) \right)^2 \right]$ simplifies to
\begin{align*}
 \mathbb{E}\left[\omega^2_{a^\prime}\frac{\int^{a^\prime}_0 \exp(2 \delta a) \pi(a \mid X) da}{\left(\int^{a^\prime}_0 \exp(\delta a) \pi(a \mid X) da \right)^2} + (1-\omega_{a^\prime})^2\frac{\int^{1}_{a^\prime} \exp(2 \delta a) \pi(a \mid X) da}{\left(\int^{1}_{a^\prime} \exp(\delta a) \pi(a \mid X) da \right)^2} \right].
\end{align*}
Notably, both of these terms are functionally equivalent to the term
\begin{align*}
    \mathbb{E}\left[\frac{\int^1_0 \exp(2 \delta a) \pi(a \mid X) da}{\left(\int^1_0 \exp(\delta a) \pi(a \mid X) da \right)^2} \right]
\end{align*}
from \cref{nonpar_efficiency_bound}, since the limits of integration match each other in the numerator and denominator in both the lower and upper expectations. Thus, by following an identical proof to that of \cref{nonpar_efficiency_bound}, one can show that we now have the summation of two terms that are both linear in $\delta$. Therefore, it follows that $\var(\varphi_R(Z;\delta, a^\prime)) \asymp \delta$.

\end{proof}

\end{document}